\documentclass[letterpaper,12pt,reqno]{article}

\usepackage[letterpaper, left=2cm,right=2cm,top=2cm,bottom=2cm, includefoot,heightrounded]{geometry}
\usepackage{setspace}

\usepackage{mathpazo} 

\usepackage[hyphens]{url}
\usepackage[hyphenbreaks]{breakurl}
\usepackage{hyperref}
\hypersetup{colorlinks=true, linkcolor=blue, urlcolor  = blue,  citecolor=black}

\usepackage{titlesec}
\usepackage[utf8]{inputenc}
\usepackage{graphicx}
\usepackage{amsmath}
\usepackage{tcolorbox}
\usepackage{MnSymbol}%
\usepackage{hieroglf}
\usepackage{wasysym}%
\graphicspath{{Images/}}
\usepackage{amsfonts}
\usepackage{commath}
\usepackage{algpseudocode}
\usepackage{enumitem} 
\usepackage{xcolor}
\usepackage[font=small, width=0.8\textwidth]{caption}
\usepackage[english]{babel}
\usepackage{amsthm}
\usepackage{tocloft}
\usepackage{bbm}
\usepackage{mathtools}
\usepackage{tikz}
\usepackage{comment}
\usetikzlibrary{decorations.pathreplacing}
\usetikzlibrary{positioning,quotes,calc}
\usetikzlibrary{intersections}
\usepackage{subcaption}
\definecolor{niceblue}{RGB}{0,104,178}
\definecolor{darkgreen}{RGB}{3, 131, 127}
\definecolor{orange}{RGB}{213,94,0}
\definecolor{darkorange}{RGB}{191,84,0}
\definecolor{paleniceblue}{RGB}{118,171,208}
\definecolor{camniceblue}{RGB}{145, 182, 177}
\definecolor{grey}{RGB}{60, 62, 61}
\definecolor{purple}{RGB}{153, 51, 255}
\definecolor{red}{RGB}{255, 51, 51}
\definecolor{darkdarkgreen}{RGB}{10, 105, 102}
\definecolor{darkred}{RGB}{148, 47, 47}

\definecolor{niceblue}{RGB}{0,104,178}
\definecolor{darkgrey}{RGB}{3, 131, 127}
\definecolor{orange}{RGB}{213,94,0}
\definecolor{darkorange}{RGB}{191,84,0}
\definecolor{paleniceblue}{RGB}{118,171,208}
\definecolor{camniceblue}{RGB}{145, 182, 177}
\definecolor{grey}{RGB}{60, 62, 61}
\definecolor{purple}{RGB}{153, 51, 255}
\definecolor{red}{RGB}{255, 51, 51}
\definecolor{darkdarkgrey}{RGB}{10, 105, 102}
\definecolor{darkred}{RGB}{148, 47, 47}

\definecolor{teal}{RGB}{0, 150, 136}         
\definecolor{peach}{RGB}{255, 160, 122}      
\definecolor{lavender}{RGB}{178, 102, 255}   
\definecolor{midnightniceblue}{RGB}{25, 25, 112} 
\definecolor{olive}{RGB}{128, 128, 0}        
\definecolor{mustard}{RGB}{204, 153, 0}      
\definecolor{rose}{RGB}{255, 102, 102}       
\definecolor{plum}{RGB}{102, 0, 102}         

\newtheorem*{theorem*}{Theorem}
\newtheorem*{corollary*}{Corollary}
\newtheorem*{proposition*}{Proposition}
\newtheorem*{lemma*}{Lemma}
\newtheorem*{fact*}{Fact}
\newtheorem*{definition*}{Definition}
\newtheorem*{conjecture*}{Conjecture}
\newtheorem{theorem}{Theorem}
\newtheorem{example}{Example}

\newtheorem{corollary}{Corollary}
\newtheorem{proposition}{Proposition}

\newtheorem{assumption}{Assumption}
\newtheorem{definition}{Definition}
\newtheorem{lemma}{Lemma}
\newtheorem{remark}{Remark}

\DeclarePairedDelimiterX{\inp}[2]{\langle}{\rangle}{#1, #2}

\newcommand{\ool}[1]{\overline{#1}}
\newcommand{\uul}[1]{\underline{#1}}
\DeclareMathOperator*{\argmax}{arg\,max}

\titleformat{\section}
		{\large\bfseries\center}
         {\thesection}
        {0.5em}
        {}
        []
\titleformat{\subsection}
        {\normalfont\bfseries}
        {\thesubsection}
        {0.5em}
        {}
\titleformat{\subsubsection}[runin]
        {\normalfont\bfseries}
        {\thesubsubsection}
        {0.5em}
        {}
        [.]

\usepackage[authoryear]{natbib}
\usepackage{soul}
\title{Screening with Tolls and Damages\thanks{I am grateful to Rafael Berriel, Benjamin Brooks, Rebecca Diamond, Laura Doval, Piotr Dworczak, Joey Feffer, Zi Yang Kang, Jacob Leshno, Federico Llarena, Michael Ostrovsky, Ilya Segal, Andrzej Skrzypacz, and Sam Wycherley for their helpful comments and suggestions.}}
\author{Filip Tokarski \\ Stanford GSB}

\begin{document}
\date{} 
\maketitle
\vspace*{-1cm} 

\begin{abstract}
A welfare-maximizing designer allocates two kinds of goods using two screening instruments: \emph{tolls}, whose costs are separable from agents' values, and \emph{damages}, which are more costly to agents whose values for the goods are higher. Tolls include payments, queues, and administrative burdens; damages include quality reductions, delays, and restrictions on use. When agents differ only in their value for one kind of good, the designer can never gain from damaging it. However, when valuations for both kinds of goods are heterogeneous, damages can be useful. I provide conditions under which the optimal mechanism includes a damaged option, as well as conditions under which it does not; in the latter case, the optimal mechanism posts ``market-clearing'' tolls for each good. Intuitively, damages are more likely to be optimal when values for the two kinds of goods are positively affiliated, and less likely when high value for one good predicts low value for the other.
\end{abstract}

\section{Introduction}

Compare two mechanisms for scheduling appointments, such as clinic visits, passport renewals, or public-benefit interviews. The first mechanism offers walk-ins allocated through a \emph{queue}: the first \(n\) people to arrive after the office opens are served. In practice, this means that an applicant may need to arrive hours ahead of time to secure a place. The second mechanism is a \emph{waitlist}: people sign up for an appointment in advance and are given the first available slot. This, in turn, entails waiting weeks for one's scheduled time. Both mechanisms impose a burden on participants: in the queue, it is the time spent waiting before the office opens; in the waitlist, it is the delay before service. However, while both burdens may discourage some applicants from using the system, they differ crucially in how they interact with people's values for the appointment: the disutility of arriving ahead of time is largely separable from one's need for the service, whereas delay is more costly for those whose need is higher. Indeed, a patient with an untreated injury suffers more from a month-long wait than someone seeking a routine check-up; a delayed caseworker appointment is more costly for a family that relies exclusively on benefits than for one with other sources of support.

Waitlists and queues exemplify the two kinds of screening instruments studied in this paper. The former are more costly to agents whose value for the allocated good is higher; I call these instruments \emph{damages}. By contrast, the costs imposed by the latter class, which I call \emph{tolls}, are separable from agents' values for the allocated good. Both categories should be understood broadly: damages include literal reductions in quality, but also delays in allocation, restrictions on use, or other policies that reduce the value of the good itself. Tolls include monetary payments, queues, bureaucratic processes, and other ordeals that burden participants without directly reducing the value of the good they receive. The distinction between tolls and damages appears in a variety of economic settings; I provide several examples below.

\vspace{-0.2cm}
\paragraph{Regulating essential services.}
Regulators often require providers of essential services, such as telecommunications or banking, to offer basic, affordable versions of their products. Such policies can combine price regulation with restrictions on the quality or scope of service. For example, in the United Kingdom, major banks are required to offer fee-free basic bank accounts, which provide core services but cannot include an overdraft facility.\footnote{\url{https://www.fca.org.uk/publications/multi-firm-reviews/retail-banking-our-review-basic-bank-accounts}.} Similarly, Portugal requires telecommunications providers to offer a basic service with low speeds and a 15 GB monthly data allowance at a fixed monthly price.\footnote{\url{https://digital.gov.pt/en/estrategia-digital/plano-de-acao-para-a-transicao-digital/01-capacitacao-e-inclusao-digital-das-pessoas/tarifa-social-de-acesso-a-internet}.} In these settings, the regulated fee is a toll, whereas the exclusion of overdraft access, speed restrictions, and data caps are damages.
\vspace{-0.2cm}

\paragraph{Referrals vs. copayments in healthcare.}
Many conditions admit both low- and high-intensity treatments: for instance, musculoskeletal injuries can often be treated with physical therapy or surgery. Because capacity is limited, health systems often ration access to more intensive treatment. This can be done by charging copayments, which constitute tolls. Healthcare providers can also require referrals, prior authorization, or documentation that less intensive treatments have been tried first. Indeed, \citet{world2023high} notes that such measures are intended to ``ensure patient access to specialist healthcare when needed, while maintaining resource efficiency.'' However, they also delay access to the treatment itself, and are therefore especially costly for patients whose injuries are more severe. In this sense, referral and authorization requirements act as damages.

\vspace{-0.2cm}
\paragraph{Self-targeting in social programs.}
Goods provided through public programs are often less desirable than their private-market counterparts. Social housing is frequently located in disadvantaged or peripheral neighborhoods, Medicaid covers fewer providers than private insurance, and food assistance delivered through in-kind or voucher programs such as SNAP restricts choice.\footnote{Some quality reductions may of course reflect lower provision costs. I interpret them as damages only to the extent that quality is reduced below the efficient level, so that the reduction serves as a screening device.} These restrictions are costly to recipients, but can also serve as self-targeting devices \citep{nichols,Besley,currieGahvari}. For instance, if subsidized housing is built in less desirable neighborhoods, families with better alternatives may decide not to apply, so the subsidy disproportionately benefits poorer households. However, quality reductions and usage restrictions are also more costly for households that value the program more: a sicker patient is more burdened by a narrow provider network, and a family relying more heavily on food vouchers is more constrained by restrictions on eligible products. Thus, in my framework, these instruments constitute damages. Nevertheless, self-targeting can also be achieved without diminishing the value of the allocated good: for example, social programs often require applicants to wait in line, complete paperwork, travel to a distant office, or recertify eligibility \citep{nichols1971discrimination,nichols,besley1992workfare,kleven2011transfer,dupas2016targeting,alatas2016self,deshpande2019screened}. Because such ordeals do not diminish the value agents derive from the allocated good, they constitute tolls.

Other environments in which the distinction between tolls and damages is relevant include appointment scheduling, discussed above, congestion pricing, discussed in Section \ref{sec:1d}, and, more generally, settings where goods can be allocated either through passive waitlists or through queues.\footnote{For flow goods that can be enjoyed in every period after receipt, long waitlists cause delays that deprive recipients of periods of use. For consumable goods, they lead to temporal discounting of the good's value.}
I microfound this interpretation in Appendix \ref{appsec:waitlists}, where I show that delays arising endogenously in a steady-state waitlist can act as damages. 

Motivated by these examples, I consider a model in which a welfarist designer allocates scarce goods using deterministic mechanisms that may combine tolls and damages. I study settings involving the rationing of a single kind of good, as well as settings in which agents choose among horizontally differentiated options. Both arise in the examples above. Unemployment insurance programs, for
instance, ration access to a specific benefit, even though application hurdles or delays may reduce its value. By contrast, affordable
housing programs offer units that vary in location and size, with households'
preferences over them having a strong horizontal component
\citep{waldinger2021targeting}. Similarly, appointment systems offer slots
at different offices, and hospitals offer treatments that differ
in intensity. This distinction matters because screening instruments play qualitatively
different roles in the two cases. With one kind of scarce good, they serve only to
exclude low-value agents from the system. With heterogeneous options,
they also affect how participants sort across the available goods. Indeed, in the public-housing context, waitlist lengths often differ
substantially across developments and households decide where to apply
not only based on how much they value a particular unit, but also on how long
they expect to wait for it \citep{van2016households}. Similar sorting motives arise in healthcare, where copays and referral requirements affect whether
patients pursue low- or high-intensity care, and in appointment systems, where
agents choose among clinics or offices that differ in proximity and congestion.

I first study the former case with one scarce good
and a common-value outside option. There, Proposition \ref{prop:1Dnodamage} shows that damages are suboptimal in a strong sense: any feasible mechanism that uses damages is Pareto-dominated by one that uses only tolls. Intuitively, damages are
inefficient because they burden high-value recipients more than the marginal
agents they are meant to screen out. Tolls, meanwhile, can implement the same allocation
while leaving higher rents to inframarginal types. When agents have heterogeneous values for different kinds of goods, however, damaging some of them can be optimal. I study when this occurs in a setting with two kinds of goods. Theorem \ref{th:1} gives sufficient conditions under which damages are not useful. The optimal mechanism is then the \emph{market-clearing toll mechanism}: the designer offers undamaged goods at tolls chosen so that both supplies are exhausted when agents choose their favorite option. Theorem \ref{thm:2} gives conditions under which the optimal mechanism instead uses damages. Intuitively, whether damaging some goods is optimal depends on the statistical dependence between values in the population. When high value for one good predicts low value for the other, tolls alone are likely to be optimal. When values are instead sufficiently positively affiliated, damages can improve welfare.

I later extend the analysis by allowing agents to differ in how costly they find the toll. For instance, when screening is done with monetary payments, poorer agents whom the program tries to target may find them more burdensome. Conversely, the same agents may be more willing to wait, travel, or endure other inconveniences in order to get the good \citep{dupas2016targeting}. I show that this extension leaves the main logic of the results unchanged: after appropriately reweighting the type distribution to account for heterogeneous ordeal costs, the same forces determine when damages can improve welfare.

My results have two main implications for market design. First, when the allocated good is homogeneous, screening should rely on tolls rather than damages: delays, usage restrictions, and quality reductions should, where possible, be replaced by instruments such as fees, application requirements, or queues. Second, when goods are heterogeneous, the appropriate instrument depends on how agents' values for them are related. If high value for one good predicts low value for another, tolls alone are likely to be optimal. This may happen when options are geographically dispersed and agents prefer those closer to where they live. In appointment systems, for instance, this means that offices in the most densely populated locations should charge higher booking fees or use walk-in queues rather than allowing long waitlists to build up. If values are instead strongly positively related, as may be the case in public housing where households' valuations are driven primarily by overall need, differentiated waitlists may be optimal, possibly in combination with toll instruments such as differential rent subsidies.

From a technical perspective, my model is an instance of a tractable multidimensional screening problem. By restricting attention to deterministic mechanisms, I am able to characterize them as pairs of toll and damage menus for the two goods. This in turn allows me to represent two-dimensional mechanisms as interconnected single-dimensional screening problems. The interaction between them is summarized by a boundary in the type space that separates the sets of types who choose each good. The multidimensional problem can then be broken up into two stages: first, determining the optimal way to implement a given boundary, and second, solving an optimal control problem to select the best boundary among all implementable ones. While my paper applies this method to the problem of a welfarist designer, similar ideas could be useful for studying other settings, such as the problem of a two-good monopolist choosing deterministic mechanisms for selling to unit-demand consumers.

The rest of the paper is organized as follows. The next section discusses the related literature and Section \ref{sec:model} presents the model. Section \ref{sec:1d} studies the case where agents differ only in their value for one type of good. Section \ref{sec:2d} extends the analysis to two-dimensional heterogeneity, introduces the boundary representation of mechanisms, and gives conditions under which damages are and are not optimal. Section \ref{sec:proofth1} presents the key steps in the proof of Theorem \ref{th:1}. Section \ref{sec:heterotollcost} extends the analysis to heterogeneous toll costs. Finally, Section \ref{sec:disc} discusses the implications for market design.

\section{Related literature}

This paper contributes to the literature on using costly screening devices and money-burning to maximize welfare. \citet{hartline2008optimal} and \citet{CONDORELLI2012613} show that when goods are allocated without monetary transfers, requiring agents to undertake socially wasteful actions can sometimes improve welfare. \citet{dworczak2026allocate} asks when a redistributive designer would like to hand out money in exchange for completing a costly ordeal. This literature, however, focuses primarily on the allocation of homogeneous or vertically differentiated goods. An exception is \citet{noda2024no}, who study a symmetric environment with heterogeneous good varieties and show that screening through money burning becomes less efficient as the number of varieties grows; nevertheless, they restrict attention to mechanisms that offer agents only their favorite variety. By allowing for richer heterogeneity, I explore an additional role for wasteful screening devices: rather than only affecting the participation margin by determining who is excluded, they can also direct agents toward different goods in a socially efficient manner. When one variety is significantly overdemanded, for instance, the designer can damage it or increase its toll to redirect agents with weaker preferences to less scarce alternatives. My work also relates to \citet{akbarpour2023comparison}, who compare various wasteful screening devices: they ask when one screening instrument dominates another for a planner seeking to maximize social welfare. Unlike them, I allow the designer to \emph{combine} instruments and study when using only one screening device dominates any combination of tolls and damages.

A related literature studies wasteful screening by profit-maximizing sellers. \citet{deneckere1996damaged} give conditions under which a monopolist may want to damage goods in order to price discriminate. \citet{yang2021costly} studies a more general problem in which the monopolist has access to both wasteful and non-wasteful instruments and describes cases in which the wasteful instrument should not be used. My analysis instead concerns how a welfare-maximizing regulator would want sellers to use such instruments. In this sense, it connects to \citet{bulow2012regulated}, who study price controls from the perspective of consumer surplus: in my framework, prices are tolls, and their objective corresponds to the case in which the designer places no value on toll revenue. They observe that price controls can be useful when buyers' values are clustered near the market-clearing price, leaving little surplus net of payments. I extend this logic to heterogeneous goods: when toll-based sorting leaves many
recipients nearly indifferent between the available goods because their value
differences are close to the market-clearing toll differential, damages can
improve welfare by changing the sorting pattern while preserving more surplus
for inframarginal recipients.

The effects of the screening instruments I classify as tolls and damages have also been documented empirically in several settings to which my paper could be applied. In healthcare, \citet{manning1987health} show that higher patient cost-sharing---a monetary toll---substantially reduces medical care use, in part by discouraging low-value care. More closely related to the sorting forces in my model, \citet{brot2023rationing} show that prior authorization restrictions in Medicare Part D substantially reduce the use of restricted drugs and lead many marginal patients to switch to cheaper options. Other work considers non-monetary ordeals in social programs. For example, \citet{finkelstein2019take} find that administrative burdens can improve targeting in SNAP. A related literature also studies self-targeting through offering inferior goods and restrictions on use \citep{nichols,Besley,currieGahvari}. Since these policies reduce the value of the allocated good, I consider them damages.

Finally, my paper relates to the literature on waitlist design. While no paper has studied combining waitlists with payments or ordeals in settings with heterogeneous goods, a substantial literature examines the design of waitlists without transfers. \citet{arnosti2020design} and \citet{waldinger2021targeting} study the effects of restricting recipients' choice on targeting. \citet{barzel1974theory}, \citet{blochCantala}, and \citet{leshno2022dynamic} observe that in environments with homogeneous waiting costs, wait times may ``act as prices,'' screening for agents with higher valuations. I refine this intuition by showing that the screening properties of wait times are different when the cost of waiting stems from delayed receipt.

\section{Model}\label{sec:model}

A designer distributes two types of goods, $A$ and $B$; their supplies are $s_A,s_B> 0$. There is a unit mass of agents whose values for the two goods, $(a,b)$, are distributed according to a nonatomic distribution $F$ on $[0,1]^2$. The designer chooses a menu of qualities and tolls for each of the goods. That is, an agent can choose which good she wants to get and then choose a quality and toll option from the relevant good's menu. She can also choose not to participate, which gives her utility $0$. When a type-$(a,b)$ agent participates and receives good $y$, her utility is:
\begin{align*}
x\cdot a - c & \ \  \text{if}\ \  y = A, \\
x\cdot b - c & \ \  \text{if}\ \  y = B,
\end{align*}
where $c \in \mathbb{R}_+$ is the toll the agent incurs and $x \in [0,1]$ is her good's quality. Whenever $x<1$, we say the good has been \emph{damaged}. The designer chooses the menu to maximize welfare, counting a fraction \(\gamma\in[0,1]\) of each toll as socially valuable. We can then reduce her problem to picking allocation rules for tolls, $c:[0,1]^2\to \mathbb{R}_+$, qualities, $x: [0,1]^2\to [0,1]$, and goods, $y:[0,1]^2 \to \{\varnothing, A,B\}$ to maximize:
\begin{equation}\label{eq:obj}
\int
\left[
u_{a,b}(a,b)
+
\gamma c(a,b)
\right]
\,dF(a,b),
\tag{O}
\end{equation}
subject to incentive-compatibility, participation, and supply constraints:
\begin{equation}\label{eq:IC}
        \text{for all } (a,b), (a',b') \in [0,1]^2, \quad u_{a,b}(a,b) \geq u_{a,b}(a',b'),
        \tag{IC}
        \end{equation}
        \begin{equation}\label{eq:IR}
        \text{for all } (a,b) \in [0,1]^2, \quad u_{a,b}(a,b) \geq 0,
        \tag{IR}
        \end{equation}
        \begin{equation}\label{eq:S}
        \int \mathbb{1}_{y(a,b)=A} \ d  F(a,b) \leq s_A, \quad \int \mathbb{1}_{y(a,b)=B} \ d  F(a,b) \leq s_B.
        \tag{S}
        \end{equation}
Here $u_{a,b}(a',b')$ denotes the utility type $(a,b)$ gets from reporting $(a',b')$ in the mechanism $(c,x,y)$. I call a mechanism $(c,x,y)$ satisfying \eqref{eq:IC},\eqref{eq:IR}, \eqref{eq:S} \emph{feasible}.

\subsection{Discussion of the model}

\paragraph{Explicit design choices and equilibrium objects.}
In the formulation above, tolls and damages are chosen by the designer.
Nevertheless, the model covers both cases in which these instruments are imposed
directly, as with fees or explicit quality restrictions, and cases in which the
designer chooses rules that make them arise endogenously, as with congestion or
waitlists. This is because, even when tolls or damages are equilibrium objects,
the equilibrium conditions can be written in the form of incentive and supply
constraints: IC constraints ensure that agents optimize given equilibrium
burdens, and supply constraints ensure that these burdens adjust so that demand
for each option matches the allocated supply. Appendix \ref{appsec:waitlists}
makes this reduction explicit for waitlists.
\vspace{-0.2cm}

\paragraph{Restriction to deterministic mechanisms.}
My model does not allow the designer to offer lotteries. This reflects practical
constraints and considerations: for instance, in many settings of interest, lotteries can raise
trust concerns as agents may not believe that the random draw
is implemented honestly. Indeed, sellers rarely use randomized allocations, presumably
because customers may fear that the lottery can be manipulated. Similar worries
arise even in the context of public programs: allocating affordable housing
randomly sometimes raises concerns about corruption or draw-faking, prompting
calls to replace lotteries with more transparent mechanisms such as
first-come-first-served waitlists.\footnote{See, for instance,
\url{https://www.camara.leg.br/noticias/523091-projeto-veda-sorteio-na-selecao-dos-beneficiarios-do-minha-casa-minha-vida/}
and
\url{https://citymeetings.nyc/meetings/new-york-city-council/2025-04-29-1000-am-committee-on-housing-and-buildings/chapter/consideration-of-moving-from-lottery-system-to-universal-waiting-list-for-affordable-housing}.}
Relatedly, even honestly executed lotteries for public housing are sometimes
perceived as unfair; for example, Whistler, Canada, ``allocates units based on a
waitlist, a method chosen due to its perceived fairness and ease of
administration, though lottery and points schemes have been used in the past''
\citep{vancouver2016aho}. Lotteries also create gaming incentives. For example,
Beijing used license-plate numbers to effectively randomize which days a car
could be driven: access to the road depended on the terminal digit of the
car's plate. This led households to circumvent the restriction by swapping or
borrowing cars with different plate numbers, or by acquiring additional cars
whose plates allowed driving on otherwise restricted days \citep{wang2014will}. In other contexts, lotteries may encourage unsuccessful applicants to re-enter illegally, for instance by applying
through proxies or other family members.

Finally, the restriction is useful analytically. It isolates the comparison between tolls and damages, which is the focus of the paper (for a model of welfare-maximizing screening with lotteries, see \citet{tokarski2026targeting}). It also gives the model enough structure to make a multidimensional screening problem tractable: together with the assumption of unit demand, it lets me represent mechanisms using boundaries in the type space, which mitigates some of the difficulties posed by multidimensional screening, as studied, among others, by \citet{rochetChone,manelliVincent} and \citet{daskalakisDeckelbaumTzamos}.
\vspace{-0.2cm}

\paragraph{Value of tolls.} To capture different interpretations of tolls, I use the parameter $\gamma$ to represent their social value. Consider first the case where tolls are monetary payments. If the designer can costlessly rebate revenue to participants, it is natural to set \(\gamma=1\). Intermediate cases with \(\gamma\in(0,1)\) may describe government programs where rebating revenue to participants is possible but administratively costly, or where distributing cash undermines the screening benefits of in-kind transfers.\footnote{
When the designer provides a free or subsidized inferior good, only relatively poor agents will want to participate, since wealthier agents can afford higher-quality alternatives. Thus, the subsidy is automatically targeted to those who need it most \citep{Besley}. Once the designer provides cash, this form of targeting disappears, since money is desired by everyone regardless of wealth.} Setting \(\gamma=0\) may be appropriate, for example, when the designer regulates a seller and cares only about consumer surplus. Beyond monetary tolls, the case of \(\gamma=0\) can also be interpreted as ``money burning'' à la \cite{hartline2008optimal}, capturing ordeals like queueing.
\vspace{-0.2cm}

\paragraph{The revelation principle.} Although the designer's problem is written as a choice over direct mechanisms, the usual revelation-principle argument is not immediate because I restrict attention to deterministic options. Appendix \ref{sec:dirrev} shows that the direct-revelation formulation used above is nevertheless without loss.

\section{One-dimensional heterogeneity}\label{sec:1d}

I first consider the case where agents differ only with respect to their value for good $A$; good $B$ plays the role of a common-value outside option. The following result shows that using damages in this setting is suboptimal in a strong sense.

\begin{proposition}\label{prop:1Dnodamage}
Suppose \(s_A<1\), $s_B\geq1$, and all agents have the same value \(b>0\) for good \(B\). Then the optimal mechanism does not use damages. Moreover, any feasible mechanism is weakly Pareto-dominated by a mechanism that uses only tolls, i.e., one with \(x(a,b)\equiv 1\).
\end{proposition}

The proof is in the appendix. To understand the intuition behind this result, note that under any mechanism, good $A$ will go to agents whose values for it lie in some upper interval $[\uul a, 1]$. The designer's problem therefore boils down to selecting a cutoff $\uul a$ and choosing how to enforce it. This requires deterring some agents from choosing good $A$, and can be done by damaging it or by pairing it with a toll.  Note, however, that damages are more burdensome to \emph{inframarginal types} than the types below $\uul a$ they are actually meant to deter. Tolls, on the other hand, are equally burdensome to everyone. Thus, enforcing the cutoff $\uul a$ with tolls yields higher welfare, as it leaves more rents to inframarginal takers of $A$ (Figure \ref{fig:intuition1}).

\begin{figure}[h!]
        \centering
        \begin{subfigure}[b]{0.45\linewidth}
            \centering
            \begin{tikzpicture}[scale=0.65]
                \draw[->] (0,0) -- (6.3,0) node[right] {$a$};
                \draw[->] (0,0) -- (0,5.3) node[above] {};
            
                \draw[ultra thick, darkgreen] 
                    (0,1.3) -- (2.2,1.3) 
                    -- (6.1,1.3+2.2) node[pos=0.5, left, xshift=-5pt] {$U(a,b)$};
            
                \node at (2.2, -0.6) {$\uul{a}$};
                \node at (6.1, -0.5) {$1$};
                \node at (0, -0.5) {$0$};
                \node at (-0.35, 1.3) {$\uul U$};
            
                \draw[decorate, decoration={brace, amplitude=5pt}, ultra thick, darkorange] (2.2,0) -- (6.1,0)
                    node[midway, above=5pt, font=\footnotesize\color{darkorange}] {Receive $A$};
            \end{tikzpicture}   
        \end{subfigure}
        \hspace{0.2cm}
        \begin{subfigure}[b]{0.45\linewidth}
            \centering
            \begin{tikzpicture}[scale=0.65]
                \draw[->] (0,0) -- (6.3,0) node[right] {$a$};
                \draw[->] (0,0) -- (0,5.3) node[above] {};
            
                \draw[ultra thick, darkgreen] 
                    (0,1.3) -- (2.2,1.3) 
                    -- (6.1,1.3+3.9) node[pos=0.5, left] {$U(a,b)$};
            
                \node at (2.2, -0.6) {$\uul{a}$};
                \node at (6.1, -0.5) {$1$};
                \node at (0, -0.5) {$0$};
                \node at (-0.35, 1.3) {$\uul U$};
            
                \draw[decorate, decoration={brace, amplitude=5pt}, ultra thick, darkorange] (2.2,0) -- (6.1,0)
                    node[midway, above=5pt, font=\footnotesize\color{darkorange}] {Receive $A$};
            \end{tikzpicture}         
        \end{subfigure}
        \caption{Indirect utilities $U(a,b) = \uul U +\int_{\uul a}^{\max\{\uul a,a\}} x(v,b) \ dv$ for mechanisms enforcing the cutoff $\uul a$ with damages (left) and tolls (right).}
        \label{fig:intuition1}
    \end{figure}

Replacing damages with tolls is also beneficial for the designer, who does not benefit from damages but values agents' utility and toll revenue.

\begin{remark}
The logic of Proposition \ref{prop:1Dnodamage} extends to a more general class of screening instruments. We could consider two wasteful screening instruments, where the cost of one increases more steeply with the value for good $A$ than does the cost of the other. An analogous result would then say that any mechanism using the former instrument is Pareto-dominated by one using only the latter instrument.
\end{remark}

While the result is simple, it provides economic insight. For example, it captures a key force in the congestion-pricing model of \citet{vickrey1973pricing} and its subsequent generalization by
\citet{van2011winning}. These models study the use of congestion pricing to spread traffic over time
on a capacity-constrained road. They show that when drivers' values for time
are identical, congestion pricing does not improve their utility. This is because the same allocation of driving times is implemented with or without tolls: the ``prices'' for traveling at specific times are then pinned down by market clearing, and it is irrelevant whether they are paid in money, through tolls, or in wasted time, through congestion. This conclusion is overturned, however, when agents' values for time are heterogeneous. Proposition \ref{prop:1Dnodamage} makes this clear, as waiting in traffic is a \emph{damage}. When the road's capacity constraint at peak time is enforced through payments, everyone pays the same price as the marginal driver. When it is enforced through waiting in traffic, the marginal driver experiences the same disutility as she would from the payment, but the inframarginal drivers with the highest values for arriving early suffer strictly more.

\section{Two-dimensional heterogeneity}\label{sec:2d}

I now turn to the case where both goods are scarce and values are heterogeneous in both dimensions. I therefore assume that $s_A,s_B>0$, $s_A+s_B\leq 1$ and that the distribution of values $F$ has full support on $[0,1]^2$. This two-dimensional heterogeneity makes the setting significantly richer. Indeed, tolls and damages will serve not only to exclude low-value agents, but also to direct recipients to choose goods in a socially efficient way. Before I discuss how screening devices ``sort'' agents into goods, however, I show that incentive-compatible mechanisms can be conveniently represented using a boundary in the type space. 

\paragraph{Boundary structure of mechanisms.} Let us first define good-specific indirect utility functions $U_A,U_B:[0,1]\to \mathbb{R}_+$, given by:
\[
U_A(a)
=
\max\Big\{
0,
\sup_{\substack{(a',b')\\ y(a',b')=A}}
\bigl(x(a',b')a-c(a',b')\bigr)
\Big\},
\ \ 
U_B(b)
=
\max\Big\{
0,
\sup_{\substack{(a',b')\\ y(a',b')=B}}
\bigl(x(a',b')b-c(a',b')\bigr)
\Big\}.
\]
Intuitively, $U_A(a)$ and $U_B(b)$ represent the highest utility that type $(a,b)$ could get from selecting some quality and toll option for goods \(A\) and \(B\), respectively, or not participating. $U_A$ and $U_B$ are convex, increasing, and depend only on one dimension of the type---an agent's value for good $B$ does not affect her choice of quality and toll option if she chooses good $A$. Notice also that agents for whom $U_A(a)>U_B(b)$ choose good $A$ and those for whom $U_A(a)<U_B(b)$ choose good $B$. These observations let us define a boundary characterizing different types' good choices.

\begin{definition}
Define a mechanism's  \textbf{participation cutoffs} as follows:
\begin{equation*}
        \uul a = \sup \{a:  \ U_A(a)=0 \}, \quad \uul b = \sup \{b:  \ U_B(b)=0 \}.
\end{equation*}
Let a \textbf{boundary} be a strictly increasing, continuous function
\(z:[\uul a,\ool a]\to[\uul b,\ool b]\) such that
\(z(\uul a)=\uul b\), \(z(\ool a)=\ool b\), \(\ool a\leq 1\),
\(\ool b\leq 1\), and at least one of the last two inequalities holds with
equality.
\end{definition} 

\begin{proposition}\label{prop:boundary}
Suppose the mechanism allocates positive masses of both goods. Then agents' choices of goods are characterized by the mechanism's participation cutoffs $\uul a,\uul b$ and a boundary $z$:
\begin{enumerate}
        \item Types $(a,b)<(\uul a,\uul b)$ do not get either good.
        \item Types for whom $a>\uul a$ and $b<\uul b$ get good $A$; types for whom $a<\uul a$ and $b>\uul b$ get good $B$. 
        \item Types \((a,b)>(\uul a,\uul b)\) with \(a\leq\ool a\) get good \(A\) if they are below the boundary \(z\), that is, if \(z(a)>b\), and good \(B\) if they are above it, that is, if \(z(a)<b\). If \(\ool a<1\), types with \(a>\ool a\) get good \(A\).
\end{enumerate}
Moreover, types on the boundary are indifferent between their preferred options for the two goods, so
\begin{equation}\label{eq:boundaryindifference}
        U_A(a) = U_B(z(a)) \quad \text{for all} \quad a \in [\uul a,\ool a].  
\end{equation}
\end{proposition}

\begin{figure}[h!]
        \centering
        \begin{tikzpicture}[scale=0.75]
                \draw[->] (0,0) -- (5.1,0) node[right] {$a$};
                \draw[->] (0,0) -- (0,5.1) node[above] {$b$};
        
                \draw[darkgreen, very thick] (1.2,1.6) .. controls (1.8,2) .. (2.5,3)
                                        .. controls (3.1,3.8) and (3.4,4.5) .. (4,5);
                
                \draw[darkgreen, dashed] (1.2,0) -- (1.2,1.6);  
                \draw[darkgreen, dashed] (0,1.6) -- (1.2,1.6);  
        
                \fill[niceblue, opacity=0.3] (0,5) -- (0,1.6) -- (1.2,1.6)
                                        .. controls (1.8,2) .. (2.5,3)
                                        .. controls (3.1,3.8) and (3.4,4.5) .. (4,5) -- cycle;
        
                \fill[orange, opacity=0.3] (5,0) -- (1.2,0) -- (1.2,1.6)
                .. controls (1.8,2) .. (2.5,3)
                .. controls (3.1,3.8) and (3.4,4.5) .. (4,5) --(5,5) -- cycle;
        
                \node[darkgreen, anchor=south west, rotate=55] at (1.8,1.9) {\fontsize{14}{14}\selectfont $z(a)$};
        
                \node[darkgreen, anchor=south] at (0.9,3) {\fontsize{20}{20}\selectfont $B$};
                \node[darkgreen, anchor=south] at (3.4,1) {\fontsize{20}{20}\selectfont $A$};
                \node[darkgreen, anchor=south] at (0.6,0.4) {\fontsize{14}{14}\selectfont $\varnothing$};
                \node at (1.2, -0.5) {$\uul{a}$}; 
                \node at (-0.35, 1.6) {$\uul{b}$}; 
            \end{tikzpicture}
                    \captionsetup{width=0.5\textwidth}
            \caption{Types below the boundary (orange) choose good $A$ and types above it (blue) choose good $B$.}
            \label{fig:bdry}
\end{figure}

When all offered \(A\)- and \(B\)-options come with tolls, types with sufficiently
low values for both goods, i.e. \((a,b)<(\uul a,\uul b)\), do not participate. The choices of participating types are then described by an increasing boundary $z$
along which the good-specific indirect utilities are equal. Types below this
boundary choose good \(A\), while types above it choose good \(B\). Indeed, if a type \((a,b)\) lies below the boundary, then
\(b<z(a)\). Since \(U_A(a)=U_B(z(a))\), we have
\[
        U_A(a)=U_B(z(a))>U_B(b),
\]
so the type prefers good \(A\). Analogously, if \((a,b)\) lies above
the boundary, then \(b>z(a)\), and hence \(U_B(b)>U_A(a)\).

I now discuss how different combinations of tolls and damages induce different boundaries $z$. This is in contrast with the one-dimensional setting studied in Section \ref{sec:1d}, where every pattern of sorting agents into goods could be achieved using either instrument.

\paragraph{Tolls and damages sort agents differently.} For illustration, suppose that supplies add up to one, \(s_A+s_B=1\), and consider mechanisms that use only tolls and only damages. In the former case, we have $x(a,b)\equiv 1$, and so \eqref{eq:IC} requires that each good be given with a single toll, $c_A$ or $c_B$. Note that type-$(a,b)$ agents will select good $A$ if
\[
a-c_A > b-c_B, 
\]
and select good $B$ otherwise. This kind of screening can only lead to sorting patterns like the one illustrated in Figure \ref{fg:nomoney}, where agents get good $A$ if their types lie below a certain $45$-degree line. It cannot create a sorting pattern like that in Figure \ref{fg:costs}, where agents get good $A$ if their types lie below a ray from the origin, that is, if
\[
\frac{b}{a}< q,
\]
for some $q$, and get good $B$ otherwise. Such a sorting pattern can be achieved with damages, however. Consider a mechanism offering $A$ with a damage, $x=q<1$, and $B$ without it, $x=1$, with no tolls for either. The set of indifferent agents will then be given by:
\[
a\cdot q = b \quad \Rightarrow \quad \frac{b}{a}= q.
\]
Agents below and above this boundary will then choose goods $A$ and $B$, respectively. 

\begin{figure}[h!]
        \centering
        \begin{subfigure}[b]{0.4\linewidth}
            \centering
                    \begin{tikzpicture}[scale=0.7]

            \draw[darkdarkgreen, very thick] (1.2,0) -- (5,3.8) node[pos=0.7, above, sloped]{\fontsize{14}{14}\selectfont $z(a)$};
    
            \fill[niceblue, opacity=0.3] (0,5) -- (0,0) -- (1.2,0) -- (5,3.8) -- (5,5) -- cycle;
    
            \fill[orange, opacity=0.3] (1.2,0) -- (5,0) -- (5,3.8) -- cycle;
    
            \node[darkdarkgreen, anchor=south] at (1.8,2.5) {\fontsize{20}{20}\selectfont $B$};
            \node[darkdarkgreen, anchor=south] at (3.6,0.6) {\fontsize{20}{20}\selectfont $A$};

            \draw[->] (0,0) -- (5.1,0) node[right] {$a$};
            \draw[->] (0,0) -- (0,5.1) node[above] {$b$};
            
        \end{tikzpicture}              
            \caption{Figure \ref{fg:nomoney}: Sorting pattern implementable with tolls alone.}
            \label{fg:nomoney}
        \end{subfigure}
        \hspace{1.3cm}
        \begin{subfigure}[b]{0.4\linewidth}
            \centering
        \begin{tikzpicture}[scale=0.7]
            \draw[darkdarkgreen, very thick] (0,0) -- (5,2.7) node[pos=0.7, above, sloped]{\fontsize{14}{14}\selectfont $z(a)$};
    
            \fill[niceblue, opacity=0.3] (0,5) -- (0,0) -- (5,2.7) -- (5,5) -- cycle;
    
            \fill[orange, opacity=0.3] (0,0) -- (5,0) -- (5,2.7) -- cycle;

            \node[darkdarkgreen, anchor=south] at (1.8,3) {\fontsize{20}{20}\selectfont $B$};
            \node[darkdarkgreen, anchor=south] at (3.1,0.4) {\fontsize{20}{20}\selectfont $A$};
                                            \draw[->] (0,0) -- (5.1,0) node[right] {$a$};
                                            \draw[->] (0,0) -- (0,5.1) node[above] {$b$};
        \end{tikzpicture}         
            \caption{Figure \ref{fg:costs}: Sorting pattern whose implementation requires damages.}
            \label{fg:costs}
        \end{subfigure}
    \end{figure}

\paragraph{Why damages can be optimal.} As the following example shows, the optimal sorting pattern may sometimes require the use of damages:

\begin{example}\label{ex:1}
Suppose the designer cares only about welfare, i.e., $\gamma=0$; consider the density
        \vspace{-0.3cm}
\begin{center}
        \hspace{-2cm}
        \vspace{-0.3cm}
\begin{minipage}[c]{0.60\textwidth}
\[
f(a,b) =
\begin{cases}
\displaystyle
\varepsilon \,\frac{2}{\bigl(\tfrac12 - \varepsilon\bigr)^2}, 
& \text{if}\ \ b - a \ge \tfrac12 + \varepsilon,\\[1.5em]
\displaystyle
\frac{1}{3}\,\frac{2}{\varepsilon - \varepsilon^2}, 
& \text{if}\ \  b - a \in \bigl[\tfrac12,\tfrac12 + \varepsilon\bigr),\\[1.5em]
\displaystyle
\tfrac87\,\bigl(\tfrac23 - \varepsilon\bigr), 
& \text{if}\ \ b - a < \tfrac12,
\end{cases}
\]
\end{minipage}
\hspace{0.5cm}
\begin{minipage}[c]{0.32\textwidth}
\centering
\begin{tikzpicture}[scale=0.75]
    \fill[gray, opacity=0.3] 
        (0,0) -- (5,0) -- (5,5) -- (2.7,5) -- (0,2.3) -- cycle;

    \fill[red, opacity=0.6] 
        (0,2.3) -- (2.7,5) -- (2.3,5) -- (0,2.7) -- cycle;

    \fill[yellow, opacity=0.65] 
        (0,2.7) -- (0,5) -- (2.3,5) -- cycle;

    \draw[->] (0,0) -- (5.1,0) node[right] {$a$};
    \draw[->] (0,0) -- (0,5.1) node[above] {$b$};

    \draw[decorate, decoration={brace, amplitude=7pt}, thick]
        (0,0) -- (0,2.3)
        node[midway, left=6pt] {$\frac{1}{2}$};

    \draw[decorate, decoration={brace, amplitude=7pt}, thick]
        (0,2.3) -- (0,2.7)
        node[midway, left=6pt] {$\varepsilon$};
\end{tikzpicture}
\end{minipage}
\end{center}
and supplies given by $s_A = \frac23-\varepsilon$ and $s_B = \frac13+\varepsilon.$ For $\varepsilon>0$ sufficiently small, a mechanism using only tolls is not optimal.
\end{example}

The distribution from the example is illustrated in the figure to the right. The probability masses equal $\varepsilon$ in the yellow area, $1/3$ in the red area, and $2/3-\varepsilon$ in the gray area. The supply of good $B$ is chosen to exactly match the total mass of the yellow and red areas, while the supply of good $A$ matches the mass in the gray area.

Consider a mechanism for this distribution which uses only tolls but not damages. Discarding any supply of either good would not be helpful, so it is without loss to consider only mechanisms in which the whole supply is allocated. Without damages, this can only be achieved by giving out good $A$ for free, and giving good $B$ with a toll $c=1/2$. Indeed, this mechanism induces the sorting pattern illustrated in Figure \ref{fig:example1}, with types shaded in blue receiving good \(B\) and types shaded in orange receiving good \(A\). As discussed above, the boundary splitting the two regions is angled at $45$ degrees.

Now, notice that for agents in the strip between the solid and dashed lines in Figure \ref{fig:example1}, the surplus from getting good $B$ over getting good $A$ is at most $\varepsilon$. This is because most of their surplus is consumed by the toll $c=1/2$. Consider then the case where $\varepsilon\to 0$. As $\varepsilon$ falls, this surplus goes to zero for the whole $1/3$-mass of agents in the aforementioned strip. Similarly, the mass of agents above the dashed line, equal to $\varepsilon$, also tends to zero. Consequently, total welfare then tends to that which would have resulted from all agents getting good $A$ for free.

\begin{figure}[h!]
        \centering
        \begin{subfigure}[b]{0.4\linewidth}
        \centering
\begin{tikzpicture}[scale=0.8]
    \draw[->] (0,0) -- (5.1,0) node[right] {$a$};
    \draw[->] (0,0) -- (0,5.1) node[above] {$b$};

    \draw[darkdarkgreen, very thick] (0,2.3) -- (2.7,5);

    \fill[orange, opacity=0.3] 
        (0,0) -- (5,0) -- (5,5) -- (2.7,5) -- (0,2.3) -- cycle;
    \fill[niceblue, opacity=0.3] 
        (0,2.3) -- (0,5) -- (2.7,5) -- cycle;

    \draw[darkdarkgreen, dashed] (0,2.7) -- (2.3,5);

    \draw[decorate, decoration={brace, amplitude=7pt}, thick]
        (0,0) -- (0,2.3)
        node[midway, left=6pt] {$\frac{1}{2}$};

    \draw[decorate, decoration={brace, amplitude=7pt}, thick]
        (0,2.3) -- (0,2.7)
        node[midway, left=6pt] {$\varepsilon$};
\end{tikzpicture}

            \caption{Figure \ref{fig:example1}: The mechanism in Example \ref{ex:1} that gives good $B$ with a toll.}
            \label{fig:example1}
        \end{subfigure}
        \hspace{1.3cm}
        \begin{subfigure}[b]{0.4\linewidth}
            \centering
\begin{tikzpicture}[scale=0.8]
    \draw[->] (0,0) -- (5.1,0) node[right] {$a$};
    \draw[->] (0,0) -- (0,5.1) node[above] {$b$};

    \draw[darkgreen, very thick] (0,0) -- (1.78571,5);

    \draw[darkred, very thick] (0,2.7) -- (2.3,5);

    \fill[orange, opacity=0.3]
        (0,0) -- (5,0) -- (5,5) -- (1.78571,5) -- cycle;

    \fill[niceblue, opacity=0.3]
        (0,0) -- (1.78571,5) -- (0,5) -- cycle;

    \draw[decorate, decoration={brace, amplitude=7pt}, thick]
        (0,0) -- (0,2.3)
        node[midway, left=6pt] {$\frac{1}{2}$};

    \draw[decorate, decoration={brace, amplitude=7pt}, thick]
        (0,2.3) -- (0,2.7)
        node[midway, left=6pt] {$\varepsilon$};
\end{tikzpicture}

            \caption{Figure \ref{fg:betterexample}: The mechanism in Example \ref{ex:1} that damages good $B$.}
            \label{fg:betterexample}
        \end{subfigure}
    \end{figure}

By contrast, consider the mechanism that uses no tolls but damages good $B$ to the point where the resulting boundary satisfies both goods' supply constraints (Figure \ref{fg:betterexample}). Here too, total welfare of agents above the red line becomes negligible as $\varepsilon \to 0$. Note, however, that agents who are below the red line and far to the left of the green line now benefit substantially from getting good $B$ over getting good $A$ for free. These agents have a strong relative preference for $B$ over $A$, and thus strongly prefer even damaged $B$ to undamaged $A$. Since the mass of such agents does not go to zero as $\varepsilon$ decreases, this mechanism creates substantial gains from allocating good $B$ over $A$ even in the limiting case. This is in contrast to the toll-only mechanism, where the gains from allocating
\(B\) rather than \(A\) are largely absorbed by the toll differential used to
screen agents. Intuitively, this occurs because most recipients of good \(B\)
under the toll-based mechanism have value differences \(b-a\) very close to the
market-clearing toll differential \(c_B-c_A=1/2\), and therefore obtain little
surplus from receiving \(B\) rather than \(A\).

I now provide sufficient conditions under which damaging goods is and is not optimal. I discuss the economic meaning of these conditions after stating both results.

\paragraph{When are damages not optimal?} 
For analytical convenience, I establish the first result under the following technical restrictions.
\begin{assumption}\label{ass:1}
$F$ has a strictly positive, Lipschitz continuous density $f$ on $[0,1]^2$.
\end{assumption}
\begin{assumption}\label{ass:2}
The designer is restricted to quality rules $x:[0,1]^2\to [0,1]$ that are piecewise continuously differentiable in each dimension of the type.
\end{assumption}
Note that the latter assumption permits allocation rules to have finitely many discontinuities. 

\begin{theorem}\label{th:1}
Suppose Assumptions \ref{ass:1} and \ref{ass:2} hold and that the inverse anti-hazard rates of the conditional distributions,
\begin{equation}\label{eq:anitihazards}
\frac{F_{A|B}(a|b)}{f_{A|B}(a|b)}, \quad \frac{F_{B|A}(b|a)}{f_{B|A}(b|a)},
\end{equation}
are increasing in $a$ and $b$, and for each ratio at least one of the two monotonicities is strict. Then the optimal mechanism offers only two options: undamaged goods $A$ and $B$ with tolls of $c_A^*$ and $c_B^*$, respectively. These tolls are chosen so that the whole supply of both goods is allocated.
\end{theorem}

Thus, under the above conditions, the optimal mechanism replicates the competitive equilibrium allocation of goods: agents whose values for both goods are sufficiently low receive nothing, while all other agents get an undamaged version of one of the goods and pay the associated toll (Figure \ref{fig:th1}). While the optimality of this mechanism is immediate when tolls are welfare-neutral, the result asserts it remains optimal when the designer values a unit of toll revenue at any $\gamma\in [0,1]$ and so, in particular, when tolls are completely wasteful. Throughout the rest of the paper, I refer to this mechanism as the \emph{market-clearing toll mechanism} and write \(c_A^*,c_B^*\) for these market-clearing tolls. I also use 
\[
 z_0(a):=a+c_B^*-c_A^*,
\]
to denote the boundary generated by this mechanism and write $(\bar a^*,\bar b^*)$ for its end point.

\begin{figure}[h!]
    \centering
\begin{tikzpicture}[scale=0.70]
    \coordinate (Z) at (1.2,1.6);   
    \coordinate (T) at (4.6,5);     

    \draw[->, thick] (0,0) -- (5.1,0) node[right] {$a$};
    \draw[->, thick] (0,0) -- (0,5.2) node[above] {$b$};

    \draw[thick] (1.2,0.08) -- (1.2,-0.08)
        node[below] {$c_A^*$};

    \draw[thick] (0.08,1.6) -- (-0.08,1.6)
        node[left] {$c_B^*$};

    \draw[darkgreen, very thick] (Z) -- (T);

    \draw[darkdarkgreen, very thick] (1.2,0) -- (Z);
    \draw[darkdarkgreen, very thick] (0,1.6) -- (Z);

    \fill[niceblue, opacity=0.3]
        (0,5) -- (0,1.6) -- (Z) -- (T) -- cycle;

    \fill[orange, opacity=0.3]
        (5,0) -- (1.2,0) -- (Z) -- (T) -- (5,5) -- cycle;

    \node[darkdarkgreen, font=\fontsize{20}{20}\selectfont] at (1.5,3.8) {$B$};
    \node[darkdarkgreen, font=\fontsize{20}{20}\selectfont] at (3.3,1.6) {$A$};
    \node[darkdarkgreen, font=\large] at (0.6,0.8) {$\varnothing$};
\end{tikzpicture}
        \caption{Optimal allocation under the conditions of Theorem \ref{th:1}.}
        \label{fig:th1}
\end{figure}

Note also that when the values for the two goods are independent, $F(a,b)=F_A(a)F_B(b)$, the distributional condition in the theorem reduces to the requirement that
\[
\frac{F_A(a)}{f_A(a)}, \quad \frac{F_B(b)}{f_B(b)}
\]
be strictly increasing. This holds, for example, when the marginals are uniform, truncated normal,
truncated decreasing exponential, or Beta\((\alpha,\beta)\) distributions with
\(\alpha,\beta\geq 1\). Before providing an intuition for this distributional condition, I state Theorem \ref{thm:2}, which provides conditions under which the optimal mechanism \emph{does} use damages.

\paragraph{When are damages optimal?} Define the following objects for \(a\in[c_A^*,1]\) and \(b\in[c_B^*,1]\):
\begin{equation}\label{eq:definitiosssspartcon}
        P_A^a
        :=
        \int_0^{\min\{1,z_0(a)\}} f(a,t)\,dt,
        \quad
        P_B^b
        :=
        \int_0^{\min\{1,z_0^{-1}(b)\}} f(t,b)\,dt,
        \quad
        P_{AB}
        :=
        \int_{c_A^*}^{{\bar a}^*}
        f(a,z_0(a))\,da.
\end{equation}
\(P_A^a\) is the density of agents who value good $A$ at $a$ and choose it under the market-clearing toll mechanism; \(P_B^b\) is defined analogously. \(P_{AB}\) measures the density of agents along the interior \(A\)-\(B\) boundary (Figure \ref{fg:kinkint1}).

\begin{theorem}\label{thm:2}
Suppose Assumption \ref{ass:1} holds and the market-clearing toll for good \(B\) satisfies \(0<c_B^*<1\). Suppose also that for some \(\tilde b\in(c_B^*,1)\), we have
\begin{equation}\label{eq:local-damage-condition-main}
        \int_{c_B^*}^{\tilde b}
        \bigl((1-\gamma)\tilde b-b\bigr)P_B^b\,db
        >
        (1-\gamma)
        \left(
        \alpha s_A+\beta s_B
        \right),
\end{equation}
where
\begin{equation}\label{eq:alphabetadef}
        \alpha
        =
        \frac{
        P_B^{c^*_B}\left(P_{AB}(\tilde b-c_B^*)-Q\right)
        }{
        (P_A^{c^*_A}+P_{AB})(P_B^{c^*_B}+P_{AB})-P_{AB}^2
        },
\qquad
        \beta
        =
        \frac{
        P_A^{c^*_A}Q
        +(P_A^{c^*_A}+P_{AB})P_B^{c^*_B}(\tilde b-c_B^*)
        }{
        (P_A^{c^*_A}+P_{AB})(P_B^{c^*_B}+P_{AB})-P_{AB}^2
        },
\end{equation}
and
\[
        Q
        :=
        \int_{c_A^*}^{{\bar a}^*}
        \bigl(\tilde b-z_0(a)\bigr)_+
        f(a,z_0(a))\,da.
\]
Then the optimal mechanism uses damages.
\end{theorem}

To understand this result, observe first that the market-clearing toll mechanism
is optimal among mechanisms that do not use damages. I prove the following lemma in the appendix.

\begin{lemma}\label{lemma:bestundamaged}
Suppose Assumption \ref{ass:1} holds. Then, for every
\(\gamma\in[0,1]\), the market-clearing toll mechanism is optimal
among mechanisms that do not use damages.
\end{lemma}

Intuitively, when damages are not used, each kind of good must be offered at a single toll. The mechanism with market-clearing tolls then maximizes welfare among all feasible ones: if either good's supply constraint were slack, its toll could be lowered slightly, increasing agents' utility while preserving feasibility. Also, by the standard efficiency property of competitive equilibrium, these tolls maximize allocative efficiency:
\begin{equation}\label{eq:allocceff}
        \int
        \left[
        \mathbb{1}_{y(a,b)=A}x(a,b)a
        +
        \mathbb{1}_{y(a,b)=B}x(a,b)b
        \right]
        dF(a,b).
\end{equation}
Since the designer's objective for any \(\gamma\in[0,1]\) is a nonnegative linear combination of welfare and allocative efficiency, the market-clearing toll mechanism is optimal among undamaged mechanisms for every \(\gamma\in[0,1]\).

It therefore suffices to show that a local perturbation introducing damages can
improve on the market-clearing toll mechanism. The proof constructs such a
perturbation in two steps. The first is the \emph{new-option step}, in which the designer adds a slightly damaged option for good $B$ to the market-clearing menu. This new
option provides \(B\) at quality \(1-\varepsilon\) and offers a toll discount
of \(\varepsilon\tilde b\). Agents who previously chose \(B\) prefer the new
damaged option if and only if
\[
        (1-\varepsilon)b-(c_B^*-\varepsilon\tilde b)
        >
        b-c_B^*, \quad \Longleftrightarrow \quad b<\tilde b.
\]
Thus, among old \(B\)-choosers, those with \(b<\tilde b\) get extra rents from this new option. However, the damaged option also attracts some agents who previously chose
either good \(A\) or nothing. Such agents are shaded dark-blue in Figure
\ref{fg:kinkint2}. Because of these additional \(B\)-consumers, demand for good \(B\) now exceeds its
supply. This necessitates a \emph{toll-correction step} in which the designer changes the tolls on goods \(A\) and \(B\) to restore market clearing (Figure \ref{fg:kinkint3}).

\begin{figure}[h!]
        \centering
        \begin{subfigure}[b]{0.31\linewidth}
            \centering
\begin{tikzpicture}[scale=0.70]
    \coordinate (Z) at (1.2,1.6);   
    \coordinate (T) at (4.6,5);     

    \coordinate (A0) at (3.15,0);
    \coordinate (A1) at (3.15,3.55); 

    \coordinate (B0) at (0,3.05);
    \coordinate (B1) at (2.65,3.05); 

    \fill[niceblue, opacity=0.3]
        (0,5) -- (0,1.6) -- (Z) -- (T) -- cycle;

    \fill[orange, opacity=0.3]
        (5,0) -- (1.2,0) -- (Z) -- (T) -- (5,5) -- cycle;

    \draw[darkred, very thick] (A0) -- (A1)
        node[pos=0.52, right, xshift=4pt, font=\large] {$P_A^a$};

    \draw[niceblue, very thick] (B0) -- (B1)
        node[pos=0.48, above, yshift=4pt, font=\large] {$P_B^b$};

    \draw[thick] (3.15,0) -- (3.15,-0.12) node[below] {$a$};
    \draw[thick] (0,3.05) -- (-0.12,3.05) node[left] {$b$};

    \draw[darkgreen, very thick] (0,1.6) -- (Z);
    \draw[darkgreen, very thick] (1.2,0) -- (Z);
    \draw[darkgreen, very thick] (Z) -- (T)
        node[pos=0.58, above, yshift=7pt, font=\large] {$P_{AB}$};

    \draw[->, thick] (0,0) -- (5.1,0) node[right] {};
    \draw[->, thick] (0,0) -- (0,5.2) node[above] {};
\end{tikzpicture}
        \vspace{0.24cm}
        \caption{Figure \ref{fg:kinkint1}}
        \label{fg:kinkint1}
        \end{subfigure}
        \begin{subfigure}[b]{0.31\linewidth}
        \centering
\begin{tikzpicture}[scale=0.70]
    \coordinate (Z) at (1.2,1.6);   
    \coordinate (T) at (4.6,5);     

    \coordinate (K) at (2.85,3.25);

    \coordinate (R) at (1.2,0.75);

    \coordinate (H) at (0,0.75);

    \fill[niceblue, opacity=0.3]
        (0,5) -- (0,1.6) -- (Z) -- (T) -- cycle;

    \fill[orange, opacity=0.3]
        (5,0) -- (1.2,0) -- (R) -- (K) -- (T) -- (5,5) -- cycle;

    \fill[niceblue, opacity=0.6]
        (0,1.6) -- (Z) -- (K) -- (R) -- (H) -- cycle;

    \draw[darkgreen, very thick] (0,1.6) -- (Z);
    \draw[darkgreen, very thick] (1.2,0) -- (Z);
    \draw[darkgreen, very thick] (Z) -- (T);

    \draw[red, very thick] (R) -- (K);
    \draw[red, very thick] (H) -- (R);

    \fill[red] (K) circle (2.5pt);

    \draw[thick] (0,3.25) -- (-0.12,3.25) node[left] {$\tilde b$};

    \draw[->, thick] (0,0) -- (5.1,0) node[right] {};
    \draw[->, thick] (0,0) -- (0,5.2) node[above] {};
\end{tikzpicture}
        \vspace{0.7cm}
        \caption{Figure \ref{fg:kinkint2}}
        \label{fg:kinkint2}
        \end{subfigure}
        \begin{subfigure}[b]{0.31\linewidth}
            \centering
\begin{tikzpicture}[scale=0.70]
    \coordinate (Z) at (1.2,1.6);      
    \coordinate (T) at (4.6,5);        

    \coordinate (V) at (1.45,0);       
    \coordinate (R) at (1.45,1.00);    
    \coordinate (H) at (0,1.00);       

    \coordinate (Knew) at (2.65,3.25);
    \coordinate (Tnew) at (4.40,5);

    \fill[niceblue, opacity=0.3]
        (0,5) -- (H) -- (R) -- (Knew) -- (Tnew) -- cycle;

    \fill[orange, opacity=0.3]
        (5,0) -- (V) -- (R) -- (Knew) -- (Tnew) -- (5,5) -- cycle;

    \draw[darkgreen, thick, dashed] (0,1.6) -- (Z);
    \draw[darkgreen, thick, dashed] (1.2,0) -- (Z);
    \draw[darkgreen, thick, dashed] (Z) -- (T);

    \draw[red, very thick] (H) -- (R);
    \draw[red, very thick] (V) -- (R);
    \draw[red, very thick] (R) -- (Knew);
    \draw[red, very thick] (Knew) -- (Tnew);

    \draw[thick] (1.45,0) -- (1.45,-0.12)
        node[below] {$c_A^*+\alpha\varepsilon$};

    \draw[->, thick] (0,0) -- (5.1,0) node[right] {};
    \draw[->, thick] (0,0) -- (0,5.2) node[above] {};
\end{tikzpicture}
        \caption{Figure \ref{fg:kinkint3}}
        \label{fg:kinkint3}
        \end{subfigure}
\caption{The first subfigure illustrates the definitions in
\eqref{eq:definitiosssspartcon}. The second shows the effect of introducing the
damaged option to the market-clearing toll menu. The last subfigure illustrates the toll-correction step
that restores the supply constraints after the damaged option is added.}
\label{figure:fig777}
\end{figure}

Condition \eqref{eq:local-damage-condition-main} then compares the welfare
consequences of these two steps. The left-hand side accounts for the gains to
old \(B\)-choosers who switched to the damaged option. Such an agent loses quality worth
\(\varepsilon b\), but receives a toll reduction worth
\(\varepsilon\tilde b\). Since the designer values a unit of toll revenue at
\(\gamma\), the net first-order gain from such a type is
\[
        (1-\gamma)\tilde b-b.
\]
The right-hand side captures the cost
of the toll-correction step. Restoring market clearing requires raising the
original tolls by \(\alpha \varepsilon\) and \(\beta \varepsilon\), and the net social cost of a
one-unit toll increase is \(1-\gamma\). Hence, the total first-order cost of the
toll correction is
\[
        (1-\gamma)
        \left(
        \alpha s_A+\beta s_B
        \right).
\]

The requirement in Theorem \ref{thm:2} becomes especially simple when the total
supply adds up to one and the designer does not value tolls:
\begin{corollary}\label{cor:damagesnormalized}
Suppose Assumption \ref{ass:1} holds, tolls are fully wasteful, \(\gamma=0\), supplies add up to one, \(s_A+s_B=1\), and the
market-clearing toll for good \(B\) satisfies \(0<c_B^*<1\). Then damages are
optimal if for some \(\tilde b\in(c_B^*,1)\) we have
\begin{equation}\label{eq:covariancecorr}
        \operatorname{Cov}
        \left(
        \frac{F_{A\mid B}(b-c_B^*\mid b)}
        {f_{A\mid B}(b-c_B^*\mid b)}, \
        (\tilde b-b)_+
        \ \middle|\ a=b-c_B^*
        \right)
        >0.
\end{equation}
\end{corollary}
Formally, the conditional covariance is taken with respect to the probability
measure on \([c_B^*,1]\) with density proportional to
\(f(b-c_B^*,b)\). I now explain this covariance condition intuitively.

First, note that since
supplies add up to \(1\) and the market-clearing toll for good \(B\) is
positive, the other toll is zero: \(c_A^*=0\). Indeed, if both tolls were
positive, some agents would not choose either good and some supply would remain
unallocated. 

Let us then consider a perturbation that introduces a slightly
damaged option in this setting, analogous to the one used for Theorem
\ref{thm:2}. Figures \ref{fg:corrperturb1} and \ref{fg:corrperturb3}
illustrate the new-option and toll-correction steps for this perturbation. 
We first focus on the new-option step, which has two effects. The first is the
\emph{inframarginal effect}: old \(B\)-choosers with \(b<\tilde b\) benefit
from the damaged option. The second is the \emph{switcher effect}:
some agents who barely preferred \(A\) now switch into \(B\); these switchers are close to the old boundary \(z_0\).

\begin{figure}[h!]
        \centering
        \begin{subfigure}[c]{0.31\linewidth}
            \centering
\begin{tikzpicture}[scale=0.70]
    \coordinate (Z) at (0,1.50);     
    \coordinate (T) at (3.50,5);     

    \coordinate (K) at (2.75,4.25);

    \coordinate (H) at (0,0.55);

    \fill[niceblue, opacity=0.3]
        (0,5) -- (Z) -- (T) -- cycle;

    \fill[orange, opacity=0.3]
        (5,0) -- (0,0) -- (H) -- (K) -- (T) -- (5,5) -- cycle;

    \fill[niceblue, opacity=0.6]
        (Z) -- (K) -- (H) -- cycle;

    \draw[red, very thick] (H) -- (K);

    \draw[darkgreen, very thick] (Z) -- (T);

    \fill[red] (K) circle (2.5pt);

    \draw[thick] (0,4.25) -- (-0.12,4.25) node[left] {$\tilde b$};

    \draw[thick] (0,1.50) -- (-0.12,1.50) node[left] {$c_B^*$};

    \draw[->, thick] (0,0) -- (5.1,0) node[right] {};
    \draw[->, thick] (0,0) -- (0,5.2) node[above] {};
\end{tikzpicture}
            \caption{Figure \ref{fg:corrperturb1}}
            \label{fg:corrperturb1}
        \end{subfigure}
        \begin{subfigure}[c]{0.31\linewidth}
            \centering
\begin{tikzpicture}[scale=0.70]
    \coordinate (Z) at (0,1.50);      
    \coordinate (T) at (3.50,5);      

    \coordinate (Hnew) at (0,0.95);

    \coordinate (Knew) at (2.45,4.25);
    \coordinate (Tnew) at (3.20,5);

    \fill[niceblue, opacity=0.3]
        (0,5) -- (Hnew) -- (Knew) -- (Tnew) -- cycle;

    \fill[orange, opacity=0.3]
        (5,0) -- (0,0) -- (Hnew) -- (Knew) -- (Tnew) -- (5,5) -- cycle;

    \draw[darkgreen, thick, dashed] (Z) -- (T);

    \draw[red, very thick] (Hnew) -- (Knew);
    \draw[red, very thick] (Knew) -- (Tnew);

    \draw[thick] (0,4.25) -- (-0.12,4.25) node[left] {$\tilde b$};

    \fill[red] (Knew) circle (2.5pt);

    \draw[->, thick] (0,0) -- (5.1,0) node[right] {};
    \draw[->, thick] (0,0) -- (0,5.2) node[above] {};
\end{tikzpicture}
            \caption{Figure \ref{fg:corrperturb3}}
            \label{fg:corrperturb3}
        \end{subfigure}
        \begin{subfigure}[c]{0.31\linewidth}
            \centering
\begin{tikzpicture}[scale=0.70]

    \coordinate (Z) at (0,1.50);
    \coordinate (T) at (3.50,5);
    \coordinate (B) at (2.15,3.65);
    \coordinate (Bleft) at (0,3.65);

    \fill[niceblue, opacity=0.3] (0,5) -- (Z) -- (T) -- cycle;
    \fill[orange, opacity=0.3] (0,0) -- (5,0) -- (5,5) -- (T) -- (Z) -- cycle;

    \draw[darkgreen, very thick] (Z) -- (T);
    \draw[black, very thick] (Bleft) -- (B);
    \fill[black] (B) circle (3.5pt);

    \draw[thick] (0,1.50) -- (-0.12,1.50) node[left] {$c_B^*$};
    \draw[thick] (0,3.65) -- (-0.12,3.65) node[left] {$b$};

    \draw[->, thick] (0,0) -- (5.1,0) node[right] {};
    \draw[->, thick] (0,0) -- (0,5.2) node[above] {};
\end{tikzpicture}%
            \caption{Figure \ref{fg:corrperturb2}}
            \label{fg:corrperturb2}
        \end{subfigure}
\caption{The first subfigure shows how the type split changes after the
damaged \(B\)-option is added but before tolls are readjusted. The second
subfigure shows how the type split changes after the toll-correction step. The
third subfigure illustrates inframarginal \(B\)-choosers
and switchers on a horizontal slice.}
\label{fig:corrperturb}
\end{figure}

Note that the inframarginal effect is beneficial: it gives additional rents to old
\(B\)-choosers who take the damaged option. The switcher effect, on the other hand, is costly:
additional switchers increase demand for \(B\), requiring a larger increase in
the \(B\)-toll during the toll-correction step. Thus, to assess whether the perturbation is worth it, we compare these two effects. To that end, fix \(b<\tilde b\) and consider the horizontal slice of the type space illustrated in Figure \ref{fg:corrperturb2}. Under the original mechanism, types with $a<b-c_B^*$ choose \(B\), while the marginal type \(a=b-c_B^*\) is indifferent between \(A\) and \(B\). The inframarginal beneficiaries are therefore the types in the black segment, whose density-weighted mass along the slice is $\int_0^{b-c_B^*} f(t,b)\,dt.$ The switchers come from a small neighborhood of the boundary point shown in black, so their mass from a marginal movement of the boundary is proportional to $f(b-c_B^*,b).$ The ratio of inframarginal beneficiaries to marginal switchers on the slice is therefore
\begin{equation}\label{eq:antihazarddd}
        \frac{
        \int_0^{b-c_B^*} f(t,b)\,dt
        }
        {
        f(b-c_B^*,b)
        }
=
        \frac{F_{A\mid B}(b-c_B^*\mid b)}
        {f_{A\mid B}(b-c_B^*\mid b)}.
\end{equation}
Thus, the ratio in \eqref{eq:antihazarddd} measures the strength of the inframarginal effect relative to the switcher effect on the horizontal slice indexed by \(b\). Condition
\eqref{eq:covariancecorr} then says that the perturbation is beneficial when
the ratio in \eqref{eq:antihazarddd} is positively related to
\[
        (\tilde b-b)_+.
\]
Intuitively, this term measures the gain from taking the damaged option over the
old \(B\)-option along the slice with value \(b\). It is largest for low values
of \(b\), since these agents are more willing to accept damage to good \(B\) in
exchange for a lower toll. Consequently, condition \eqref{eq:covariancecorr}
requires that the ratio of the inframarginal to switcher effects be more
favorable on the horizontal slices where the damaged option has the strongest
effect.

\paragraph{Intuition for the anti-hazard-rate condition in Theorem \ref{th:1}.}
The covariance condition also clarifies the role of the monotonicity condition
on the conditional inverse anti-hazard rates in \eqref{eq:anitihazards}. This
condition ensures that
\begin{equation}\label{eq:hazardfunctionaaa}
        b\mapsto
        \frac{F_{A\mid B}(b-c_B^*\mid b)}
        {f_{A\mid B}(b-c_B^*\mid b)}
\end{equation}
is increasing along the old boundary \(z_0\). Since \(b\mapsto(\tilde b-b)_+\) is decreasing, the covariance in \eqref{eq:covariancecorr} is then negative, and introducing a damaged option is not locally beneficial.

\paragraph{Affiliated values and the optimality of damages.}
The monotonicity condition on the inverse anti-hazard ratio in \eqref{eq:antihazarddd} is related to affiliation. Following \citet{milgromweber}, values for goods $A$ and $B$ are affiliated when the joint density is log-supermodular. If $f$ is strictly positive and twice continuously differentiable, this is equivalent to
\begin{equation}
\frac{\partial^2}{\partial a\,\partial b}\log f(a,b) \geq 0.
\end{equation}
The reverse inequality corresponds to log-submodularity, or negative affiliation.

Note that the conditions of Theorem \ref{th:1}, which require increasing conditional
inverse anti-hazard rates, rule out strictly positive affiliation
between the two values but are compatible with negative dependence. To see this, suppose \(f\) is differentiable; then:
\begin{equation}\label{eq:aff}
\frac{d}{db}\frac{F_{A\mid B}(a\mid b)}{f_{A\mid B}(a\mid b)}
=
\frac{1}{f(a,b)}
\int_0^a
f(t,b)
\left(
\frac{\partial}{\partial b}\log f(t,b)
-
\frac{\partial}{\partial b}\log f(a,b)
\right)
\,dt .
\end{equation}
Since positive affiliation makes \(\partial_b \log f(a,b)\) increasing in \(a\), the integrand in \eqref{eq:aff} is negative, so affiliation pushes the inverse anti-hazard ratio downward as \(b\) increases. Negative affiliation reverses this inequality: the integrand is positive, and
the inverse anti-hazard ratio is increasing in \(b\), as required by the
anti-hazard-rate condition in Theorem \ref{th:1}.

Moreover, positive affiliation makes damages more likely to be optimal. Along the old boundary \(z_0\), affiliation tends to make the ratio in \eqref{eq:antihazarddd} lower at high values of \(b\), and therefore relatively higher at low values of \(b\). Since \((\tilde b-b)_+\) is also highest at low values of \(b\), this makes the covariance in \eqref{eq:covariancecorr} more likely to be positive.

The connection between affiliation and the optimality of damages has a simple
economic interpretation. If values are positively affiliated, then a high value
of \(b\) is good news about \(a\). Thus, on high-\(b\) slices, old
\(B\)-choosers tend to have relatively high values of \(a\), placing many of
them close to the old \(A\)-\(B\) boundary. Making \(B\) more attractive on
these slices therefore creates relatively many switchers per inframarginal
beneficiary. Conversely, a low value of \(b\) is bad news about \(a\), so old
\(B\)-choosers on low-\(b\) slices are more likely to lie far to the left of
the boundary. Low-\(b\) slices therefore contain many inframarginal takers
relative to switchers, whereas high-\(b\) slices contain few. This is precisely
the pattern that makes damages attractive. Negative affiliation reverses it, making the market-clearing toll mechanism more likely to be optimal.

The following example illustrates how the sign and strength of affiliation affect the optimality of damages.

\begin{example}
\label{ex:affiliated-profitable-damages}
Assume tolls are fully wasteful, so \(\gamma=0\). For each \(\lambda\in\mathbb{R}\), consider the density
\[
        f_\lambda(a,b)
        \propto
        e^{\lambda ab}
        \qquad (a,b)\in[0,1]^2 .
\]
Fix \(c_B^*\in(0,1)\). For each \(\lambda\), choose supplies equal to the
masses cleared by the tolls \(c_A^*=0\) and \(c_B^*\). Then, for all \(\lambda\leq 0\), the optimal mechanism
does not use damages, while, for all sufficiently large \(\lambda\), the optimal
mechanism uses damages.
\end{example}

Indeed, \(\lambda\) directly indexes the strength of affiliation, since $\frac{\partial^2}{\partial a\,\partial b}\log f_\lambda(a,b)
        =
        \lambda .$

\begin{remark}
The strength and direction of affiliation reflect the relative importance of horizontal and vertical components of preferences. My analysis therefore yields different prescriptions depending on which component is likely to dominate in a particular setting. 

Consider, for example, routine administrative appointments, such as vehicle registration renewals, for which agents can choose among several geographically dispersed offices. In this case, preferences are mostly horizontal: overall need for an appointment varies little, but people prefer offices located closer to where they live; indeed, preferences for proximity naturally generate negative affiliation in values across offices. In such settings, tolls alone are more likely to be optimal: more congested offices should use higher booking fees or walk-in queues rather than allowing waitlists to grow long.

By contrast, consider public housing programs with broad eligibility criteria, where households may differ substantially in their overall need for subsidized housing.\footnote{In many European countries, affordable housing is not reserved for the most disadvantaged households but is available to a much wider segment of the population \citep{whitehead2007social}.
} Although they also have geographic preferences across developments, these may be relatively weak compared with this vertical component of demand, generating positive affiliation in values. In such settings, differentiated waitlists may be optimal, possibly in combination with toll instruments such as differential rent subsidies.
\end{remark}

\section{Proof of Theorem \ref{th:1}}\label{sec:proofth1}

This section gives the main steps of the proof of Theorem \ref{th:1}; I show the referenced lemmas in the appendix. The argument is based on the
boundary representation from the previous section: any implementable mechanism
is summarized by good-specific indirect utilities \(U_A\) and \(U_B\), together
with a boundary \(z\) that sorts types between the two goods. This lets us treat
the two-dimensional problem as two endogenously connected one-dimensional screening problems.

We first make a useful normalization. Recall that the terminal point of the
boundary lies on the boundary of the unit square: either \(\ool a=1\), so the
boundary reaches the right edge first, or \(\ool b=1\), so it reaches the top
edge first. Since the boundary can equivalently be parametrized by its inverse
\(z^{-1}(b)\) on \([\uul b,\ool b]\), we can break symmetry and without loss
assume that \(\ool b=1\).

\paragraph{Reformulating welfare.} Throughout most of the proof, I show the result for a purely welfarist objective, corresponding to the special case \(\gamma=0\); I then show that if the market-clearing toll mechanism is optimal when tolls have no social value, it remains optimal when the designer puts positive weight on toll revenue. I begin by rewriting aggregate welfare in terms of the boundary \(z\) and the \(A\)-indirect utility \(U_A\). To that end, define the \emph{extended boundary} \(\hat z\) by
\[
\hat{z}(a) =
\begin{cases}
0, & \text{if } a < \uul a,\\
z(a), & \text{if } a \in [\uul a,\, \ool a],\\
1, & \text{if } a > \ool a.
\end{cases}
\]
That is, $\hat z$ equals $z$ on the latter's domain, takes value zero below it and takes value 1 above it (Figure \ref{fig:ccccc}). Throughout, I use $\hat z^{-1}$ to denote its generalized inverse.

\begin{figure}[h!]
    \centering
    \begin{subfigure}[b]{0.3\linewidth}
        \centering
\begin{tikzpicture}[scale=0.70]
    \draw[->] (0,0) -- (5.1,0) node[right] {$a$};
    \draw[->] (0,0) -- (0,5.1) node[above] {$b$};

    \draw[red, very thick] (0,0) -- (1.2,0);

    \coordinate (w1) at (1.4,1.9);
    \coordinate (w2) at (2.05,2.8);
    \coordinate (w3) at (2.5,3.75);
    \coordinate (w4) at (2.9,4.5);

    \draw[red, very thick] plot[smooth, tension=0.7] coordinates {(1.2,1.6) (w1) (w2) (w3) (w4) (3.325,5)};

    \draw[red, very thick] (3.325,5) -- (5,5);

    \draw[darkdarkgreen, thick, dashed] (1.2,0) -- (1.2,1.6);
    \draw[darkdarkgreen, thick, dashed] (0,1.6) -- (1.2,1.6);

    \fill[niceblue, opacity=0.3]
        (0,5) -- (0,1.6) -- plot[smooth, tension=0.7] coordinates {(1.2,1.6) (w1) (w2) (w3) (w4) (3.325,5)} -- cycle;

    \fill[orange, opacity=0.3]
        (5,0) -- (1.2,0) -- (1.2,1.6) -- plot[smooth, tension=0.7] coordinates {(1.2,1.6) (w1) (w2) (w3) (w4) (3.325,5)} -- (5,5) -- cycle;

    \node[red, anchor=south west, rotate=58] at (1.7,2.4) {\fontsize{14}{14}\selectfont $\hat z(a)$};
    \node at (1.2, -0.5) {$\uul{a}$};
    \node at (-0.35, 1.6) {$\uul{b}$};
\end{tikzpicture}
        \caption{Figure \ref{fig:ccccc}: Extended boundary $\hat z$.}
        \label{fig:ccccc}
        \vspace{0.42cm}
    \end{subfigure}
    \hspace{2cm}
    \begin{subfigure}[b]{0.3\linewidth}
        \centering
        \begin{tikzpicture}[scale=0.70]
            \draw[->] (0,0) -- (5.1,0) node[right] {$a$};
            \draw[->] (0,0) -- (0,5.1) node[above] {$b$};

            \draw[darkgreen, dashed] (1.2,0) -- (1.2,1.6);
            \draw[darkgreen, dashed] (0,1.6) -- (1.2,1.6);

            \fill[niceblue, opacity=0.3] (0,5) -- (0,1.6) -- (1.2,1.6)
                                .. controls (1.8,2) .. (2.5,3)
                                .. controls (3.1,3.8) and (3.4,4.5) .. (4,5) -- cycle;

            \fill[orange, opacity=0.3] (5,0) -- (1.2,0) -- (1.2,1.6)
            .. controls (1.8,2) .. (2.5,3)
            .. controls (3.1,3.8) and (3.4,4.5) .. (4,5) -- (5,5) -- cycle;

            \draw[darkred, thick] (1.9,2.2) -- (1.9,0);
            \draw[darkred, thick] (1.9,2.2) -- (0,2.2);

            \draw[darkred, thick] (2.5,3) -- (2.5,0);
            \draw[darkred, thick] (2.5,3) -- (0,3);

            \draw[darkred, thick] (3.5,4.4) -- (3.5,0);
            \draw[darkred, thick] (3.5,4.4) -- (0,4.4);

            \draw[darkgreen, very thick] (1.2,1.6) .. controls (1.8,2) .. (2.5,3)
                                .. controls (3.1,3.8) and (3.4,4.5) .. (4,5);

            \node at (1.2, -0.5) {$\uul{a}$};
            \node at (-0.35, 1.6) {$\uul{b}$};
        \end{tikzpicture}
        \caption{Figure \ref{fig:Lshapes}: Agents on the same $L$-shaped curves have equal utilities.}
        \label{fig:Lshapes}
    \end{subfigure}
\end{figure}
\begin{lemma}\label{lem:objintbyparts}
Consider a mechanism with a boundary $z$ and $A$-indirect utility $U_A$. Total welfare under this mechanism is then given by:
\begin{equation}\label{eq:Wprime}
U_A(1) -\int_0^1 U_A'(a) \cdot F(a,\hat z(a))da.
\end{equation}
\end{lemma}
 
To see why welfare can be expressed using $z$ and $U_A$, note that the utility of an agent choosing good \(A\) depends only on \(a\), while the utility of an agent choosing good \(B\) depends only on \(b\). Since agents on the boundary are indifferent between the two goods, all types lying on the same inverted \(L\)-shaped curve in Figure \ref{fig:Lshapes} have the same utility. We can therefore calculate welfare by integrating over types who choose $A$ while also taking into account the $B$-taking types on the same $L$-shaped curves. Such a calculation yields the expression \eqref{eq:Wprime}.

This form of the objective also bears a resemblance to a Myersonian virtual value which would appear in a single-dimensional setting with a welfarist objective (see, e.g., \cite{CONDORELLI2012613}). Moreover, by the envelope theorem, we can think of $U_A'(a)$ as the quality allocated to agents who receive good $A$ and value it at $a$. Unfortunately, however, the setting does not lend itself to a Myersonian solution method. This is because the shape of the ``virtual value'' itself is endogenous to the choice of boundary. In fact, it is optimizing over the shape of the boundary that poses the greatest difficulty in proving the result.

The proof of the $\gamma=0$ case proceeds in three steps. First, I characterize the feasible pairs \((z,U_A)\). Second, I fix a boundary \(z\) and find the welfare-maximizing \(U_A\) compatible with it. Third, I optimize over the boundary itself. The main step is the last one: using optimal control arguments, I show that the welfare-maximizing boundary must be linear with slope \(1\). Such a boundary is implemented by undamaged goods and market-clearing tolls.

\paragraph{Characterizing feasible $(z,U_A)$.} We say $(z,U_A)$ is \emph{feasible} if there exists a mechanism $(c,x,y)$ with $A,B$-indirect utilities $U_A,U_B$ such that:
\begin{equation}\label{eq:III}
        U_A(a) = U_B(z(a)) \quad \text{for all} \quad a \in [\uul a,\ool a].
\end{equation}

\begin{lemma}\label{lemma:indutilityimplement}
Suppose Assumptions \ref{ass:1} and \ref{ass:2} hold. Then the pair \((z,U_A)\) is feasible if and only if:
\begin{enumerate}
\item[(i)] \(U_A'\) and \(U_A'/z'\) are non-decreasing and take values in \([0,1]\),
\item[(ii)] the boundary \(z\) has finite, strictly positive one-sided derivatives at every \(a\in(\uul a,\ool a)\), and a finite, strictly positive left derivative at \(\ool a\),
\item[(iii)] the supply constraint \eqref{eq:SS} holds:
\begin{equation}\label{eq:SS}
        \int_{\uul a}^{1} \int_0^{\hat{z}(a)} f(a,v)dv\  da \leq s_A, \quad \int_{\uul b}^{1} \int_0^{\hat{z}^{-1}(b)} f(v,b)dv \ db \leq s_B,\tag{S'}
\end{equation}
\item[(iv)] \(U_A'\) and \(U_A'/z'\) are piecewise continuously differentiable, and \(z\) is piecewise twice continuously differentiable.
\end{enumerate}
\end{lemma}
To understand these conditions, note that the \(A\)-indirect utility \(U_A\) and the boundary \(z\) pin down \(U_B\) through \eqref{eq:III}. In
particular, differentiating this condition gives
\begin{equation}\label{eq:DIII}
U_A'(a) = U_B'(z(a))\cdot z'(a)
\quad \Rightarrow \quad
U_B'(z(a)) = \frac{U_A'(a)}{z'(a)}.
\end{equation}

Since \(U_A\) and \(U_B\) must be convex, condition \((i)\) requires the
implied marginal utility profiles to be nondecreasing. Condition \((ii)\)
records the regularity of the boundary implied by \eqref{eq:DIII}. Condition
\((iii)\) rewrites the supply constraints in terms of the boundary: types below
\(z\) receive good \(A\), while types above \(z\) receive good \(B\). Condition \((iv)\) is the regularity imposed by Assumption \ref{ass:2}: the induced marginal utilities and boundary must be piecewise
smooth. Without this assumption, this condition would be omitted.

\paragraph{Welfare-maximizing \(U_A\) for a fixed boundary \(z\).}
I now fix a boundary \(z\) and find the \(A\)-indirect utility \(U_A\) that
maximizes total welfare \eqref{eq:Wprime} subject to \((z,U_A)\) being
feasible. This is the first step in the proof that needs Assumption
\ref{ass:2}, which rules out pathological boundaries and ensures that \(z\), \(U_A'\), and the implied \(B\)-marginal utility \(U_A'/z'\) are piecewise smooth.

\begin{lemma}\label{lemma:optimplement}
Fix a piecewise twice continuously differentiable boundary \(z\). Then
\((z,U_A)\), with \(U_A\) defined by \eqref{eq:optlower}, maximizes total
welfare \eqref{eq:Wprime} among all feasible pairs \((z,\check U_A)\).
\begin{equation}\label{eq:optlower}
U_A'(a)=
\begin{cases}
0, & a \in (0,\uul a),\\[6pt]
m(a)\cdot k, & a \in (\uul a, \ool a),\\[6pt]
1, & a \in (\ool a,1).
\end{cases}
\end{equation}
where
\[
m(a) = \exp \left(
\int_{\uul a}^{a}\max\!\left[0,\,\frac{z''(s)}{z'(s)}\right]\,ds
\right)
\prod_{\substack{z^{\prime +}(t)> z^{\prime -}(t), \\ t\leq a}} \  \  \frac{z^{\prime +}(t)}{z^{\prime -}(t)},
\quad \quad  k=\frac{1}{\max\left\{m(\ool a),\frac{m(\ool a)}{z^{\prime -}(\ool a)}\right\}}.
\]
\end{lemma}

The main idea is that, for a fixed boundary \(z\), the best implementation damages goods as little as possible while still satisfying incentive compatibility. Indeed, \eqref{eq:optlower} has a simple implication: \(U_A'\) is constant on intervals where \(z\) is concave, while \(U_A'\) is proportional to \(z'\) on intervals where \(z\) is convex. Recall that
\begin{equation}\label{eq:DIIIIII}
U_B'(z(a))= \frac{U_A'(a)}{z'(a)}.
\end{equation}
Thus, we can equivalently say that $U_A'(a)$ is constant on concave intervals of $z$ while $U_B'(z(a))$ is constant on its convex intervals (Figure \ref{fig:convexconcave}).
\begin{figure}[h!]
        \centering
        \hspace{-1.8cm} 
        \begin{tikzpicture}[scale=0.7]

          \draw[darkdarkgreen, very thick] (1.2,1.6) .. controls (1.6,2.5) and (2,2.8) .. (2.5,3)
                                      .. controls (3,3.2) and (3.4,3.4) .. (3.8,4.5)
                                      .. controls (3.88,4.7) .. (4,5);

          \fill[niceblue, opacity=0.3] (0,5) -- (0,1.6) -- (1.2,1.6)
                                  .. controls (1.6,2.5) and (2,2.8) .. (2.5,3)
                                  .. controls (3,3.2) and (3.4,3.4) .. (3.8,4.5)
                                  .. controls (3.9,4.8) .. (4,5) -- cycle;
      
          \fill[orange, opacity=0.3] (5,0) -- (1.2,0) -- (1.2,1.6)
          .. controls (1.6,2.5) and (2,2.8) .. (2.5,3)
          .. controls (3,3.2) and (3.4,3.4) .. (3.8,4.5)
          .. controls (3.9,4.8) .. (4,5) --(5,5) -- cycle;

          \draw[darkred, dashed] (1.2,0) -- (1.2,1.6);  
          \draw[darkred, dashed] (2.5,0) -- (2.5,3);

          \draw[darkred, thick, decorate, decoration={brace, amplitude=10pt}] (0,3) -- (0,5) node[midway, left=12pt] {$U_B'$ constant};

          \draw[darkdarkgreen, dashed] (0,1.6) -- (1.2,1.6); 
          \draw[darkred, dashed] (0,3) -- (2.5,3);  
          \draw[darkred, dashed] (0,5) -- (4,5);  
          \draw[darkred, thick, decorate, decoration={brace, amplitude=10pt, mirror}] (1.2,0) -- (2.5,0) node[midway, below=12pt] {$U_A'$ constant};
                                    \draw[->] (0,0) -- (5.1,0) node[right] {$a$};
                                    \draw[->] (0,0) -- (0,5.1) node[above] {$b$};
        \end{tikzpicture}
        \captionsetup{width=0.5\textwidth}
        \caption{\(U_A'(a)\) and \(U_B'\left(z(a)\right)\) are constant where the boundary \(z\) is concave and convex, respectively.}
        \label{fig:convexconcave}
\end{figure}
Now, by Lemma \ref{lemma:indutilityimplement}, $U_A'(a)$ and $U_B'\left(z(a)\right)$ must be increasing. The above observation therefore means that at least one of these monotonicity constraints must always bind. Intuitively, were neither constraint to bind on some interval, we could increase $U_A'(a)$ and $U_B'\left(z(a)\right)$ pointwise in a \eqref{eq:DIIIIII}-preserving manner until one of them started binding. This would create pointwise higher utility profiles and thus produce superior welfare \eqref{eq:Wprime}. 

\paragraph{Showing the welfare-maximizing boundary is linear.}

Having pinned down the best \(U_A\) for each fixed boundary \(z\), we can
optimize over the boundary itself. This is where the monotonicity condition on
the conditional inverse anti-hazard rates in \eqref{eq:anitihazards} enters.

\begin{lemma}\label{lemma:vararg}
Under the conditions of Theorem \ref{th:1}, the welfare-maximizing mechanism,
among mechanisms that allocate positive masses of both goods, features a linear
boundary $z^{*}:[\uul a^*,\ool a^*]\to[\uul b^*,\ool b^*].$
\end{lemma}

Intuitively, I show that under the anti-hazard-rate condition in Theorem
\ref{th:1}, no boundary with strictly convex or concave parts, or with kinks,
can satisfy the necessary optimality conditions. The optimal boundary must
therefore be linear. Consider some interval $[\uul a,\ool a]$ on which the boundary is concave and, for simplicity, assume it consists of multiple small, linear pieces (Figure \ref{fg:preperturb}). We will consider the perturbations to these linear pieces on this part of the boundary. Notice, however, that such perturbations have to respect the supply constraint \eqref{eq:SS}, and thus must preserve the probability mass below and above the boundary. Still, we can construct perturbations preserving the supply constraint by perturbing one piece of the boundary upwards and another one downwards in a ratio that leaves the probability masses unchanged (Figure \ref{fg:postperturb-c}). First-order optimality conditions then tell us that, when perturbing one such piece, we can capture the effect of this perturbation on the supply constraint by a Lagrange multiplier $\mu$.

\begin{figure}[h!]
  \centering
  \begin{subfigure}[b]{0.32\linewidth}
    \centering
\begin{tikzpicture}[scale=0.7]
  \draw[->] (0,0) -- (5.8,0) node[right] {$a$};
  \draw[->] (0,0) -- (0,5.5) node[above] {$b$};

  \draw[ultra thick, darkgreen]
      (0.5,0.9)
      -- (1.625,2.9775)   
      -- (2.85,4.2025)    
      -- (4.075,4.8150)   
      -- (5,5);           

  \draw[dashed] (1.625,2.9775) -- (1.625,0) node[below] {$\tilde a$};
  \draw[dashed] (2.85,4.2025)   -- (2.85,0)   node[below] {$\tilde a + \delta$};
  \draw[dashed] (4.075,4.8150)  -- (4.075,0);

  \draw[dashed] (0,2.9775) -- (1.625,2.9775);
  \node[left] at (0,2.9775) {$\tilde b$};
\end{tikzpicture}

    \caption{Figure \ref{fg:preperturb}}
    \label{fg:preperturb}
  \end{subfigure}\hfill
  \begin{subfigure}[b]{0.32\linewidth}
    \centering
\begin{tikzpicture}[scale=0.7]
  \draw[->] (0,0) -- (5.8,0) node[right] {$a$};
  \draw[->] (0,0) -- (0,5.1) node[above] {$b$};

    \draw[ultra thick, red]
      (0.5,0.9)
      -- (1.455693,2.6638)
    -- (3.32023,4.5452)   
      -- (3.75,4.7550)   
      -- (5,5);           

  \draw[ultra thick, darkgreen]
      (0.5,0.9)
      -- (1.625,2.9775)   
      -- (2.85,4.2025)    
      -- (4.075,4.8150)   
      -- (5,5);           

  \draw[dashed] (1.625,2.9775) -- (1.625,0) node[below] {$\tilde a$};
  \draw[dashed] (2.85,4.2025)   -- (2.85,0)   node[below] {$\tilde a + \delta$};
  \draw[dashed] (4.075,4.8150)  -- (4.075,0);

  \draw[dashed] (0,2.9775) -- (1.625,2.9775);
  \node[left] at (0,2.9775) {$\tilde b$};
\end{tikzpicture}

    \caption{Figure \ref{fg:postperturb-c}}
    \label{fg:postperturb-c}
  \end{subfigure}\hfill
  \begin{subfigure}[b]{0.32\linewidth}
    \centering
\begin{tikzpicture}[scale=0.7]
  \draw[->] (0,0) -- (5.8,0) node[right] {$a$};
  \draw[->] (0,0) -- (0,5.1) node[above] {$b$};

\fill[niceblue, opacity=0.3] 
    (0,2.36938)  -- (1.29693,2.36938) 
-- (1.625,2.9775)
 -- (2.85,4.2025) 
  -- (3.30023,4.42892)
  -- (0,4.42892)-- cycle;

  \fill[orange, opacity=0.5]
     (1.29693,2.36938) -- (3.30023,4.42892)
      -- (2.85,4.20250) -- (1.625,2.97750) -- cycle;

  \draw[red, ultra thick]
      (1.29693,2.36938) -- (3.30023,4.42892);

  \draw[ultra thick, darkgreen]
      (0.5,0.9)
      -- (1.625,2.97750)   
      -- (2.85,4.20250)    
      -- (4.075,4.81500)   
      -- (5,5);            

  \draw[dashed] (1.625,2.97750) -- (1.625,0) node[below] {$\tilde a$};
  \draw[dashed] (2.85,4.20250)   -- (2.85,0)   node[below] {$\tilde a + \delta$};
  \draw[dashed] (4.075,4.81500)  -- (4.075,0);
\end{tikzpicture}

    \caption{Figure \ref{fg:postperturb}}
    \label{fg:postperturb}
  \end{subfigure}
\end{figure}

Recall that by Lemma \ref{lem:objintbyparts}, total welfare is given by \eqref{eq:Wprime}:
\[
U_A(1) -\int_0^1 U_A'(a) \cdot F(a,\hat z(a))da.
\]
Moreover, Lemma \ref{lemma:optimplement}  tells us that $U_A'$ is equal to some constant on the region where the boundary is concave. Thus, we can write the effect of the boundary in the region $[\tilde a,\tilde a+\delta]$, up to scaling, as follows:
\[
-\int_{\tilde a}^{\tilde a+\delta}  F\left(a,\tilde b + (a-\tilde a)\cdot s\right)da.
\]
Now, consider the effect of a small downward perturbation to the height of the boundary on this interval, as illustrated in Figure \ref{fg:postperturb}. The first-order effect of this perturbation is given by:
\[
-
\frac{d}{d\tilde b}
\int_{\tilde a}^{\tilde a+\delta}
F\left(a,\tilde b+(a-\tilde a)\cdot s\right)\,da
+
\mu
\frac{d}{d\tilde b}
\int_{\tilde a}^{\tilde a+\delta}
\int_0^{\tilde b+(a-\tilde a)\cdot s}
f(a,b)\,db\,da.
\]
The latter term captures the effect of the perturbation on the probability mass under the boundary, which is valued according to the aforementioned multiplier $\mu$ on the supply constraint \eqref{eq:S}. Performing a small perturbation like this (in either direction) is not beneficial when
\[
0 \ \approx \ 
      \int_{\mathcal{K}} \ f(a,b)\,da\  db \;-\; \mu \,\int_{\mathcal{L}} \ f(a,b)\,da\  db
\]
where $\mathcal{K}$ and $\mathcal{L}$ are, respectively, the blue region in Figure \ref{fg:postperturb} and the orange region in Figure \ref{fg:postperturb}. When the length $\delta$ of the perturbed interval becomes small, this gives:
\begin{equation}\label{eq:perturbationintuition}
0 \ = \ \frac{F_{A|B}(\tilde a|z(\tilde a))}{f_{A|B}( \tilde a|z(\tilde a))} \ - \ \mu.
\end{equation}
Consequently, when the boundary is strictly concave on some region, a profitable perturbation analogous to that in Figure \ref{fg:postperturb-c} does not exist only if \eqref{eq:perturbationintuition} holds at every point in that interval. However, this cannot be the case under the condition in
Theorem \ref{th:1}, which guarantees that the relevant conditional inverse
anti-hazard rate is strictly increasing along the boundary. Thus, if we were indifferent about perturbing the boundary slightly downward at some level of \(a\), the same downward perturbation would be strictly profitable at any higher \(a\). The proof formalizes this reasoning in the continuous case using optimal control methods.

\paragraph{The welfare-maximizing linear boundary has unit slope.}

Lemma \ref{lemma:vararg} tells us the welfare-maximizing boundary is linear and Lemma \ref{lemma:optimplement} pins down the best implementation of any linear
boundary: in particular, if it has slope \(s\geq 1\), then good \(A\)
is undamaged and good \(B\) is delivered at quality \(1/s\). It remains to show
that the optimal slope is \(1\).

\begin{lemma}\label{lem:slope1}
Under the conditions of Theorem \ref{th:1}, the welfare-maximizing mechanism,
among mechanisms that allocate positive masses of both goods, features a linear
boundary with slope \(1\).
\end{lemma}

Suppose the boundary \(z\) has slope \(s>1\) and
reaches the ceiling of the unit square before reaching its right wall, as in
Figure \ref{fig:133}. By Lemma \ref{lemma:optimplement}, this boundary is
implemented by offering good \(A\) undamaged at toll \(\uul a_z\), and good
\(B\) at quality \(1/s\) and toll \(\uul b_z/s\).

\begin{figure}[h!]
        \centering
        \begin{subfigure}[b]{0.45\linewidth}
            \centering
\begin{tikzpicture}[scale=0.70]
    \draw[->, thick] (0,0) -- (5.1,0) node[right] {$a$};
    \draw[->, thick] (0,0) -- (0,5.2) node[above] {$b$};

    \draw[red, very thick] (1.2,1.6) -- (2.5,5);

    \draw[darkdarkgreen, very thick, dashed] (1.2,0) -- (1.2,1.6);
    \draw[darkdarkgreen, very thick, dashed] (0,1.6) -- (1.2,1.6);

    \fill[niceblue, opacity=0.3]
        (0,5) -- (0,1.6) -- (1.2,1.6) -- (2.5,5) -- cycle;

    \fill[orange, opacity=0.3]
        (5,0) -- (1.2,0) -- (1.2,1.6) -- (2.5,5) -- (5,5) -- cycle;

    \node[red, anchor=south west, rotate=69] at (1.7,2.9) {\fontsize{14}{14}\selectfont $z(a)$};

    \node at (1.2, -0.5) {$\uul{a}_z$};
    \node at (-0.35, 1.6) {$\uul{b}_z$};
\end{tikzpicture}
            \caption{Figure \ref{fig:133}: $s$-sloped boundary $z$.}
            \label{fig:133}
        \end{subfigure}
        \hspace{0.2cm}
        \begin{subfigure}[b]{0.45\linewidth}
        \centering

\begin{tikzpicture}[scale=0.70, >=latex]

    \draw[->, thick] (0,0) -- (5.1,0) node[right] {$a$};
    \draw[->, thick] (0,0) -- (0,5.2) node[above] {$b$};

    \coordinate (ra) at (1.2,1.6); 
    \coordinate (rb) at (2.5,5);   
    \coordinate (ba) at (0.9,2.5); 
    \coordinate (bb) at (4,5);   

    \path[name path=redline]   (ra) -- (rb);
    \path[name path=blueline]  (ba) -- (bb);
    \path[name intersections={of=redline and blueline, by=I}];

    \fill[teal,opacity=0.5]
      (ba) -- (ba|-0,0) -- (ra|-0,0) -- (ra) -- cycle;

    \fill[teal,opacity=0.5]
      (ba) -- (ra) -- (I) -- cycle;

    \fill[gray,opacity=0.5]
      (I) -- (rb) -- (bb) -- cycle;

    \draw[red, very thick] (ra) -- (rb)
      node[midway, anchor=south west, rotate=69] {\fontsize{14}{14}\selectfont $z(a)$};

    \draw[darkdarkgreen, very thick, dashed] (ra) -- (ra|-0,0);
    \draw[darkdarkgreen, very thick, dashed] (ra) -- (ra-|0,0);
    \node at (1.43,-0.5) {$\uul{a}_z$};
    \node at (-0.35,1.6) {$\uul{b}_z$};

    \draw[blue, very thick] (ba) -- (bb)
      node[midway, anchor=north west, rotate=24] {\fontsize{14}{14}\selectfont $r(a)$};

    \draw[darkdarkgreen, very thick, dashed] (ba) -- (ba|-0,0);
    \draw[darkdarkgreen, very thick, dashed] (ba) -- (ba-|0,0);
    \node at (0.85,-0.5) {${\uul{a}_r}$};
    \node at (-0.35,2.5) {${\uul{b}_r}$};

    \draw[gray, very thick, dashed] (I|-0,0) -- (I);
\node at (2.07,-0.36)  {$a^*$};

\end{tikzpicture}

        \caption{Figure \ref{fig:1444}: Regions $\ool {\mathcal{D}}$ (gray) and $\uul {\mathcal{D}}$ (green).}
        \label{fig:1444}
        \end{subfigure}
    \end{figure}

Now compare \(z\) to a slightly flatter boundary \(r\) that preserves the masses
allocated to both goods, as in Figure \ref{fig:1444}. The two boundaries cross
once, at \(a^*\), and the flatter boundary has a lower toll for good \(A\):
\(\uul a_r<\uul a_z\). Let
\[
\uul {\mathcal{D}}
=
\{(a,b)\colon 0<a<a^*,\ \hat z(a)<b<\hat r(a)\},
\qquad
\ool {\mathcal{D}}
=
\{(a,b)\colon a^*<a<1,\ \hat r(a)<b<\hat z(a)\}.
\]
These are the types who enter and leave the allocation of good \(A\),
respectively. Types in \(\uul{\mathcal D}\) switch to \(A\), either from
good \(B\) or from nonparticipation, while types in
\(\ool{\mathcal D}\) switch from \(A\) to \(B\). Since \(r\) preserves
the mass assigned to good \(A\),
\begin{equation}\label{eq:equalmasses}
    \int_{\ool{\mathcal D}} f(a,b)\,da\  db
    =
    \int_{\uul{\mathcal D}} f(a,b)\,da\  db.
\end{equation}
Let \(\Delta\) be the welfare difference between \(r\) and \(z\). Since good \(A\) is undamaged under both boundaries,
\(U_A'\equiv 1\), and so
\begin{align*}
\Delta 
&= (1-\uul a_r)
   - \int_{0}^1 F\left(a,\hat{r}(a)\right)\,da
   - \biggl(\,(1-\uul a_z)
     - \int_{0}^1 F\left(a,\hat z(a)\right)\,da\biggr) \\
&= (\uul a_z - \uul a_r)
   - \left(
    \int_{\uul {\mathcal{D}}}
   \frac{F_{A\mid B}(a \mid b)}{f_{A\mid B}(a \mid b)}\cdot
   f(a,b)\,da\  db
\;-\;
\int_{\ool {\mathcal{D}}}
   \frac{F_{A\mid B}(a \mid b)}{f_{A\mid B}(a \mid b)}\cdot
   f(a,b)\,da\  db
   \right).
\end{align*}
The first term is positive because \(\uul a_r<\uul a_z\). The bracketed term is negative by the anti-hazard-rate condition in Theorem
\ref{th:1}: the inverse anti-hazard rate is increasing in both arguments,
\(\ool{\mathcal D}\) lies to the north-east of \(\uul{\mathcal D}\), and the
two regions have equal probability mass by \eqref{eq:equalmasses}.

Since \(\Delta>0\), any boundary with slope \(s>1\) is dominated by a flatter boundary, and so the optimal linear boundary has slope \(1\), which in turn is optimally implemented without damages.

\paragraph{Optimality of the market-clearing toll mechanism.}
The preceding argument establishes that, under the purely welfarist objective
\((\gamma=0)\), the optimal mechanism uses no damages.\footnote{While the preceding argument considers
mechanisms that allocate positive masses of both kinds of goods, any optimal mechanism
must do so. If one kind of good were unallocated, offering it undamaged at a toll just
below \(1\) would attract a sufficiently small positive mass of agents and
strictly increase welfare while preserving feasibility.}
Recall, however, that Lemma \ref{lemma:bestundamaged} implies that the market-clearing toll mechanism maximizes welfare among such mechanisms. It therefore remains to extend this
conclusion to \(\gamma>0\).

To that end, we can use an argument analogous to the proof of Lemma \ref{lemma:bestundamaged} to show that the market-clearing toll mechanism also maximizes allocative efficiency \eqref{eq:allocceff} among \emph{all} feasible mechanisms. Finally, observe that for any \(\gamma\in[0,1]\), the designer's objective is a nonnegative linear
combination of welfare and allocative efficiency. Since the
market-clearing toll mechanism separately maximizes both components, it also
maximizes their combination and is therefore optimal for every
\(\gamma\in[0,1]\).

\section{Extension: heterogeneous toll costs}\label{sec:heterotollcost}\label{sec:ext}

The baseline model assumes that a given toll imposes the same cost on every agent. This assumption may be restrictive in some applications. In the context of social programs, for example, screening through monetary payments may impose a greater burden on poorer agents whom the program aims to target. To accommodate such statistical dependence, I now allow the cost of a toll to vary across agents. Suppose an agent's type is \((a,b,r)\), where \(r>0\) is her unit cost of the toll. If she receives good \(A\) or \(B\), her utility is, respectively,
\begin{align*}
x\cdot a-r\cdot c & \ \  \text{if}\ \  y=A,\\
x\cdot b-r\cdot c & \ \  \text{if}\ \  y=B.
\end{align*}
Assume that \(r\in[\uul r,\ool r]\), with \(\uul r>0\), and that types are distributed according to \(\tilde F\) with full support on
\([0,1]^2\times[\uul r,\ool r]\). Throughout this section, I assume that the designer cares only about agents' utility, so that \(\gamma=0\); extending the analysis to incorporate revenue motives is left for future work.

When \(\gamma=0\), the problem can be reformulated in terms of transformed values and welfare weights. Following \cite{DKA}, observe that a type \((a,b,r)\) is behaviorally equivalent to a type with unit toll cost and values
\[
        (\hat a,\hat b)
        =
        \left(\frac{a}{r},\frac{b}{r}\right).
\]
Let \(G\) denote the distribution of transformed values \((\hat a,\hat b)\) induced by \(\tilde F\). Its support is contained in
\([0,1/\uul r]^2\), and it has full support on this square. The transformation also changes the designer's objective. Since the utility of the original type \((a,b,r)\) is equal to \(r\) times the utility of the transformed type \((a/r,b/r,1)\), transformed types receive welfare weights
\[
        \lambda(\hat a,\hat b)
        =
        \mathbb{E}\left[
        r \ \Big| \ \frac{a}{r}=\hat a,\ 
        \frac{b}{r}=\hat b
        \right].
\]
Thus, after renormalizing toll costs to one, the designer's objective becomes
\begin{equation}\label{eq:objweight}
        \int
        \lambda(\hat a,\hat b)
        u_{\hat a,\hat b}(\hat a,\hat b)
        \,dG(\hat a,\hat b).
\end{equation}
Theorem \ref{th:1} then has the following direct analogue.

\begin{corollary}\label{corr:th1}
Suppose the designer is restricted to quality rules
\(x:[0,1/\uul r]^2\to[0,1]\) that are piecewise continuously differentiable
in each dimension, and suppose that total supply does not exceed the mass of
agents: \(s_A+s_B\leq1\). Let \(g\) denote the density of \(G\), suppose that \(g\) is strictly positive on the interior of the type space, and that \(g\) and \(\lambda g\) are Lipschitz continuous. Moreover, assume that the weighted inverse anti-hazard ratios
\[
        \frac{\int_0^{\hat a} \lambda(\hat v,\hat b)g(\hat v,\hat b)\,d\hat v}
        {g(\hat a,\hat b)},
        \qquad
        \frac{\int_0^{\hat b} \lambda(\hat a,\hat v)g(\hat a,\hat v)\,d\hat v}
        {g(\hat a,\hat b)}
\]
are increasing in both arguments, and that for each ratio at least one of the two monotonicities is strict. Then the market-clearing toll mechanism is optimal.
\end{corollary}

The proof is the same after replacing the density \(f\) in the welfare terms by the weighted density \(\lambda g\), while keeping the supply constraints expressed in terms of the unweighted density \(g\). Consequently, the same argument applies with the weighted inverse anti-hazard ratios above in place of the corresponding ratios from the baseline model.

Theorem \ref{thm:2} can be modified analogously: the supply-preserving terms are still computed using the unweighted transformed density \(g\), while utility changes are computed using the weighted density \(\lambda g\). The following is the normalized analogue of Corollary \ref{cor:damagesnormalized}:

\begin{corollary}\label{corr:damagesnormalized-weighted}
Let \(\bar v:=1/\uul r\). Suppose that \(g\) is strictly positive on the interior of the type space, that \(g\) and \(\lambda g\) are Lipschitz continuous, that the market-clearing toll for good \(B\) satisfies \(0<c_B^*<\bar v\), and that \(s_A+s_B=1\). Then damages are optimal if, for some \(\tilde b\in(c_B^*,\bar v)\),
\begin{equation}\label{eq:covariancecorr-weighted}
        \operatorname{Cov}
        \left(
        \frac{
        \int_0^{\hat b-c_B^*}
        \lambda(\hat v,\hat b)g(\hat v,\hat b)\,d\hat v
        }
        {g(\hat b-c_B^*,\hat b)},
        (\tilde b-\hat b)_+
        \ \middle|\ \hat a=\hat b-c_B^*
        \right)
        >0.
\end{equation}
\end{corollary}
The conditional covariance is taken with respect to the probability measure on
\([c_B^*,\bar v]\) with density proportional to
\(g(\hat b-c_B^*,\hat b)\).

\paragraph{How heterogeneous toll costs affect the results.}
To understand the intuition for these modified results, note first that the transformed values $\hat a=a/r$ and $\hat b=b/r$ measure the \emph{amounts of the toll} agents are willing to incur to obtain the goods, rather than their raw values for them. The welfare weight $\lambda(\hat a,\hat b)$ then measures the expected cost of one unit of the toll among agents with these transformed values. Consequently, $\lambda(\hat a,\hat b)\hat a$ and $\lambda(\hat a,\hat b)\hat b$ are their expected raw values for goods $A$ and $B$, respectively.

Recall that, in the baseline model, the key objects are the ratios of
inframarginal beneficiaries to marginal switchers. In the transformed model,
these ratios become
\[
\frac{
\int_0^{\hat b-c_B^*}
\lambda(\hat v,\hat b)g(\hat v,\hat b)\,d\hat v
}{
g(\hat b-c_B^*,\hat b)
}
\qquad\text{and}\qquad
\frac{
\int_0^{\hat a+c_B^*}
\lambda(\hat a,\hat v)g(\hat a,\hat v)\,d\hat v
}{
g(\hat a,\hat a+c_B^*)
}.
\]
Their numerators measure the welfare-weighted masses of existing recipients who would benefit from introducing a discounted damaged option for a good, while their denominators measure the densities of agents at the boundary $z_0$ who would switch into this good if such an option were introduced. Analogously to the intuition in Section \ref{sec:2d}, damages are more attractive when these ratios are higher at lower values of $\hat b$ and $\hat a$, respectively.\footnote{Inframarginal beneficiaries are now counted using the weighted density $\lambda g$, whereas marginal switchers are still counted using the unweighted density $g$. This is intuitive: the former term measures the \emph{welfare gains} received by agents, while the latter measures the \emph{number} of agents who switch goods and hence the resulting change in demand.}

The effect of heterogeneous toll costs can then be understood through their influence on these ratios. In particular, the covariance in \eqref{eq:covariancecorr-weighted} is more likely to be positive when welfare weights $\lambda(\hat a,\hat b)$ are relatively high at low values of $\hat a$ and $\hat b$. Equivalently, damages become more attractive when agents with low willingness to incur tolls for goods $A$ and $B$ tend to have high toll costs $r$. In that case, their low values of $\hat a=a/r$ and $\hat b=b/r$ reflect the high cost of the toll rather than low underlying values for the goods. Reducing their toll burden through the discount attached to a damaged option therefore generates large welfare gains. Conversely, when agents with low values of $\hat a$ and $\hat b$ also have low toll costs, their low willingness to incur the toll is more likely to reflect a genuinely low value for the good. The welfare gains from offering them a toll discount are then smaller, making damages less attractive.

This matters for applications in which agents' values and toll costs are correlated. Consider first an administrative toll. Agents with the greatest need for the good may also face language barriers, limited digital access, or less familiarity with bureaucratic procedures, making the application process especially costly for them. They may nevertheless complete it because they value the good highly. By contrast, agents with less need may find the process easier but still choose not to apply because the benefit is not valuable enough to justify even a modest burden. If differences in values dominate differences in administrative costs, agents with low transformed values $\hat a$ and $\hat b$ are then likely to have both low raw values and low toll costs. Introducing a discounted damaged option for them therefore generates relatively small welfare gains. Thus, in this case, heterogeneity works in favor of screening with tolls. Conversely, suppose the toll is a monetary application fee and the good, such as subsidized housing, is especially valuable to low-income households, who also find payments especially burdensome. If the latter effect is stronger, such households may have low transformed values $\hat a$ and $\hat b$ despite having high raw values. Low willingness to pay then reflects the high cost of money rather than a low value for the good, thereby strengthening the case for damages.

\section{Discussion}\label{sec:disc}

The main contribution of this paper is to distinguish between two kinds of
screening instruments: tolls, whose costs are separable from agents' values for
the allocated good, and damages, which are more burdensome for agents who value
the good more. I study when a welfare-maximizing designer should use each class
of instruments. When the allocated good is homogeneous, my results unambiguously predict that
tolls are superior to damages. Consider, for instance, rationed access to
specialist appointments or elective procedures. If access is allocated through
a waitlist, the burden falls especially heavily on patients whose need for
treatment is greatest. When such waitlists become long, introducing or raising
a copayment is likely to be welfare-improving: in equilibrium, it screens out
the lowest-need visits and reduces wait times for the sickest patients.
Similarly, in food assistance programs such as SNAP, usage restrictions are
disproportionately burdensome for households that rely on the program most. My
model suggests they
should be replaced by more stringent application or recertification
requirements.

In settings where the allocated goods are heterogeneous, however, the prediction
is more nuanced and depends on the statistical relationship between agents'
values for the different goods. When these are strongly negatively related, the
optimal mechanism is likely to use only tolls. Consider, for instance, the
allocation of DMV appointments across offices dispersed throughout a city.
Applicants may place high value on an appointment near where they live, while
placing low value on appointments at more distant offices. In this case, the
designer should use office-specific toll instruments---such as location-specific
booking fees, application burdens, or walk-in queues---to clear excess demand
and sort applicants across offices. Long waitlists at the most popular offices
are likely to be suboptimal. When values for different goods are strongly positively related, the conclusion
is likely to change. This can be the case for certain affordable housing programs, where the neediest households value all units highly, while those with better outside options do not. When such
correlation is strong, market-clearing tolls are more likely to absorb much of
the surplus generated by the allocation, and damages may improve sorting across
goods. Thus, in some housing settings, differentiated waitlists may be part of
an optimal mechanism, possibly alongside project-specific toll-like
instruments, such as higher rent contributions for overdemanded projects or
rent subsidies for less-demanded ones. Characterizing the exact optimal design
of such mechanisms is an avenue for further research.

More broadly, the paper makes a methodological contribution by developing a
tractable approach to two-dimensional screening problems under the restriction
to deterministic mechanisms. The two-good case captures qualitatively important
economic forces, such as screening agents into and out of different options,
while remaining analytically simple. The approach could therefore be useful in
other screening problems where deterministic mechanisms are economically
plausible. For example, it may be applied to two-good monopoly pricing with
unit-demand consumers, or to optimal tax problems in which a couple chooses
whether to file jointly or separately and faces different labor-supply
incentives under the two filing regimes.

Future work could generalize my results by relaxing some of the restrictions imposed in the
analysis. While the restriction to deterministic mechanisms reflects practical concerns in many settings of interest, my model could be extended to accommodate some forms of randomization. Suppose, for example, that agents who pay a toll can enter a lottery that gives them one of the goods with some probability, and that unsuccessful agents can re-enter by paying the same toll again. When represented as a direct-revelation mechanism, a menu of such pure lotteries with re-entry would be equivalent to a deterministic
mechanism of the kind studied here (see Appendix \ref{appsec:waitlists} for a
related discussion). In such cases, my restriction would only rule out ``mixed lotteries'' that allocate either good with some probability.

Some aspects of the analysis could also be extended beyond two goods. The
two-good case is the simplest and most tractable environment in which screening
instruments serve not only to exclude low-value types, but also to sort agents
between available options. Nevertheless, the main forces identified in the
paper also apply in more general settings. In particular, the approach underlying
Theorem \ref{thm:2} generalizes: starting from the \(N\)-good market-clearing
toll mechanism, the designer can introduce a slightly damaged option for one of
the goods and ask whether the resulting inframarginal gains dominate the toll
correction needed to restore market clearing. As in the two-good model, this comparison depends on whether the ratio of inframarginal beneficiaries to marginal switchers is most favorable where the
gains from taking the damaged option are highest. However, with many goods, marginal
switchers no longer lie on a single boundary, but instead are on ``indifference
surfaces'' between the regions picking the damaged good and each alternative good. Nevertheless, the
affiliation intuition still applies: when values are positively
related, low-value regions for the damaged good tend to contain more
inframarginal consumers relative to marginal switchers, creating the same force
favoring damages as in the two-good case.

\appendix
\setstretch{1}

\section{Justifying the direct-revelation formulation}\label{sec:dirrev}

In the main model, I restrict the designer to deterministic mechanisms, meaning that each menu option allocates good \(A\), good \(B\), or nothing with certainty. In this environment, the standard revelation principle does not immediately apply, because there may exist equilibria that are more favorable for the designer in which agents randomize over menu options. I now show that this is not the case: for any feasible randomized selection among agents' favorite menu options, there exists a deterministic selection that allocates the same aggregate quantities of both goods, collects the same toll revenue, and leaves every agent's utility unchanged. Then, under such a selection rule, the usual revelation argument applies.

First, note that we can restrict attention to menus contained in
\([0,1]^2\), because any option charging a toll \(c>1\) is dominated by
non-participation. We can also restrict attention to closed, and hence compact,
menus. Indeed, replacing a menu \(M_y\subseteq[0,1]^2\) by its closure
\(\overline M_y\) does not change any type's maximal attainable utility, since
utility is continuous in \((x,c)\). Moreover, any option selected from the
original menu remains utility-maximizing after taking the closure, so the same
selection rule, allocations, tolls, and welfare can be preserved. Let us
therefore fix two compact menus
\(M_A,M_B\subseteq[0,1]\times[0,1]\), and denote the outside option by
\(\varnothing\). I use \(\Omega\) to denote the set of options available to
agents:
\[
\Omega
:=
\{\varnothing\}
\ \cup\
\bigl(\{A\}\times M_A\bigr)
\ \cup\
\bigl(\{B\}\times M_B\bigr).
\]
I write the utility of type \((a,b)\) from option \(\omega\in \Omega\) as $u_{a,b}(\omega)$ and use $\kappa(\omega)$ to denote the toll paid under that option:
\[
u_{a,b}(\omega):=
\begin{cases}
ax-c & \text{if }\omega=(A,x,c),\\
bx-c & \text{if }\omega=(B,x,c),\\
0 & \text{if }\omega=\varnothing.
\end{cases}
,\qquad
\kappa(\omega):=
\begin{cases}
c & \text{if }\omega=(A,x,c)\text{ or }\omega=(B,x,c),\\
0 & \text{if }\omega=\varnothing.
\end{cases}
\]

Given the menus \(M_A\) and \(M_B\), a \emph{selection rule} is a measurable
stochastic kernel \(\sigma:[0,1]^2\to\Delta(\Omega)\), where
\(\sigma_{a,b}(E)\) is the probability that type \((a,b)\) selects an option in the
Borel set \(E\subseteq\Omega\). An extended mechanism consists of the menus
together with such a selection rule. 

The extended mechanism satisfies the supply constraints if
\[
\int \sigma_{a,b}\bigl(\{A\}\times M_A\bigr)\,dF(a,b)\ \le\ s_A,
\qquad
\int \sigma_{a,b}\bigl(\{B\}\times M_B\bigr)\,dF(a,b)\ \le\ s_B.
\]
For each type \((a,b)\), let
\[
G(a,b):=\argmax_{\omega\in\Omega}u_{a,b}(\omega).
\]
We say \(\sigma\) is \emph{agent-optimal} if $\sigma_{a,b}(G(a,b))=1$ for every $(a,b).$ Note \(G\) is nonempty, compact-valued, and measurable because \(\Omega\) is compact and \(u_{a,b}(\omega)\) is continuous in \(\omega\).

The following result shows that it is without loss to use a deterministic selection rule.

\begin{proposition}\label{prop:purify_goods}
Fix compact menus \(M_A,M_B\subseteq[0,1]\times[0,1]\), and let
\(\sigma\) be an agent-optimal selection rule satisfying the supply constraints. Then there exists a deterministic selection rule \(\tilde\sigma\), of the form $\tilde\sigma_{a,b}=\delta_{\tilde m(a,b)}$ for some measurable map \(\tilde m:[0,1]^2\to\Omega\), such that
\[
\int \tilde\sigma_{a,b}\bigl(\{A\}\times M_A\bigr)\,dF(a,b)
=
\int \sigma_{a,b}\bigl(\{A\}\times M_A\bigr)\,dF(a,b),
\]
\[
\int \tilde\sigma_{a,b}\bigl(\{B\}\times M_B\bigr)\,dF(a,b)
=
\int \sigma_{a,b}\bigl(\{B\}\times M_B\bigr)\,dF(a,b),
\]
\[
\int \kappa(\tilde m(a,b))\,dF(a,b)
=
\int \int_\Omega \kappa(\omega)\,d\sigma_{a,b}(\omega)\,dF(a,b),
\]
and
\[
u_{a,b}\bigl(\tilde m(a,b)\bigr)
=
\int_{\Omega} u_{a,b}(\omega)\,d\sigma_{a,b}(\omega)
\quad \text{for every }(a,b).
\]
Consequently, the deterministic selection rule has the same aggregate use of each good, the same aggregate toll, and the same value of the designer's objective as the randomized selection rule.
\end{proposition}

\begin{proof}Define the finite-dimensional payoff vector
\[
h(\omega)
:=
\left(
\mathbb{1}_{\omega\in \{A\}\times M_A},
\mathbb{1}_{\omega\in \{B\}\times M_B},
\kappa(\omega)
\right).
\]
Since the menus are bounded, \(h\) is bounded and measurable. Moreover, \(F\) is
atomless and \(\sigma_{a,b}(G(a,b))=1\) for every \((a,b)\), so a standard
purification theorem applies (see, for example,
\citet{feinberg2006dvoretzky}). Hence, there exists a measurable selection
\(\tilde m(a,b)\in G(a,b)\), up to an \(F\)-null set, such that
\begin{equation}\label{eq:invector}
\int h(\tilde m(a,b))\,dF(a,b)
=
\int \int_\Omega h(\omega)\,d\sigma_{a,b}(\omega)\,dF(a,b).
\end{equation}
Changing \(\tilde m\) on the null set if necessary, using any measurable
selection from \(G\), we may take \(\tilde m(a,b)\in G(a,b)\) for every
\((a,b)\).
Define \(\tilde\sigma_{a,b}:=\delta_{\tilde m(a,b)}\). The first two coordinates of \eqref{eq:invector} give
\[
\int \tilde\sigma_{a,b}\bigl(\{A\}\times M_A\bigr)\,dF(a,b)
=
\int \sigma_{a,b}\bigl(\{A\}\times M_A\bigr)\,dF(a,b),
\]
\[
\int \tilde\sigma_{a,b}\bigl(\{B\}\times M_B\bigr)\,dF(a,b)
=
\int \sigma_{a,b}\bigl(\{B\}\times M_B\bigr)\,dF(a,b).
\]
Hence the deterministic selection satisfies the same supply constraints as \(\sigma\). The third coordinate gives
\[
\int \kappa(\tilde m(a,b))\,dF(a,b)
=
\int \int_\Omega \kappa(\omega)\,d\sigma_{a,b}(\omega)\,dF(a,b).
\]

It remains only to note that every type's utility is unchanged. Recall
\(\sigma_{a,b}\) is supported on \(G(a,b)\) and \(\tilde m(a,b)\in G(a,b)\), so both selections give type \((a,b)\) her maximal utility:
\[
u_{a,b}\bigl(\tilde m(a,b)\bigr)
=
\int_\Omega u_{a,b}(\omega)\,d\sigma_{a,b}(\omega)
=
\max_{\omega\in\Omega}u_{a,b}(\omega).
\]
Since the deterministic selection preserves each type's utility and the aggregate toll, it also preserves the value of the objective.
\end{proof}

\section{Waitlists}\label{appsec:waitlists}

I now show that the baseline model can be interpreted as a reduced-form
description of a steady-state waitlist environment with free re-entry.

\paragraph{Waitlist model.}
Consider a stationary waitlist environment in continuous time. There are two
goods, \(A\) and \(B\). Agents have values \((a,b)\in[0,1]^2\), where \(a\) is
the value of receiving good \(A\) and \(b\) is the value of receiving good \(B\).
At each instant \(\tau\in\mathbb R\), flow masses \(s_A,s_B>0\) of goods \(A\)
and \(B\) arrive, with \(s_A+s_B\leq 1\). At the same time, a unit flow mass of
agents arrives, with types distributed according to \(F\). Agents discount utility at rate \(\rho>0\). Receiving good \(A\) gives utility
\(a\), receiving good \(B\) gives utility \(b\), and receiving no good gives
utility zero. Tolls enter utility additively.

\paragraph{Mechanisms.}
There is a separate waitlist for each good. For each waitlist, the designer
chooses a menu of triples \((c,t,p)\), where \(c\in\mathbb R_+\) is a toll,
\(t\in\mathbb R_+\) is a wait time, and \(p\in[0,1]\) is the probability of
receiving the good at the end of the wait. Options with \(p=0\) are identified
with the outside option. An agent who chooses \((c,t,p)\)
incurs the toll immediately, waits \(t\) units of time, and then receives the
good with probability \(p\). If she does not receive the good, she may re-enter
the mechanism.

I restrict attention to stationary mechanisms that admit a steady state. Thus,
the same menus are offered at every date, and all agents of the same type choose
the same good and the same menu option, independently of their arrival date and
of how many times they have already participated. Steady-state feasibility then requires
\begin{equation}\label{eq:SSS}
\int \mathbb{1}_{(a,b)\text{ chooses }A}\,dF(a,b)\leq s_A,
\qquad
\int \mathbb{1}_{(a,b)\text{ chooses }B}\,dF(a,b)\leq s_B .
\end{equation}
These constraints take this form because agents can freely re-enter after an
unsuccessful attempt. Hence, any agent who chooses a waitlist option with \(p>0\)
eventually receives the corresponding good almost surely. Per-attempt success
probabilities affect the timing of receipt, but not the eventual mass of goods
consumed.

\paragraph{Reduction to the baseline model.}
Consider a type-\((a,b)\) agent who chooses option \((c,t,p)\) on the waitlist
for good \(A\), with \(p>0\). Let \(V_A(a;c,t,p)\) be her expected discounted
utility at the moment she chooses the option. After an unsuccessful attempt, she
faces the same continuation problem, shifted forward by \(t\). Hence
\[
V_A(a;c,t,p)
=
-c
+
p e^{-\rho t}a
+
(1-p)e^{-\rho t}V_A(a;c,t,p).
\]
Solving gives
\begin{equation}\label{eq:waitlist_value}
V_A(a;c,t,p)
=
\frac{-c+p e^{-\rho t}a}{1-(1-p)e^{-\rho t}}
=
x a-\tilde c,
\end{equation}
where
\[
x:=
\frac{p e^{-\rho t}}{1-(1-p)e^{-\rho t}}\in[0,1],
\qquad
\tilde c:=
\frac{c}{1-(1-p)e^{-\rho t}}\in\mathbb R_+ .
\]
The case of good \(B\) is analogous. Thus, each waitlist option has a reduced-form representation \((x,\tilde c)\).
Conversely, any \((x,\tilde c)\in(0,1]\times\mathbb R_+\) can be implemented by
a waitlist option: take \(p=1\), \(t=-(1/\rho)\log x\), and \(c=\tilde c\). The
case \(x=0\) corresponds to the outside option.

Using this reduced form and the revelation principle established in Section \ref{sec:dirrev}, the designer's problem reduces to choosing an allocation rule
\[
(\tilde c,x,y):[0,1]^2\to \mathbb R_+\times[0,1]\times\{A,B,\varnothing\},
\]
where \(y(a,b)\) is the waitlist chosen by type \((a,b)\), \(x(a,b)\) is the
discounted allocation intensity, and \(\tilde c(a,b)\) is the reduced-form
toll cost. The utility of type \((a,b)\) from reporting \((a',b')\) is
\[
u_{a,b}(a',b')
:=
\begin{cases}
a\,x(a',b')-\tilde c(a',b') & \text{if }y(a',b')=A,\\
b\,x(a',b')-\tilde c(a',b') & \text{if }y(a',b')=B,\\
0 & \text{if }y(a',b')=\varnothing .
\end{cases}
\]
The direct mechanism must satisfy incentive compatibility, individual rationality, and the steady-state resource constraints
\[
\int \mathbb{1}_{y(a,b)=A}\,dF(a,b)\leq s_A,
\qquad
\int \mathbb{1}_{y(a,b)=B}\,dF(a,b)\leq s_B .
\]
This matches the form of the baseline problem, provided that the designer discounts each toll payment back to the agent's initial entry date at the same rate \(\rho\) at which the agent discounts utility.

\section{Omitted proofs}

First, recall that \(U_A\) and \(U_B\) are convex and strictly increasing on
\([\uul a,\ool a]\) and \([\uul b,\ool b]\), respectively, with
\(U_A(\uul a)=U_B(\uul b)=0\). Hence each is differentiable except at countably
many points, and is pinned down by its derivative wherever the derivative exists.
At nondifferentiability points, the one-sided derivatives are obtained as limits
of derivatives at nearby differentiability points. For example, if \(D_A\) is the
set of points at which \(U_A'\) exists, then
\[
        U_A^{\prime -}(t')
        =
        \lim_{t\uparrow t',\, t\in D_A} U_A'(t),
        \qquad
        U_A^{\prime +}(t')
        =
        \lim_{t\downarrow t',\, t\in D_A} U_A'(t).
\]
Analogous statements hold for \(U_B\).

\subsection{Proof of Proposition \ref{prop:1Dnodamage}}
By single crossing, any feasible mechanism is characterized by a cutoff
\(\uul a\) such that all types above the cutoff receive good \(A\):
\[
        \uul a
        :=
        \inf\{a\in[0,1]: y(a,b)=A\},
        \qquad
        \inf\emptyset:=1.
\]
Fix such a cutoff. I show that any mechanism implementing it using damages can be replaced with a toll mechanism without lowering any agent's utility or the designer's objective.

First, incentive compatibility implies that all types below the cutoff receive
the same indirect utility \(\uul U\). By the envelope theorem, the indirect utility of type \(a\) can then be written as
\[
        U(a,b)
        =
        \uul U
        +
        \int_{\uul a}^{\max\{\uul a,a\}} x(v,b)\,dv .
\]
Since \(x(v,b)\leq 1\), we have
\begin{equation}\label{eq:boundimplementtt}
        U(a,b)
        \leq
        \uul U+\max\{\uul a,a\}-\uul a .
\end{equation}
This upper bound is implementable without damages by offering good $B$ with a toll of $b-\uul U$ and good $A$ with a toll of $\uul a-\uul U$. These tolls are nonnegative: \(\uul U\leq b\), because the
outside option \(B\) has value \(b\), and \(\uul U\leq \uul a\), because no
\(A\)-option can give utility above \(a\) to a type with value \(a\). Note that this menu implements the same cutoff and attains the utility bound in \eqref{eq:boundimplementtt}.

It remains to check the designer's objective. We can write type $(a,b)$'s contribution to \eqref{eq:obj} as:
\[
        u_{a,b}+\gamma c(a,b)
        =
        u_{a,b}+\gamma(xv-u_{a,b})
        =
        (1-\gamma)u_{a,b}+\gamma xv,
\]
where $v$ denotes her value for the good she receives. Note the above replacement weakly raises \(u_{a,b}\) and sets \(x=1\). Since \(\gamma\in[0,1]\), this contribution weakly increases for all types.
\subsection{Proof of Proposition \ref{prop:boundary}}

Agents choose among good \(A\), good \(B\), and non-participation by comparing
\(U_A(a)\), \(U_B(b)\), and \(0\). By the definitions of \(\uul a\) and
\(\uul b\), types with \(a<\uul a\) get no positive utility from good \(A\), and
types with \(b<\uul b\) get no positive utility from good \(B\). Hence every type
\((a,b)<(\uul a,\uul b)\) chooses non-participation. If \(a>\uul a\) and \(b<\uul b\), then good \(A\) gives positive utility while
good \(B\) does not, so the type chooses good \(A\). Symmetrically, if
\(a<\uul a\) and \(b>\uul b\), the type chooses good \(B\). Since positive masses
of both goods are allocated, \(\uul a,\uul b<1\). It remains to describe choices on the upper rectangle. If
\(U_B(1)\geq U_A(1)\), set $\ool a=1$ and $\ool b=U_B^{-1}(U_A(1));$ otherwise, set $\ool a=U_A^{-1}(U_B(1))$ and $\ool b=1.$ On \([\uul a,\ool a]\), define \(z=U_B^{-1}\circ U_A\). Then \(z\) is
continuous and strictly increasing, with \(z(\uul a)=\uul b\) and
\(z(\ool a)=\ool b\). For \(a\leq\ool a\), types below the boundary choose
good \(A\), while types above it choose good \(B\). If \(\ool a<1\), types
with \(a>\ool a\) choose good \(A\), since $U_A(a)>U_A(\ool a)=U_B(1)\geq U_B(b).$ Types on the boundary are indifferent.

\subsection{Proof of Lemma \ref{lemma:bestundamaged}}

We first show that the market-clearing toll
mechanism maximizes welfare among all undamaged mechanisms. By \eqref{eq:IC}, we may restrict attention to mechanisms that offer each good at a single toll, \(c_A\) and \(c_B\) (possibly equal to $1$). Consider such a mechanism where the supply constraint for good \(A\) is slack. It must then have \(c_A\in (0,1]\); otherwise, almost every agent would receive some good, contradicting slackness together with
\(s_A+s_B\leq1\). By full support and continuity of demand, \(c_A\) could then
be lowered so as to increase the utility of a positive mass of agents while
leaving demand for \(A\) below its supply. Demand for \(B\) could only fall as a result, so
the modified mechanism would remain feasible, contradicting optimality. The
same argument applies to good \(B\). Thus, both supply constraints must bind at the optimal mechanism, and so the optimum must be the market-clearing toll mechanism.

Next, I show that the market-clearing toll mechanism maximizes allocative efficiency. Let
\[
q_A^*(a,b)
=
\mathbb{1}\{a-c_A^*>\max\{0,b-c_B^*\}\},
\qquad
q_B^*(a,b)
=
\mathbb{1}\{b-c_B^*>\max\{0,a-c_A^*\}\},
\]
up to a measure-zero set of indifferent types. Consider any feasible undamaged
mechanism, and let \(q_A,q_B\) denote its allocation indicators. Pointwise,
\[
q_A^*(a,b)(a-c_A^*)+q_B^*(a,b)(b-c_B^*)
\geq
q_A(a,b)(a-c_A^*)+q_B(a,b)(b-c_B^*),
\]
because \((q_A^*,q_B^*)\) selects a maximizer among
\(0\), \(a-c_A^*\), and \(b-c_B^*\). Integrating and using market clearing gives
\[
\int \bigl(q_A^*a+q_B^*b\bigr)\,dF
-c_A^*s_A-c_B^*s_B\geq
\int \bigl(q_Aa+q_Bb\bigr)\,dF
-c_A^*\int q_A\,dF-c_B^*\int q_B\,dF.
\]
Since \(c_A^*,c_B^*\geq0\) and the mechanism is feasible,
\[
\int \bigl(q_A^*a+q_B^*b\bigr)\,dF
\geq
\int \bigl(q_Aa+q_Bb\bigr)\,dF.
\]

Finally, for any undamaged mechanism,
\[
\int (u+\gamma c)\,dF
=
(1-\gamma)\int u\,dF
+
\gamma\int \bigl(q_Aa+q_Bb\bigr)\,dF.
\]
The market-clearing toll mechanism maximizes both terms on the right-hand side
and is therefore optimal among undamaged mechanisms for every
\(\gamma\in[0,1]\).

\subsection{Proof of Theorem \ref{thm:2}}

By Lemma \ref{lemma:bestundamaged}, the market-clearing toll mechanism is
optimal for every \(\gamma\in[0,1]\) among mechanisms that do not use damages. It therefore suffices to construct a feasible mechanism using damages that strictly improves on it. We do so through a local
perturbation.

\paragraph{Step 1: The perturbing menu.} Consider the market-clearing toll menu consisting of the
\(A\)-option \((1,c_A^*)\) and the \(B\)-option \((1,c_B^*)\). Fix \(\tilde b\in(c_B^*,1)\). For small \(\varepsilon>0\), take the alternative menu
\[
        \left\{
        (A,1,c_A^*+\varepsilon\alpha),\,
        (B,1,c_B^*+\varepsilon\beta),\,
        (B,1-\varepsilon,c_B^*+\varepsilon\beta-\varepsilon\tilde b)
        \right\}.
\]
Let \(U_A\) and \(U_B\) denote the \(A\)- and \(B\)-indirect utilities induced by
this perturbed menu. Then
\[
        U_A(a)
        =
        (a-c_A^*-\varepsilon\alpha)_+,
\quad \text{and} \quad
        U_B(b)
        =
        \max\left\{
        0,\,
        b-c_B^*-\varepsilon\beta,\,
        (1-\varepsilon)b-c_B^*-\varepsilon\beta+\varepsilon\tilde b
        \right\}.
\]

For any fixed \(b>c_B^*\), an agent with \(B\)-value \(b\) continues to participate for all sufficiently
small \(\varepsilon\). Hence, away from the lower participation margin, we only
need to compare the two \(B\)-options:
\[
\begin{aligned}
        \max\left\{
        b-c_B^*-\varepsilon\beta,\,
        (1-\varepsilon)b-c_B^*-\varepsilon\beta+\varepsilon\tilde b
        \right\}
        &=
        b-c_B^*
        +
        \varepsilon\left[(\tilde b-b)_+-\beta\right].
\end{aligned}
\]
Thus, the damaged option raises the \(B\)-side payoff by
\(\varepsilon(\tilde b-b)\) for types with \(b<\tilde b\), and has no direct
effect on the \(B\)-side payoff of types with \(b\geq \tilde b\), apart from the
common toll increase \(\varepsilon\beta\).

\paragraph{Step 2: Mass change from the interior \(A\)-\(B\) boundary.} We now compute the first-order change in the masses assigned to the two goods. Define
\[
        D_\varepsilon(a,b)
        :=
        U_A(a)-U_B(b),
\]
and fix a compact interval \(K\subset(c_A^*,\bar a^*)\). For \(a\in K\) and \(b\) in a
small neighborhood of \(z_0(a)\), all agents strictly prefer either good to nothing. Hence
\[
        D_\varepsilon(a,b)
        =
        z_0(a)-b
        +
        \varepsilon
        \left[
        \beta-\alpha-\bigl(\tilde b-b\bigr)_+
        \right].
\]
Let \(z_\varepsilon(a)\) denote the perturbed \(A\)-\(B\) boundary over \(K\); then $D_\varepsilon(a,z_\varepsilon(a))=0.$ Since $\beta-\alpha-(\tilde b-b)_+$ is bounded uniformly on \(K\), this equation
implies $ z_\varepsilon(a)-z_0(a)=O(\varepsilon),$ uniformly in \(a\in K\). Since \(b\mapsto(\tilde b-b)_+\) is Lipschitz, $\bigl(\tilde b-z_\varepsilon(a)\bigr)_+
        =
        \bigl(\tilde b-z_0(a)\bigr)_+
        +
        O(\varepsilon),$ uniformly in \(a\in K\). Substituting this into
\(D_\varepsilon(a,z_\varepsilon(a))=0\) gives
\begin{equation}\label{eq:expansionofz}
        z_\varepsilon(a)
        =
        z_0(a)
        +
        \varepsilon
        \left[
        \beta-\alpha-\bigl(\tilde b-z_0(a)\bigr)_+
        \right]
        +
        o(\varepsilon),
\end{equation}
uniformly in \(a\in K\).

Now, fix \(\eta>0\) and set $K_\eta:=[c_A^*+\eta,\bar a^*-\eta].$ Define
\[
        \Delta M_A^{AB}(\varepsilon;K_\eta)
        :=
        \int_{K_\eta}
        \int_{z_0(a)}^{z_\varepsilon(a)}
        f(a,b)\,db\,da .
\]
This is the signed change in the mass assigned to good \(A\) coming from the
movement of the \(A\)-\(B\) boundary over \(K_\eta\). Using
\eqref{eq:expansionofz} and the continuity of \(f\), for each fixed \(\eta>0\),
\[
        \Delta M_A^{AB}(\varepsilon;K_\eta)
        =
        \varepsilon
        \int_{K_\eta}
        \left[
        \beta-\alpha-\bigl(\tilde b-z_0(a)\bigr)_+
        \right]
        f(a,z_0(a))\,da
        +
        o_\eta(\varepsilon).
\]

We now pass from \(K_\eta\) to the full old interior boundary. Let \(E_\eta:=[c_A^*,\bar a^*]\setminus K_\eta\). This set has length at most \(2\eta\). In
addition, all relevant cutoffs and the \(A\)-\(B\) boundary move by an amount of order \(\varepsilon\), uniformly for small \(\varepsilon\). Hence the region
swept out over \(E_\eta\) has area at most \(C_0\varepsilon\eta\), for some
constant \(C_0\) independent of \(\eta\). Since \(f\) is bounded, there is a
constant \(C_1\), also independent of \(\eta\), such that
\[
        \left|
        \Delta M_A^{AB}(\varepsilon;E_\eta)
        \right|
        \leq
        C_1\varepsilon\eta .
\]
Therefore, for every fixed \(\eta>0\),
\[
\begin{aligned}
        \limsup_{\varepsilon\downarrow0}
        \left|
        \frac{\Delta M_A^{AB}(\varepsilon;[c_A^*,\bar a^*])}{\varepsilon}
        -
        \int_{K_\eta}
        \left[
        \beta-\alpha-\bigl(\tilde b-z_0(a)\bigr)_+
        \right]
        f(a,z_0(a))\,da
        \right|
        \leq C_1\eta .
\end{aligned}
\]
Now, let \(\eta\downarrow0\). Since the integrand is bounded and continuous on
\([c_A^*,\bar a^*]\), the integral over \(K_\eta\) converges to the integral over
\([c_A^*,\bar a^*]\). Hence
\[
        \Delta M_A^{AB}(\varepsilon;[c_A^*,\bar a^*])
        =
        \varepsilon
        \int_{c_A^*}^{\bar a^*}
        \left[
        \beta-\alpha-\bigl(\tilde b-z_0(a)\bigr)_+
        \right]
        f(a,z_0(a))\,da
        +o(\varepsilon).
\]
Since good \(B\) lies on the other side of the same boundary, the corresponding contribution to the mass assigned to good
\(B\) is given by $\Delta M_B^{AB} = - \Delta M_A^{AB}$.

\paragraph{Step 3: The participation margins.} It remains to compute the participation-margin contributions. Along the lower \(A\)-participation margin, the cutoff changes from \(c_A^*\) to
\(c_A^*+\varepsilon\alpha\). Therefore the signed change in the mass assigned to
good \(A\) along this margin is
\[
        -\int_0^{c_B^*}
        \int_{c_A^*}^{c_A^*+\varepsilon\alpha}
        f(a,b)\,da\,db
        =
        -\varepsilon\alpha
        \int_0^{c_B^*} f(c_A^*,b)\,db
        +o(\varepsilon).
\]
By the definition of \(P_A^{c^*_A}\), this equals $-\varepsilon\alpha P_A^{c^*_A}+o(\varepsilon).$ The possible overlap with the lower \(B\)-participation margin is confined to an
\(O(\varepsilon)\times O(\varepsilon)\) neighborhood of \((c_A^*,c_B^*)\), and hence
has mass \(o(\varepsilon)\).

Along the lower \(B\)-participation margin, the undamaged \(B\)-option has cutoff
\(c_B^*+\varepsilon\beta\), while the damaged \(B\)-option has cutoff solving $ (1-\varepsilon)b-c_B^*-\varepsilon\beta+\varepsilon\tilde b=0.$ Thus, the lower cutoff for participation in the \(B\)-menu is
\[
        b_\varepsilon
        =
        c_B^*+\varepsilon(\beta+c_B^*-\tilde b)+o(\varepsilon),
\]
because \(\tilde b>c_B^*\) makes the damaged option the lower-threshold
\(B\)-option. Hence the signed change in the mass assigned to good \(B\) along
this margin is
\[
        -\int_0^{c_A^*}
        \int_{c_B^*}^{b_\varepsilon}
        f(a,b)\,db\,da
        =
        -\varepsilon(\beta+c_B^*-\tilde b)
        \int_0^{c_A^*}f(a,c_B^*)\,da
        +o(\varepsilon).
\]
By the definition of \(P_B^{c^*_B}\), this equals $\varepsilon P_B^{c^*_B}(\tilde b-c_B^*-\beta)+o(\varepsilon).$ Again, the possible overlap with the lower \(A\)-participation margin is
confined to an \(O(\varepsilon)\times O(\varepsilon)\) neighborhood of
\((c_A^*,c_B^*)\), and hence has mass \(o(\varepsilon)\).

\paragraph{Step 4: Preserving supplies to first order.} We now choose $\alpha,\beta$ so that this first-order perturbation leaves the total allocated amounts of both goods unchanged. Combining the lower participation-margin contributions with the interior
\(A\)-\(B\) boundary contribution gives
\[
        \Delta M_A
        =
        \varepsilon
        \left\{
        -\alpha P_A^{c^*_A}
        +
        \int_{c_A^*}^{\bar a^*}
        \left[\beta-\alpha-\bigl(\tilde b-z_0(a)\bigr)_+\right] f(a,z_0(a))\,da
        \right\}
        +o(\varepsilon).
\]
\[
        \Delta M_B
        =
        \varepsilon
        \left\{
        P_B^{c^*_B}(\tilde b-c_B^*-\beta)
        -
        \int_{c_A^*}^{\bar a^*}
        \left[\beta-\alpha-\bigl(\tilde b-z_0(a)\bigr)_+\right] f(a,z_0(a))\,da
        \right\}
        +o(\varepsilon).
\]
Setting these first-order changes equal to zero gives
\begin{equation}\label{eq:A-supply-preserve-damages}
        -\alpha P_A^{c^*_A}
        +
        \int_{c_A^*}^{\bar a^*}
        \left[
        \beta-\alpha-\bigl(\tilde b-z_0(a)\bigr)_+
        \right]
        f(a,z_0(a))\,da
        =0.
\end{equation}
\begin{equation}\label{eq:B-supply-preserve-damages}
        P_B^{c^*_B}(\tilde b-c_B^*-\beta)
        +
        \int_{c_A^*}^{\bar a^*}
        \left[
        \bigl(\tilde b-z_0(a)\bigr)_+
        +\alpha-\beta
        \right]
        f(a,z_0(a))\,da
        =0.
\end{equation}
Using the definitions of \(P_{AB}\) and \(Q\), equations
\eqref{eq:A-supply-preserve-damages} and
\eqref{eq:B-supply-preserve-damages} can be written as
\[
        (P_A^{c^*_A}+P_{AB})\alpha-P_{AB}\beta
        =
        -Q,
\quad \text{and} \quad
        -P_{AB}\alpha+(P_B^{c^*_B}+P_{AB})\beta
        =
        P_B^{c^*_B}(\tilde b-c_B^*)+Q.
\]
The coefficient matrix is
\begin{equation}\label{eq:matrixx}
        \begin{pmatrix}
        P_A^{c^*_A}+P_{AB} & -P_{AB}\\
        -P_{AB} & P_B^{c^*_B}+P_{AB}
        \end{pmatrix};
\end{equation}
it has the determinant
\begin{equation}\label{eq:determinattt}
        (P_A^{c^*_A}+P_{AB})(P_B^{c^*_B}+P_{AB})-P_{AB}^2
        =
        P_A^{c^*_A}P_B^{c^*_B}+P_A^{c^*_A}P_{AB}+P_B^{c^*_B}P_{AB} >0,
\end{equation}
which is positive because \(P_A^{c^*_A}>0\) and \(P_{AB}>0\). The latter follows from \(c_A^*<1\), \(c_B^*<1\), and the fact that the old interior boundary has positive length. The supply-preserving first-order corrections are thus unique and satisfy \eqref{eq:alphabetadef}.

Moreover, these corrections are nonnegative. Indeed, for \(a\geq c_A^*\), we have $\bigl(\tilde b-z_0(a)\bigr)_+
        \leq \tilde b-c_B^*,$ and hence \(Q\leq P_{AB}(\tilde b-c_B^*)\). Therefore
\(\alpha \geq0\). The formula for \(\beta \) immediately
gives \(\beta \geq0\). Thus the old \(A\)- and \(B\)-option tolls
remain nonnegative. Since \(c_B^*>0\), the damaged \(B\)-option also has a
nonnegative toll for all sufficiently small \(\varepsilon>0\).

\paragraph{Step 5: The first-order objective gain.}
We now compute the first-order effect on the objective.
The movements of the choice boundaries do not create additional first-order
utility terms, because types on the old \(A\)-\(B\) boundary are indifferent
between the two goods, and types on the old participation margins receive zero
utility. Hence the first-order utility change is obtained by integrating the
direct utility changes over the old assignment regions.

Old \(A\)-choosers receive the first-order utility change \(-\alpha\). Since
the undamaged menu clears the supply of good \(A\), their utility contribution is $-\alpha s_A.$
Old \(B\)-choosers with value \(b\) receive the first-order utility change $ (\tilde b-b)_+-\beta.$
Since \(P_B^b\) is the density of old \(B\)-choosers with value \(b\), their
utility contribution is
\[
        \int_{c_B^*}^{1}
        \left[(\tilde b-b)_+-\beta\right]P_B^b\,db.
\]
Therefore the first-order change in the utility component is
\[
        \varepsilon
        \left\{
        \int_{c_B^*}^{\tilde b}
        (\tilde b-b)P_B^b\,db
        -
        \alpha s_A
        -
        \beta s_B
        \right\}
        +o(\varepsilon).
\]
Now consider the toll component. Since the perturbation preserves the masses
assigned to both goods to first order, the zeroth-order toll terms
\(c_A^*s_A+c_B^*s_B\) do not change to first order. The first-order change in total
tolls is
\[
        \varepsilon
        \left\{
        \alpha s_A+\beta s_B
        -
        \tilde b
        \int_{c_B^*}^{\tilde b}P_B^b\,db
        \right\}
        +o(\varepsilon).
\]
Indeed, old \(A\)-choosers face the toll increase \(\varepsilon\alpha\), old
\(B\)-choosers face the common toll increase \(\varepsilon\beta\), and old
\(B\)-choosers with \(b<\tilde b\) select the damaged option, whose toll is
lower by \(\varepsilon\tilde b\).

Combining the utility and toll components gives
\[
        \varepsilon
        \left\{
        \int_{c_B^*}^{\tilde b}
        \bigl((1-\gamma)\tilde b-b\bigr)P_B^b\,db
        -
        (1-\gamma)(\alpha s_A+\beta s_B)
        \right\}
        +o(\varepsilon).
\]
Substituting from \eqref{eq:alphabetadef}, condition
\eqref{eq:local-damage-condition-main} says that the coefficient of
\(\varepsilon\) is strictly positive. Hence the objective strictly increases for
all sufficiently small \(\varepsilon>0\).

\paragraph{Step 6: Exact feasibility.}
The previous steps choose \(\alpha \) and \(\beta \) so that
the supply constraints are preserved to first order. I now perturb these two
numbers slightly so that the constraints hold exactly. For \(i\in\{A,B\}\), let \(M_i(\varepsilon,\alpha,\beta)\) denote the mass
assigned to good \(i\) by the perturbed menu, and let \(M_i^0\) denote the mass
assigned to good \(i\) by the original menu. For \(\varepsilon>0\), define the
normalized supply errors
\[
        H_i(\varepsilon,\alpha,\beta)
        :=
        \frac{
        M_i(\varepsilon,\alpha,\beta)-M_i^0
        }{\varepsilon}.
\]
The expansions above are uniform in a small neighborhood
of \((\alpha,\beta)\). Therefore \(H_A\) and \(H_B\) extend continuously to
\(\varepsilon=0\), with
\[
H_A(0,\alpha,\beta)
        =
        -\alpha P_A^{c^*_A}
        +
        \int_{c_A^*}^{\bar a^*}
        \left[
        \beta-\alpha-\bigl(\tilde b-z_0(a)\bigr)_+
        \right]
        f(a,z_0(a))\,da,
\]
\[
H_B(0,\alpha,\beta)
        =
        P_B^{c^*_B}(\tilde b-c_B^*-\beta)
        -
        \int_{c_A^*}^{\bar a^*}
        \left[
        \beta-\alpha-\bigl(\tilde b-z_0(a)\bigr)_+
        \right]
        f(a,z_0(a))\,da .
\]
Since \(f\) is Lipschitz, these extensions are
continuously differentiable in \((\alpha,\beta)\). By construction, $H_A(0,\alpha ,\beta )
        =
        H_B(0,\alpha ,\beta )
        =
        0 .$ Moreover, the derivative of \((H_A,H_B)\) with respect to
\((\alpha,\beta)\) at \((0,\alpha,\beta)\) is the negative of
\eqref{eq:matrixx}, and hence is invertible by \eqref{eq:determinattt}. 

If \(c_A^*>0\), the implicit function theorem implies that, for all
sufficiently small \(\varepsilon>0\), there exist
\[
        \alpha_\varepsilon=\alpha+o(1),
        \qquad
        \beta_\varepsilon=\beta+o(1),
\]
such that $H_A(\varepsilon,\alpha_\varepsilon,\beta_\varepsilon)
        =
        H_B(\varepsilon,\alpha_\varepsilon,\beta_\varepsilon)
        =
        0.$ Since \(c_A^*>0\), all three tolls remain nonnegative for sufficiently small
\(\varepsilon\).

Suppose instead that \(c_A^*=0\). Then \(P_B^{c_B^*}=0\), and hence
\(\alpha=0\). Set \(\alpha_\varepsilon=0\), so that the \(A\)-option remains
free. Almost every type then chooses one of the two goods, and therefore
\[
        M_A(\varepsilon,0,\beta)+M_B(\varepsilon,0,\beta)=1
        =
        M_A^0+M_B^0.
\]
Moreover,
\[
        \frac{\partial H_B}{\partial\beta}(0,0,\beta)
        =
        -P_{AB}
        <0.
\]
The one-dimensional implicit function theorem therefore gives $\beta_\varepsilon=\beta+o(1)$ such that $M_B(\varepsilon,0,\beta_\varepsilon)=M_B^0.$ The preceding identity then also implies $M_A(\varepsilon,0,\beta_\varepsilon)=M_A^0.$

Thus, in either case, the adjusted perturbed menu preserves both supplies
exactly and has nonnegative tolls for all sufficiently small
\(\varepsilon>0\). Since the exact-feasibility adjustments converge to the
first-order corrections, they do not change the first-order coefficient of the
objective gain computed in Step 5.

By condition
\eqref{eq:local-damage-condition-main}, that coefficient is strictly positive.
Therefore, for all sufficiently small \(\varepsilon>0\), the adjusted perturbed
menu is feasible and strictly improves on the market-clearing toll mechanism.

\subsection{Proof of Corollary \ref{cor:damagesnormalized}}

Since \(s_A+s_B=1\), the market-clearing toll mechanism must allocate some good to almost
every type. If both tolls were strictly positive, all types with \(a<c_A^*\) and \(b<c_B^*\) would choose neither good, contradicting market clearing. Since \(c_B^*>0\), it
follows that \(c_A^*=0\).

By Theorem \ref{thm:2} and the assumption that \(\gamma=0\), it is enough to show
\[
        \int_{c_B^*}^{\tilde b}
        (\tilde b-b)
        \left[
        \int_0^{b-c_B^*} f(t,b)\,dt
        \right]db
        >
        \alpha s_A+\beta s_B.
\]
Because \(c_A^*=0\),
\[
        P_B^{c^*_B}=\int_0^{c_A^*}f(a,c_B^*)\,da=0,
\]
and hence \(\alpha =0\). Moreover,
\[
        P_{AB}
        =
        \int_0^{1-c_B^*} f(a,a+c_B^*)\,da
\]
and
\[
        Q
        =
        \int_0^{1-c_B^*}
        (\tilde b-a-c_B^*)_+
        f(a,a+c_B^*)\,da
        =
        \int_0^{\tilde b-c_B^*}
        (\tilde b-c_B^*-a)f(a,a+c_B^*)\,da.
\]
Thus
\[
        \beta 
        =
        \frac{
        \int_0^{\tilde b-c_B^*}
        (\tilde b-c_B^*-a)f(a,a+c_B^*)\,da
        }{
        \int_0^{1-c_B^*} f(a,a+c_B^*)\,da
        }.
\]
Changing variables \(b=a+c_B^*\), the left-hand side of the damage condition becomes
\[
\begin{aligned}
        \int_{c_B^*}^{\tilde b}
        (\tilde b-b)
        \left[
        \int_0^{b-c_B^*} f(t,b)\,dt
        \right]db
        &=
        \int_0^{\tilde b-c_B^*}
        (\tilde b-c_B^*-a)
        \left[
        \int_0^a f(t,a+c_B^*)\,dt
        \right]da .
\end{aligned}
\]
Also, since the undamaged menu clears supplies and \(s_A+s_B=1\),
\[
        s_B
        =
        \int_{c_B^*}^1
        \left[
        \int_0^{b-c_B^*} f(t,b)\,dt
        \right]db
        =
        \int_0^{1-c_B^*}
        \left[
        \int_0^a f(t,a+c_B^*)\,dt
        \right]da .
\]
Thus the sufficient condition from Theorem \ref{thm:2} becomes
\[
        \int_0^{\tilde b-c_B^*}
        (\tilde b-c_B^*-a)
        \left[
        \int_0^a f(t,a+c_B^*)\,dt
        \right]da
        >
        \frac{
        \int_0^{\tilde b-c_B^*}
        (\tilde b-c_B^*-a)f(a,a+c_B^*)\,da
        }{
        \int_0^{1-c_B^*} f(a,a+c_B^*)\,da
        }
        \int_0^{1-c_B^*}
        \left[
        \int_0^a f(t,a+c_B^*)\,dt
        \right]da .
\]
Equivalently,
\[
\frac{
\int_0^{\tilde b-c_B^*}
(\tilde b-c_B^*-a)
\left[
\int_0^a f(t,a+c_B^*)\,dt
\right]da
}{
\int_0^{\tilde b-c_B^*}
(\tilde b-c_B^*-a)f(a,a+c_B^*)\,da
}
>
\frac{
\int_0^{1-c_B^*}
\left[
\int_0^a f(t,a+c_B^*)\,dt
\right]da
}{
\int_0^{1-c_B^*}
f(a,a+c_B^*)\,da
},
\]
which is equivalent to \eqref{eq:covariancecorr}.

\subsection{Proof of Lemma \ref{lem:objintbyparts}}

By Proposition \ref{prop:boundary}, welfare can be written as
\[
\int_{\uul a}^{\ool a}\int_0^{z(a)} U_A(a)f(a,b)\,db\,da
+
\int_{\uul b}^{1}\int_0^{\hat z^{-1}(b)} U_B(b)f(a,b)\,da\,db
+
\int_{\ool a}^{1}\int_0^1 U_A(a)f(a,b)\,db\,da .
\]
Using the boundary relation \(U_B(z(a))=U_A(a)\), the middle term can be
rewritten by the change of variables \(b=z(a)\). Hence the utility part equals
\[
        \int_{\uul a}^{\ool a} U_A(a)\,dF(a,z(a))
        +
        \int_{\ool a}^{1} U_A(a)\,dF(a,1).
\]
Equivalently, using the extended boundary \(\hat z\), this is
\[
        \int_0^1 U_A(a)\,dF(a,\hat z(a)).
\]
Integrating by parts gives
\[
        U_A(1)F(1,\hat z(1))
        -
        U_A(0)F(0,\hat z(0))
        -
        \int_0^1 U_A'(a)F(a,\hat z(a))\,da.
\]
Since \(F(1,\hat z(1))=F(1,1)=1\) and \(F(0,\hat z(0))=0\), this becomes
\[
        U_A(1)
        -
        \int_0^1 U_A'(a)F(a,\hat z(a))\,da.
\]

\subsection{Proof of Lemma \ref{lemma:indutilityimplement}}

We first prove necessity. Fix a feasible mechanism implementing \((z,U_A)\), and
let \(U_B\) be the induced \(B\)-indirect utility. Since \(U_A\) and \(U_B\) are
convex, their derivatives are non-decreasing wherever they exist. By Proposition
\ref{prop:boundary}, $U_A(a)=U_B(z(a))$ for all $a\in[\uul a,\ool a].$ Differentiating at points where the relevant derivatives exist gives
\[
        U_A'(a)=U_B'(z(a))z'(a),
        \qquad\text{so}\qquad
        \frac{U_A'(a)}{z'(a)}=U_B'(z(a)).
\]
Since \(z\) is increasing and \(U_B'\) is non-decreasing, \(U_A'/z'\) is
non-decreasing. Moreover, the envelope theorem and \(x\in[0,1]\) imply that
\(U_A'\) and \(U_B'=U_A'/z'\) take values in \((0,1]\) on their active
intervals. This proves \((i)\).

Next, Proposition \ref{prop:boundary} also gives $z=U_B^{-1}\circ U_A$ on the active interval. Hence the one-sided derivatives of \(z\) are determined
by the corresponding one-sided derivatives of \(U_A\) and \(U_B\), i.e.,
\[
        z^{\prime +}(a)
        =
        \frac{U_A^{\prime +}(a)}
        {(U_B^{\prime +}\circ z)(a)},
\]
and the analogous formula holds for left derivatives. Since the relevant
one-sided derivatives of \(U_A\) and \(U_B\) are finite and strictly positive on
the active intervals, \(z\) has finite, strictly positive one-sided derivatives
at every \(a\in(\uul a,\ool a)\), and a finite, strictly positive left derivative
at \(\ool a\). This proves \((ii)\).

The supply constraints follow from Proposition \ref{prop:boundary}. Up to null
sets, the agents receiving good \(A\) are those below the extended boundary, and
the agents receiving good \(B\) are those above it. Therefore
\[
        \int \mathbb{1}_{y(a,b)=A}\,dF(a,b)
        =
        \int_{\uul a}^{1}\int_0^{\hat z(a)} f(a,v)\,dv\,da, \quad
        \int \mathbb{1}_{y(a,b)=B}\,dF(a,b)
        =
        \int_{\uul b}^{1}\int_0^{\hat z^{-1}(b)} f(v,b)\,dv\,db.
\]
Feasibility of the original mechanism therefore implies \((iii)\).

It remains to prove the regularity condition. Under Assumption \ref{ass:2}, the
quality rule is piecewise continuously differentiable. By Proposition
\ref{prop:boundary}, types \((a,0)\) with \(a>\uul a\) choose good \(A\), and
types \((0,b)\) with \(b>\uul b\) choose good \(B\). The envelope theorem then
implies that, wherever the relevant derivatives exist, $U_A'(a)=x(a,0)$ and $U_B'(b)=x(0,b)$.
Thus \(U_A'\) and \(U_B'\) are piecewise continuously differentiable. Since
\(z=U_B^{-1}\circ U_A\), the boundary \(z\) is piecewise continuously
differentiable. On each interval on which the relevant functions are
differentiable,
\[
        z'(a)=\frac{U_A'(a)}{U_B'(z(a))}.
\]
The right-hand side is piecewise continuously differentiable, so \(z\) is
piecewise twice continuously differentiable. Finally,
\[
        \frac{U_A'(a)}{z'(a)}=U_B'(z(a)),
\]
so \(U_A'/z'\) is piecewise continuously differentiable. This proves \((iv)\).

We now prove sufficiency. Fix \((z,U_A)\) satisfying \((i)\)-\((iv)\). Define
\(U_B\) on the active interval by
\[
        U_B(b):=U_A(z^{-1}(b))
        \quad\text{for } b\in[\uul b,\ool b].
\]
Set $U_B(b)=0$ for $b\leq \uul b.$ For \(b\geq \ool b\), extend \(U_B\) linearly with slope
\[
        q:=\lim_{a\uparrow \ool a}\frac{U_A'(a)}{z'(a)},
        \quad \text{that is,} \quad
        U_B(b)=U_B(\ool b)+q(b-\ool b)
        \quad\text{for } b\geq \ool b.
\]
By \((i)\), this extension is convex and non-decreasing. Moreover, on the active
interval,
\[
        U_B'(z(a))=\frac{U_A'(a)}{z'(a)}
\]
wherever the derivative exists, and hence $U_B(z(a))=U_A(a)$ for all $ a\in[\uul a,\ool a]$. We now construct a direct mechanism. If neither good gives positive indirect
utility, assign non-participation and set \(x(a,b)=c(a,b)=0\). Otherwise assign
the good that gives the larger indirect utility:
\[
        y(a,b)=
        \begin{cases}
        A, & \text{if } U_A(a)\geq U_B(b) \text{ and } U_A(a)>0,\\
        B, & \text{if } U_B(b)>U_A(a),\\
        \varnothing, & \text{otherwise.}
        \end{cases}
\]
For types assigned goods \(A\) and $B$, respectively, set
\[
        x(a,b)=U_A^{\prime-}(a),
        \qquad
        c(a,b)=aU_A^{\prime-}(a)-U_A(a).
\]
\[
        x(a,b)=U_B^{\prime-}(b),
        \qquad
        c(a,b)=bU_B^{\prime-}(b)-U_B(b).
\]
For types assigned \(\varnothing\), keep \(x(a,b)=c(a,b)=0\). By \((i)\), these qualities lie in \([0,1]\). Convexity and
\(U_A(\uul a)=U_B(\uul b)=0\) also imply that the tolls are nonnegative.

It remains to verify incentive compatibility. Consider the \(A\)-options. Since
\(U_A\) is convex, \(U_A^{\prime-}(a)\) is a subgradient of \(U_A\) at \(a\). Thus
for every \(a'\),
\[
        U_A(a')
        \geq
        U_A(a)+U_A^{\prime-}(a)(a'-a).
\]
Equivalently, type \(a'\) gets weakly higher utility from the \(A\)-option
designed for \(a'\) than from the \(A\)-option designed for \(a\). Therefore no
type wants to misreport among \(A\)-options. The same argument applies to
\(B\)-options.

The utility from the assigned \(A\)-option is $        aU_A^{\prime-}(a)
        -
        \bigl(aU_A^{\prime-}(a)-U_A(a)\bigr)
        =
        U_A(a),$ and the utility from the assigned \(B\)-option is \(U_B(b)\). Since the mechanism
assigns each type the good that gives the larger of \(U_A(a)\) and \(U_B(b)\), no
type wants to switch goods. Individual rationality follows because both
indirect utilities are nonnegative and the outside option is available.

By construction $U_A(a)=U_B(z(a))$ for all $a\in[\uul a,\ool a],$ so the mechanism implements the boundary \(z\). Finally, the supply constraints
hold by \((iii)\), because Proposition \ref{prop:boundary} identifies the masses
assigned to goods \(A\) and \(B\) with the two integrals in \eqref{eq:SS}. Hence
the constructed mechanism is feasible and implements \((z,U_A)\).

\subsection{Proof of Lemma \ref{lemma:optimplement} }

Fix the boundary \(z\). By Lemma \ref{lemma:indutilityimplement}, feasibility of
\((z,\check U_A)\) is equivalent to \(\check U_A'\) and
\(\check U_A'/z'\) being non-decreasing, together with the bounds $0\leq \check U_A'\leq 1,$ $0\leq {\check U_A'}/{z'}\leq 1,$ and the regularity requirements stated there. Since the boundary is fixed, the
supply constraints do not depend on \(\check U_A\).

Let $p(a):=m(a)k$ for $a\in(\uul a,\ool a).$ We first show that the \(U_A\) defined in \eqref{eq:optlower} is feasible. The
function \(p\) is non-decreasing by construction. Moreover, on every interval on
which \(z\) is twice continuously differentiable,
\[
\bigl(\log(p/z')\bigr)'
=
\max\left\{0,\frac{z''}{z'}\right\}-\frac{z''}{z'}
\geq 0.
\]
At an upward jump of \(z'\), the product in the definition of \(m\) makes
\(p/z'\) continuous. At a downward jump of \(z'\), the function \(p\) is
continuous while \(z'\) falls, so \(p/z'\) jumps upward. Hence both \(p\) and
\(p/z'\) are non-decreasing. Finally,
\[
p^-(\ool a)
=
m(\ool a)k
=
\frac{1}{\max\{1,1/z^{\prime-}(\ool a)\}}
=
\min\{1,z^{\prime-}(\ool a)\}.
\]
Since \(p\) and \(p/z'\) are non-decreasing, this implies \(p\leq 1\) and
\(p/z'\leq 1\) on \((\uul a,\ool a)\). Thus \(U_A\) satisfies the feasibility
conditions of Lemma \ref{lemma:indutilityimplement}.

Now take any other feasible pair \((z,\check U_A)\), and write
\(\check p=\check U_A'\). We show that \(\check p\leq p\) almost everywhere.
On every smooth interval on which \(\check p>0\),
\[
(\log \check p)' \geq 0
\qquad\text{and}\qquad
(\log(\check p/z'))'\geq 0;
\quad \text{thus}\quad 
(\log \check p)'
\geq
\max\left\{0,\frac{z''}{z'}\right\}
=
(\log p)'.
\]
Therefore \(\log(p/\check p)\) is non-increasing on each such interval.

It remains only to check jumps. If \(z'\) jumps upward at \(t\), then
\[
\frac{p^+(t)}{p^-(t)}
=
\frac{z^{\prime+}(t)}{z^{\prime-}(t)}.
\]
Since \(\check p/z'\) is non-decreasing,
\[
\frac{\check p^+(t)}{\check p^-(t)}
\geq
\frac{z^{\prime+}(t)}{z^{\prime-}(t)}.
\]
Hence \((p/\check p)^+(t)\leq (p/\check p)^-(t)\). If \(z'\) jumps downward, then
\(p\) is continuous and \(\check p\) is non-decreasing, so again
\((p/\check p)^+(t)\leq (p/\check p)^-(t)\). The same conclusion is immediate at
any jump of \(\check p\) at which \(z'\) is continuous. Hence \(p/\check p\) is
non-increasing on \((\uul a,\ool a)\).

By feasibility, $\check p^-(\ool a)\leq 1$ and $\frac{\check p^-(\ool a)}{z^{\prime-}(\ool a)}\leq 1,$ so
\[
\check p^-(\ool a)\leq \min\{1,z^{\prime-}(\ool a)\}=p^-(\ool a).
\]
Since \(p/\check p\) is non-increasing, it follows that $\check p(a)\leq p(a)$ for a.e. $a\in(\uul a,\ool a).$ Also, both marginal utilities are zero on \((0,\uul a)\), while on
\((\ool a,1)\) feasibility gives \(\check U_A'\leq 1=U_A'\). Thus $\check U_A'(a)\leq U_A'(a)$ for a.e. $a\in(0,1).$

Finally, since \(U_A(a)=0\) below \(\uul a\), we can write \eqref{eq:Wprime} as
\[
\int_0^1 \bigl(1-F(a,\hat z(a))\bigr)U_A'(a)\,da.
\]
For the fixed boundary \(z\), the weight \(1-F(a,\hat z(a))\) is non-negative
and independent of \(U_A\). Since the constructed \(U_A'\) pointwise dominates
the derivative of every other feasible \(\check U_A\), it maximizes
\eqref{eq:Wprime} among all feasible pairs \((z,\check U_A)\).

\subsection{Proof of Lemma \ref{lemma:vararg}}

First, since \(z^*\) and its inverse are strictly increasing, the
anti-hazard-rate condition in Theorem \ref{th:1} implies that
\begin{equation}\label{eq:strict-antihazard-along-boundary}
\frac{F_{A|B}(a|z^*(a))}
     {f_{A|B}(a|z^*(a))}
\quad \text{is strictly increasing in \(a\),}
\qquad
\frac{F_{B|A}(b|(z^*)^{-1}(b))}
     {f_{B|A}(b|(z^*)^{-1}(b))}
\quad \text{is strictly increasing in \(b\).}
\end{equation}
The proof has two steps. First, I show that \(z^*\) is piecewise linear. Then I
rule out kinks.

\paragraph{Step 1: \(z^*\) is piecewise linear.}
Suppose towards a contradiction that \(z^*\) is not affine on some concave \(C^2\) piece. Then there is a
closed interval $[\alpha,\beta]\subset(\uul a^*,\ool a^*)$ on which $z^{*\prime\prime}(a)<0$ for all $a\in[\alpha,\beta].$

By Lemma \ref{lemma:optimplement}, the welfare-maximizing \(A\)-indirect utility
for a fixed boundary satisfies \(U_A'(a)=k\) on any concave interval of the
boundary, for some \(k>0\). Consider local perturbations of \(z^*\) on
\([\alpha,\beta]\) that keep fixed \(z^*(\alpha),z^*(\beta)\), the one-sided
slopes \(z^{*\prime}(\alpha),z^{*\prime}(\beta)\), and the mass below the
boundary on this interval:
\[
        \int_\alpha^\beta \int_0^{z(a)} f(a,b)\,db\,da
        =
        \int_\alpha^\beta \int_0^{z^*(a)} f(a,b)\,db\,da .
\]
For all sufficiently small perturbations of this kind that preserve concavity
and keep \(z'>0\), the same \(U_A\) remains feasible. Indeed, \(U_A'\) is still
constant and \(U_A'/z'\) remains non-decreasing. Since the endpoint slopes are
fixed, these monotonicity constraints also remain valid where the perturbed
interval is pasted back into the unchanged boundary. The fixed-mass condition
preserves the supply of good \(A\) on the interval, and therefore also preserves
the supply of good \(B\) on the same vertical strip. By Lemma \ref{lem:objintbyparts}, such perturbations affect welfare only through
\[
        -k\int_\alpha^\beta F(a,z(a))\,da .
\]
Thus \(z^*|_{[\alpha,\beta]}\) must locally maximize
\[
        -\int_\alpha^\beta F(a,z(a))\,da
\]
among concave paths with the same endpoint values, endpoint slopes, and mass
below the boundary. Equivalently, writing \(y=z'\), \(u=y'\) and $q'(a)=\int_0^{z(a)} f(a,b)\,db,$ the local maximization has fixed values of \(z\) and \(y\) at both endpoints and
fixed \(q(\beta)-q(\alpha)\).

Let \(\xi,\phi,\mu\) denote the costates associated with \(z,y,q\),
respectively. The Hamiltonian is
\[
        H(a,z,y,u,q,\xi,\phi,\mu)
        =
        -F(a,z)
        +
        \mu\int_0^{z} f(a,b)\,db
        +
        \xi y+\phi u .
\]
The costate on \(q\) is constant, so \(\mu\) is constant. The remaining costate
equations are
\[
        \xi'(a)
        =
        \int_0^a f(v,z^*(a))\,dv
        -
        \mu f(a,z^*(a)),
        \qquad
        \phi'(a)=-\xi(a).
\]
Hence
\begin{equation}\label{eq:phi-second-vararg}
        \phi''(a)
        =
        \mu f(a,z^*(a))
        -
        \int_0^a f(v,z^*(a))\,dv
        =
        f(a,z^*(a))
        \Bigl(\mu-\frac{F_{A|B}(a|z^*(a))}{f_{A|B}(a|z^*(a))}\Bigr).
\end{equation}
On \([\alpha,\beta]\), the control \(u^*=z^{*\prime\prime}\) is strictly below
the constraint \(u\leq 0\). Therefore the maximum condition for the free control
gives $ \phi(a)=0$ for all $a\in(\alpha,\beta).$ Thus, \(\phi''(a)=0\) on \((\alpha,\beta)\). By
\eqref{eq:phi-second-vararg},
\[
        \frac{F_{A|B}(a|z^*(a))}{f_{A|B}(a|z^*(a))}=\mu
        \qquad\text{for all }a\in(\alpha,\beta).
\]
This, however, contradicts the strict monotonicity in \eqref{eq:strict-antihazard-along-boundary}. Hence
\(z^*\) cannot be strictly concave on any open interval. Applying a symmetric argument rules out strictly
convex intervals. Therefore
\(z^{*\prime\prime}=0\) wherever the second derivative exists. Since \(z^*\) is
piecewise twice continuously differentiable by Lemma
\ref{lemma:indutilityimplement}, it is piecewise linear.

\paragraph{Step 2: \(z^*\) has no kinks.}
We now know that \(z^*\) is piecewise linear. Pick \(\ool v\) such that $[\uul a^*,\ool v]$ is the largest interval starting at \(\uul a^*\) on which \(z^*\) is either
convex or concave. Suppose first that \(z^*\) is concave on this interval. The
convex case is symmetric after passing to the inverse boundary and using the
\(B\)-side inverse anti-hazard rate.

On this initial concave interval, possible nonlinearity of \(z^*\) can only take
the form of downward jumps in the slope. Consider local variations of
\(z^*\) on \([\uul a^*,\ool v]\) that keep fixed the endpoint values
\(z^*(\uul a^*)\) and \(z^*(\ool v)\), keep fixed the terminal slope
\(z^{*\prime}_{-}(\ool v)\), and preserve the mass below the boundary:
\[
        \int_{\uul a^*}^{\ool v}
        \left(\int_0^{z(a)} f(a,b)\,db\right) da
        =
        \int_{\uul a^*}^{\ool v}
        \left(\int_0^{z^*(a)} f(a,b)\,db\right) da .
\]
The initial slope is free. The admissible paths are concave, so away from kinks
we can write
\[
        z'(a)=y(a), \qquad y'(a)=u(a), \qquad u(a)\leq 0,
\]
and at each kink \(a_i\), $ y_+(a_i)-y_-(a_i)=v_i\leq 0,$ while \(z\) and the accumulated-mass state do not jump.

The restriction of \(z^*\) to \([\uul a^*,\ool v]\) must solve this local
problem. Indeed, any admissible improvement can be pasted into the original
boundary. The endpoint conditions preserve continuity of the boundary and the
terminal slope, and the mass condition preserves the supply of good \(A\) on the
interval, hence also the supply of good \(B\) on the same vertical strip. Since
Lemma \ref{lemma:optimplement} gives \(U_A'=k\) on a concave interval, for some
constant \(k>0\), the same pasting argument implies that any improvement in
\[
        -\int_{\uul a^*}^{\ool v} F(a,z(a))\,da
\]
would strictly improve total welfare, contradicting the fact that \(z^*\) maximizes it.

We now show that such a local optimum cannot have a downward jump in
slope. Let \((\xi,\phi,\mu)\) be the costates associated with \(z,y\), and the
accumulated-mass state. Away from jump
points, the Hamiltonian is
\begin{equation}\label{eq:21}
        \mathcal{H}
        =
        -F(a,z^{*}(a))
        +
        \mu\int_0^{z^{*}(a)} f(a,b)\,db
        +
        \xi(a)y(a)
        +
        \phi(a)u(a).
\end{equation}
The costate on the accumulated-mass state is constant, so write it as \(\mu\).
The remaining costate equations are
\begin{equation}\label{eq:22}
        \xi'(a)
        =
        \int_0^a f(v,z^{*}(a))\,dv
        -
        \mu f(a,z^{*}(a)),
\end{equation}
and
\begin{equation}\label{eq:23}
        \phi'(a)=-\xi(a).
\end{equation}
Since the initial value of \(y\) is free, the transversality condition gives
\begin{equation}\label{eq:24}
        \phi(\uul a^*)=0;
\end{equation}
see \citet[p.~234]{neustadtoptimization}.

Because \(u\leq 0\) and the Hamiltonian is linear in \(u\), the maximum
condition implies
\begin{equation}\label{eq:phinonnegative}
        \phi(a)\geq 0
\end{equation}
away from jump points. Moreover, by the maximum principle with jumps
\cite[Theorem~7, pp.~196--197]{seierstad1986optimal}, the adjoints are
continuous across each jump in the present problem. Hence \(\xi\) and \(\phi\) are continuous across
every jump. Since \(\phi'=-\xi\) on both sides of each jump, it follows that
\begin{equation}\label{eq:phiprimecontinuous}
        \lim_{a\uparrow a_i}\phi'(a)
        =
        \lim_{a\downarrow a_i}\phi'(a).
\end{equation}
At any genuine downward jump, the jump size satisfies \(v_i<0\). Since \(v_i\)
enters the jump map only through \(y_+=y_-+v_i\), the first-order condition for
the jump size gives
\begin{equation}\label{eq:jump_phi_zero}
        \phi(a_i)=0.
\end{equation}
Combining \eqref{eq:phinonnegative} and \eqref{eq:jump_phi_zero}, every genuine
interior jump \(a_i\in(\uul a^*,\ool v)\) must be a local minimum of \(\phi\).
Together with \eqref{eq:phiprimecontinuous}, this implies
\begin{equation}\label{eq:jump_derivative_zero}
        \phi'(a_i)=0.
\end{equation}
I now show that no such point can exist.

Away from jump points, differentiating \eqref{eq:23} and using \eqref{eq:22}
gives
\begin{align}
        \phi''(a)
        &=
        -\xi'(a)  
        =
        f(a,z^{*}(a))
        \Big(
        \mu
        -
        \frac{F_{A|B}(a|z^*(a))}{f_{A|B}(a|z^*(a))}
        \Big).
        \label{eq:phisecond2}
\end{align}
Since \(f>0\), \eqref{eq:strict-antihazard-along-boundary} implies that \(\phi''\) changes
sign at most once. It is either always weakly positive, always weakly negative,
or positive up to some cutoff and negative after it.

Because \(\phi(\uul a^*)=0\) and \(\phi\geq 0\), we have $\phi'_+(\uul a^*)\geq 0.$ If \(\phi''\geq 0\) throughout the interval, then \(\phi'\) is non-decreasing
between jumps and continuous across jumps. Hence \(\phi'\) cannot vanish at an
interior local minimum of \(\phi\), so no interior jump can occur.

If \(\phi''\leq 0\) throughout the interval, then \(\phi'\) is non-increasing.
The case \(\phi'_+(\uul a^*)=0\) is impossible, because then \(\phi'\) would be
negative immediately to the right of \(\uul a^*\), forcing \(\phi<0\) nearby.
Thus \(\phi'_+(\uul a^*)>0\). Since \(\phi'\) is non-increasing, it can cross
zero at most once. Such a crossing is a local maximum of \(\phi\), not a local
minimum. Hence no interior jump can occur.

Finally, suppose there is a cutoff \(\tilde a\) such that \(\phi''>0\) on
\((\uul a^*,\tilde a)\) and \(\phi''<0\) on \((\tilde a,\ool v)\). Then
\(\phi'\) is increasing up to \(\tilde a\), so, since
\(\phi'_+(\uul a^*)\geq 0\), we have \(\phi'(a)>0\) on
\((\uul a^*,\tilde a]\). Hence no jump can occur weakly before \(\tilde a\).
For \(a>\tilde a\), \(\phi'\) is strictly decreasing and can cross zero at most
once. At such a crossing, \(\phi\) is strictly positive, because \(\phi\) has
already been increasing on \((\uul a^*,\tilde a]\). After the crossing,
\(\phi'<0\). Thus there is again no interior point at which both
\(\phi(a)=0\) and \(\phi'(a)=0\).

We have thus shown that \(z^*\) cannot have a downward jump in slope on an
initial concave interval. To complete the argument, suppose that \(z^*\) has
at least one kink, and let \(a_1\) denote its first kink. If the slope jumps
downward at \(a_1\), then \(z^*\) is concave on an interval beginning at
\(\uul a^*\) and ending just to the right of \(a_1\), contradicting the
preceding argument. If the slope jumps upward at \(a_1\), then the inverse
boundary has a downward jump at the corresponding point, and the symmetric
argument using the \(B\)-side inverse anti-hazard rate yields the same
contradiction. Hence \(z^*\) has no first kink. Since \(z^*\) is piecewise
linear, it must therefore consist of a single linear piece.

\subsection{Proof of Lemma \ref{lem:slope1}}

By Lemma \ref{lemma:vararg}, the optimal boundary \(z^*\) is linear. Let
\(\sigma>0\) denote its slope. We now show \(\sigma=1\). For each \(s\) in a neighborhood of \(\sigma\), let \(z_s\) denote the linear
boundary with slope \(s\) that allocates the same masses of goods \(A\) and
\(B\) as \(z^*\). Write \(\hat z_s\) for its extended boundary, and write
\(\uul a_s,\uul b_s\) for its participation cutoffs. Since \(f\) is
strictly positive and continuous, these endpoints are locally uniquely pinned
down by the two mass constraints and vary continuously with \(s\).

We first record a monotonicity fact. If \(s_1>s_2\), then the two extended
boundaries \(\hat z_{s_1}\) and \(\hat z_{s_2}\) cross exactly once. Let
\((a^*,b^*)\) be their crossing point. The flatter boundary \(\hat z_{s_2}\) is
above \(\hat z_{s_1}\) to the left of \(a^*\), and below it to the right of
\(a^*\). In particular, \(\uul a_s\) is increasing in \(s\). Define
\[
        I_A(s):=\int_0^1 F(a,\hat z_s(a))\,da .
\]
I claim that \(I_A(s)\) is strictly increasing in \(s\). To see this, fix
\(s_1>s_2\), and define
\[
\mathcal D^- :=
\{(a,b): a<a^*,\ \hat z_{s_1}(a)<b<\hat z_{s_2}(a)\},
\quad \text{and} \quad
\mathcal D^+ :=
\{(a,b): a>a^*,\ \hat z_{s_2}(a)<b<\hat z_{s_1}(a)\}.
\]
Then
\begin{align*}
I_A(s_1)-I_A(s_2)
&=
\int_0^1
\int_{\hat z_{s_2}(a)}^{\hat z_{s_1}(a)}
        \int_0^a f(v,b)\,dv\,db\,da                                                   \\
&=
\int_{\mathcal D^+}
        \frac{F_{A|B}(a|b)}{f_{A|B}(a|b)}
        f(a,b)\,da\,db
-
\int_{\mathcal D^-}
        \frac{F_{A|B}(a|b)}{f_{A|B}(a|b)}
        f(a,b)\,da\,db .
\end{align*}
By the anti-hazard-rate condition in Theorem \ref{th:1}, we have:
\[
        \frac{F_{A|B}(a|b)}{f_{A|B}(a|b)}>\frac{F_{A|B}(a^*|b^*)}{f_{A|B}(a^*|b^*)} \quad \text{on } \mathcal D^+,
        \qquad
        \frac{F_{A|B}(a|b)}{f_{A|B}(a|b)}<\frac{F_{A|B}(a^*|b^*)}{f_{A|B}(a^*|b^*)} \quad \text{on } \mathcal D^-.
\]
Moreover, because \(z_{s_1}\) and \(z_{s_2}\) allocate the same mass of good
\(A\),
\[
        \int_{\mathcal D^+} f(a,b)\,da\,db
        =
        \int_{\mathcal D^-} f(a,b)\,da\,db,
\]
so $ I_A(s_1)-I_A(s_2)>0.$

Now suppose, toward a contradiction, that \(\sigma>1\). For any linear boundary
with slope \(s>1\), Lemma \ref{lemma:optimplement} gives \(m\equiv 1\) and
\[
        k=\frac{1}{\max\{1,1/s\}}=1.
\]
Thus the optimal \(A\)-indirect utility satisfies
\[
        U_{A,s}'(a)=0 \quad \text{for } a<\uul a_s,
        \qquad
        U_{A,s}'(a)=1 \quad \text{for } a>\uul a_s .
\]
Using Lemma \ref{lem:objintbyparts}, welfare under \(z_s\) is therefore
\[
\begin{aligned}
W[z_s]
&=
\int_0^1 \bigl(1-F(a,\hat z_s(a))\bigr)U_{A,s}'(a)\,da  \\
&=
\int_{\uul a_s}^1 \bigl(1-F(a,\hat z_s(a))\bigr)\,da    \\
&=
1-\uul a_s-\int_0^1 F(a,\hat z_s(a))\,da                 =
1-\uul a_s-I_A(s),
\end{aligned}
\]
where the third equality uses \(F(a,0)=0\). Since \(\uul a_s\) is weakly increasing in \(s\) and \(I_A(s)\) is strictly increasing in \(s\), welfare is strictly decreasing in \(s\) on the region
\(s>1\). Hence, if \(\sigma>1\), lowering the slope slightly
while preserving the masses of both goods would strictly increase welfare. This
contradicts the optimality of \(z^*\).

It remains to rule out \(\sigma<1\). Let $\hat w_s(b):=\hat z_s^{-1}(b)$ denote the generalized inverse boundary. When \(s<1\), the inverse boundary has
slope \(1/s>1\). Applying the preceding argument with the roles of \(A\) and
\(B\) reversed, the expression
\[
        I_B(s):=\int_0^1 F(\hat w_s(b),b)\,db
\]
is strictly decreasing in \(s\) on the region \(s<1\), and \(\uul b_s\) is also
strictly decreasing in \(s\). Moreover, applying Lemma
\ref{lemma:optimplement} to the inverse boundary gives
\[
        U_{B,s}'(b)=0 \quad \text{for } b<\uul b_s,
        \qquad
        U_{B,s}'(b)=1 \quad \text{for } b>\uul b_s .
\]
Using the symmetric version of Lemma \ref{lem:objintbyparts}, welfare can be
written as
\[
W[z_s]
=
\int_0^1 \bigl(1-F(\hat w_s(b),b)\bigr)U_{B,s}'(b)\,db =
1-\uul b_s-I_B(s).
\]
Since both \(\uul b_s\) and \(I_B(s)\) are strictly decreasing in \(s\), welfare
is strictly increasing in \(s\) on the region \(s<1\). Hence, if
\(\sigma<1\), raising the slope slightly while preserving the masses of both
goods would strictly increase welfare. This again contradicts the optimality of
\(z^*\).

\section{Verifying examples}

\subsection{Verifying Example \ref{ex:1}}
Consider mechanisms which do not use damages. By an argument analogous to that for Lemma \ref{lemma:bestundamaged}, it suffices to consider the market-clearing toll mechanism. It takes the form:
\[
        y(a,b)=B,\ c(a,b)=1/2
        \quad \text{when } b-a>1/2, \quad \text{ and }
\quad
        y(a,b)=A,\ c(a,b)=0
        \quad \text{when } b-a<1/2.
\]
Thus, the total welfare from this mechanism is
\[
\int_{\{b-a>1/2\}} \left(b-\tfrac12\right) f(a,b)\,da\,db
+
\int_{\{b-a<1/2\}} a\, f(a,b)\,da\,db
=
\frac{14\varepsilon^2-9\varepsilon+23}{42}
-
\frac{28\varepsilon^2-46\varepsilon+25}{252(1-\varepsilon)}.
\]
This converges to $\frac{113}{252}$ as \(\varepsilon\to0^+\).

Now consider an alternative mechanism which offers two options: good \(A\) with
\(x_A=1\) and good \(B\) with \(x_B^\varepsilon<1\), with no tolls for either.
For small enough \(\varepsilon\), there exists \(x_B^\varepsilon\) for which both
supply constraints hold with equality. We can verify that
\[
x_B^\varepsilon
=
\frac{7}{16}
-
\frac{287}{1024}\,\varepsilon
+
O(\varepsilon^2),
\]
and so \(x_B^\varepsilon\to 7/16\) as \(\varepsilon\to0^+\). Calculation confirms
that the limit welfare from this mechanism is
\[
\lim_{\varepsilon\to 0^+}
\left[
\int_{\{x_B^\varepsilon b>a\}} x_B^\varepsilon b\, f(a,b)\,da\,db
+
\int_{\{x_B^\varepsilon b<a\}} a\, f(a,b)\,da\,db
\right]
=
\frac{17497}{36288}.
\]
$\frac{17497}{36288}
        >
        \frac{113}{252},$ so this mechanism dominates the market-clearing toll mechanism for \(\varepsilon>0\)
sufficiently small.

\subsection{Verifying Example \ref{ex:affiliated-profitable-damages}}

First consider \(\lambda\leq 0\). The conditional inverse anti-hazard rate of
\(A\) given \(B=b\) is
\[
        \frac{F_{A\mid B}(a\mid b)}
        {f_{A\mid B}(a\mid b)}
        =
        \frac{\int_0^a e^{\lambda tb}\,dt}
        {e^{\lambda ab}} .
\]
For \(\lambda b\neq 0\), this equals
\[
        \frac{1-e^{-\lambda ab}}{\lambda b},
\]
with the continuous extension equal to \(a\) when \(\lambda b=0\). When
\(\lambda=0\), the conditional inverse anti-hazard rate therefore equals \(a\),
so the condition in Theorem \ref{th:1} holds.

Now suppose \(\lambda<0\). We have
\[
        \frac{\partial}{\partial a}
        \left[
        \frac{1-e^{-\lambda ab}}{\lambda b}
        \right]
        =
        e^{-\lambda ab}
        >
        0
\]
and
\[
        \frac{\partial}{\partial b}
        \left[
        \frac{1-e^{-\lambda ab}}{\lambda b}
        \right]
        =
        \frac{(1+\lambda ab)e^{-\lambda ab}-1}{\lambda b^2}.
\]
Note that
\[
        (1+x)e^{-x}-1\leq 0
\]
for all \(x\), with equality only at \(x=0\). Therefore
\[
        \frac{\partial}{\partial b}
        \left[
        \frac{1-e^{-\lambda ab}}{\lambda b}
        \right]
        \geq 0
        \qquad\text{whenever}\qquad
        \lambda<0.
\]
Applying the same argument to the other conditional inverse anti-hazard rate
shows that the condition in Theorem \ref{th:1} holds for every
\(\lambda<0\). Together with the \(\lambda=0\) case above, this implies that
the market-clearing toll mechanism is optimal, and hence that the optimal
mechanism does not use damages, for every \(\lambda\leq0\).

Now consider large positive \(\lambda\). Note Corollary \ref{cor:damagesnormalized} applies and
\[
        \frac{F_{A\mid B}(b-c_B^*\mid b)}
        {f_{A\mid B}(b-c_B^*\mid b)}
        =
        \frac{\int_0^{b-c_B^*} e^{\lambda tb}\,dt}
        {e^{\lambda (b-c_B^*)b}}
        =
        \frac{1-e^{-\lambda b(b-c_B^*)}}{\lambda b}.
\]
Fix any \(\tau>0\), and set $\tilde b_\lambda
        :=
        1-\frac{\tau}{\lambda}.$ For all sufficiently large \(\lambda\), we have
\(\tilde b_\lambda\in(c_B^*,1)\). Let \(B_\lambda\) denote a random variable
distributed according to the boundary measure, with density proportional to
\(f(b-c_B^*,b)\) on \([c_B^*,1]\). Define $Y_\lambda:=\lambda(1-B_\lambda).$ Its density is obtained from the change of variables
\(b=1-y/\lambda\), so for \(y\in[0,\lambda(1-c_B^*)]\), the density of
\(Y_\lambda\) is proportional to
\[
        \exp\left\{
        \lambda
        \left(1-c_B^*-\frac{y}{\lambda}\right)
        \left(1-\frac{y}{\lambda}\right)
        \right\}
        =
        \exp\left\{
        \lambda(1-c_B^*)-(2-c_B^*)y+\frac{y^2}{\lambda}
        \right\}.
\]
After normalizing, the constant term \(\lambda(1-c_B^*)\) drops out. Hence
\(Y_\lambda\) converges in distribution to an exponential random variable
\(Y\) with rate \(2-c_B^*\). Then
\[
        (\tilde b_\lambda-B_\lambda)_+
        =
        \frac{(Y_\lambda-\tau)_+}{\lambda}.
\]
Also,
\[
        \frac{1-e^{-\lambda B_\lambda(B_\lambda-c_B^*)}}
        {\lambda B_\lambda}
        =
        \frac{1}{\lambda(1-Y_\lambda/\lambda)}
        +
        o(\lambda^{-2})
        =
        \frac{1}{\lambda}
        +
        \frac{Y_\lambda}{\lambda^2}
        +
        o(\lambda^{-2}).
\]
Therefore
\[
\begin{aligned}
        \lambda^3
        \operatorname{Cov}
        \left(
        \frac{F_{A\mid B}(b-c_B^*\mid b)}
        {f_{A\mid B}(b-c_B^*\mid b)},
        (\tilde b_\lambda-b)_+
        \ \middle|\ a=b-c_B^*
        \right)
        &\to
        \operatorname{Cov}\left(Y,(Y-\tau)_+\right).
\end{aligned}
\]
Since \(Y\sim\operatorname{Exp}(2-c_B^*)\),
\[
        \operatorname{Cov}\left(Y,(Y-\tau)_+\right)
        =
        \mathbb E\left[Y(Y-\tau)_+\right]
        -
        \mathbb E[Y]\mathbb E[(Y-\tau)_+]  
        =
        e^{-(2-c_B^*)\tau}
        \left(
        \frac{\tau}{2-c_B^*}
        +
        \frac{1}{(2-c_B^*)^2}
        \right)
        >
        0.
\]
Hence, for all sufficiently large \(\lambda\),
\[
        \operatorname{Cov}
        \left(
        \frac{F_{A\mid B}(b-c_B^*\mid b)}
        {f_{A\mid B}(b-c_B^*\mid b)},
        (\tilde b_\lambda-b)_+
        \ \middle|\ a=b-c_B^*
        \right)
        >
        0.
\]

\begin{thebibliography}{28}
\newcommand{\enquote}[1]{``#1''}
\expandafter\ifx\csname natexlab\endcsname\relax\def\natexlab#1{#1}\fi


\bibitem[\protect\citeauthoryear{Akbarpour, Dworczak, and Kominers}{Akbarpour et~al.}{2024}]{akbarpour2024redistributive}
\textsc{Akbarpour, M., P.~Dworczak, and S.~D. Kominers} (2024): \enquote{Redistributive allocation mechanisms,} \emph{Journal of Political Economy}, 132, 1831--1875.


\bibitem[\protect\citeauthoryear{Akbarpour, Dworczak, and Yang}{Akbarpour et~al.}{2023}]{akbarpour2023comparison}
\textsc{Akbarpour, M., P.~Dworczak, and F.~Yang} (2023): \enquote{Comparison of Screening Devices,} in \emph{Proceedings of the 24th ACM Conference on Economics and Computation}, 60--60.

\bibitem[\protect\citeauthoryear{Alatas, Purnamasari, Wai-Poi, Banerjee, Olken, and Hanna}{Alatas et~al.}{2016}]{alatas2016self}
\textsc{Alatas, V., R.~Purnamasari, M.~Wai-Poi, A.~Banerjee, B.~A. Olken, and R.~Hanna} (2016): \enquote{Self-targeting: Evidence from a field experiment in Indonesia,} \emph{Journal of Political Economy}, 124, 371--427.

\bibitem[\protect\citeauthoryear{Arnosti and Shi}{Arnosti and Shi}{2020}]{arnosti2020design}
\textsc{Arnosti, N. and P.~Shi} (2020): \enquote{Design of lotteries and wait-lists for affordable housing allocation,} \emph{Management Science}, 66, 2291--2307.

\bibitem[\protect\citeauthoryear{Barzel}{Barzel}{1974}]{barzel1974theory}
\textsc{Barzel, Y.} (1974): \enquote{A theory of rationing by waiting,} \emph{The Journal of Law and Economics}, 17, 73--95.

\bibitem[\protect\citeauthoryear{Besley and Coate}{Besley and Coate}{1991}]{Besley}
\textsc{Besley, T. and S.~Coate} (1991): \enquote{Public Provision of Private Goods and the Redistribution of Income,} \emph{The American Economic Review}, 81, 979--984.

\bibitem[\protect\citeauthoryear{Besley and Coate}{Besley and Coate}{1992}]{besley1992workfare}
---\hspace{-.1pt}---\hspace{-.1pt}--- (1992): \enquote{Workfare versus welfare: Incentive arguments for work requirements in poverty-alleviation programs,} \emph{The American Economic Review}, 82, 249--261.

\bibitem[\protect\citeauthoryear{Bloch and Cantala}{Bloch and Cantala}{2017}]{blochCantala}
\textsc{Bloch, F. and D.~Cantala} (2017): \enquote{Dynamic Assignment of Objects to Queuing Agents,} \emph{American Economic Journal: Microeconomics}, 9, 88--122.


\bibitem[\protect\citeauthoryear{Brot-Goldberg, Burn, Layton, and Vabson}{Brot-Goldberg et~al.}{2023}]{brot2023rationing}
\textsc{Brot-Goldberg, Z.~C., S.~Burn, T.~Layton, and B.~Vabson} (2023): \enquote{Rationing medicine through bureaucracy: authorization restrictions in Medicare,} Tech. rep., National Bureau of Economic Research.


\bibitem[\protect\citeauthoryear{Bulow and Klemperer}{Bulow and Klemperer}{2012}]{bulow2012regulated}
\textsc{Bulow, J. and P.~Klemperer} (2012): \enquote{Regulated prices, rent seeking, and consumer surplus,} \emph{Journal of Political Economy}, 120, 160--186.

\bibitem[\protect\citeauthoryear{City of Vancouver}{City of Vancouver}{2016}]{vancouver2016aho}
\textsc{City of Vancouver} (2016): \enquote{Affordable Home Ownership Pilot Program,} Policy Report to the Standing Committee on City Finance and Services, Vancouver City Council.


\bibitem[\protect\citeauthoryear{Condorelli}{Condorelli}{2012}]{CONDORELLI2012613}
\textsc{Condorelli, D.} (2012): \enquote{What money can't buy: Efficient mechanism design with costly signals,} \emph{Games and Economic Behavior}, 75, 613--624.

\bibitem[\protect\citeauthoryear{Currie and Gahvari}{Currie and Gahvari}{2008}]{currieGahvari}
\textsc{Currie, J. and F. Gahvari} (2008): \enquote{Transfers in Cash and In-Kind: Theory Meets the Data,} \emph{Journal of Economic Literature}, 46, 333--383.

\bibitem[\protect\citeauthoryear{Daskalakis, Deckelbaum, and Tzamos}{Daskalakis, Deckelbaum, and Tzamos}{2017}]{daskalakisDeckelbaumTzamos}
\textsc{Daskalakis, C., A. Deckelbaum, and C. Tzamos} (2017): \enquote{Strong duality for a multiple-good monopolist,} \emph{Econometrica}, 85, 735--767.

\bibitem[\protect\citeauthoryear{Deneckere and McAfee}{Deneckere and McAfee}{1996}]{deneckere1996damaged}
\textsc{Deneckere, R.~J. and P.~R. McAfee} (1996): \enquote{Damaged goods,} \emph{Journal of Economics \& Management Strategy}, 5, 149--174.

\bibitem[\protect\citeauthoryear{Deshpande and Li}{Deshpande and Li}{2019}]{deshpande2019screened}
\textsc{Deshpande, M. and Y.~Li} (2019): \enquote{Who is screened out? Application costs and the targeting of disability programs,} \emph{American Economic Journal: Economic Policy}, 11, 213--248.

\bibitem[\protect\citeauthoryear{Dupas, Hoffmann, Kremer, and Zwane}{Dupas et~al.}{2016}]{dupas2016targeting}
\textsc{Dupas, P., V.~Hoffmann, M.~Kremer, and A.~P. Zwane} (2016): \enquote{Targeting health subsidies through a nonprice mechanism: A randomized controlled trial in Kenya,} \emph{Science}, 353, 889--895.

\bibitem[\protect\citeauthoryear{Dworczak}{Dworczak}{2026}]{dworczak2026allocate}
\textsc{Dworczak, P.} (2026): \enquote{How to Allocate Money?} \emph{American Economic Journal: Microeconomics}, 18, 312--339.

\bibitem[\protect\citeauthoryear{Dworczak, Kominers, and Akbarpour}{Dworczak et~al.}{2021}]{DKA}
\textsc{Dworczak, P., S.~D. Kominers, and M.~Akbarpour} (2021): \enquote{Redistribution Through Markets,} \emph{Econometrica}, 89, 1665--1698.


\bibitem[\protect\citeauthoryear{Feinberg and Piunovskiy}{Feinberg and Piunovskiy}{2006}]{feinberg2006dvoretzky}
\textsc{Feinberg, E.~A. and A.~B. Piunovskiy} (2006): \enquote{On the Dvoretzky--Wald--Wolfowitz theorem on nonrandomized statistical decisions,} \emph{Theory of Probability \& Its Applications}, 50, 463--466.


\bibitem[\protect\citeauthoryear{Finkelstein and Notowidigdo}{Finkelstein and Notowidigdo}{2019}]{finkelstein2019take}
\textsc{Finkelstein, A. and M.~J. Notowidigdo} (2019): \enquote{Take-up and targeting: Experimental evidence from SNAP,} \emph{The Quarterly Journal of Economics}, 134, 1505--1556.


\bibitem[\protect\citeauthoryear{Hartline and Roughgarden}{Hartline and Roughgarden}{2008}]{hartline2008optimal}
\textsc{Hartline, J.~D. and T.~Roughgarden} (2008): \enquote{Optimal mechanism design and money burning,} in \emph{Proceedings of the fortieth annual ACM symposium on Theory of computing}, 75--84.


\bibitem[\protect\citeauthoryear{Kleven and Kopczuk}{Kleven and Kopczuk}{2011}]{kleven2011transfer}
\textsc{Kleven, H.~J. and W.~Kopczuk} (2011): \enquote{Transfer program complexity and the take-up of social benefits,} \emph{American Economic Journal: Economic Policy}, 3, 54--90.

\bibitem[\protect\citeauthoryear{Leshno}{Leshno}{2022}]{leshno2022dynamic}
\textsc{Leshno, J.~D.} (2022): \enquote{Dynamic matching in overloaded waiting lists,} \emph{American Economic Review}, 112, 3876--3910.

\bibitem[\protect\citeauthoryear{Manelli and Vincent}{Manelli and Vincent}{2006}]{manelliVincent}
\textsc{Manelli, A.~M. and D.~R. Vincent} (2006): \enquote{Bundling as an optimal selling mechanism for a multiple-good monopolist,} \emph{Journal of Economic Theory}, 127, 1--35.


\bibitem[\protect\citeauthoryear{Manning, Newhouse, Duan, Keeler, and Leibowitz}{Manning et~al.}{1987}]{manning1987health}
\textsc{Manning, W.~G., J.~P. Newhouse, N.~Duan, E.~B. Keeler, and A.~Leibowitz} (1987): \enquote{Health insurance and the demand for medical care: evidence from a randomized experiment,} \emph{The American Economic Review}, 251--277.


\bibitem[\protect\citeauthoryear{Milgrom and Weber}{Milgrom and Weber}{1982}]{milgromweber}
\textsc{Milgrom, P.~R. and R.~J. Weber} (1982): \enquote{A Theory of Auctions and Competitive Bidding,} \emph{Econometrica}, 50, 1089--1122.


\bibitem[\protect\citeauthoryear{Neustadt}{Neustadt}{1976}]{neustadtoptimization}
\textsc{Neustadt, L.~W.} (1976): \emph{Optimization: A Theory of Necessary Conditions}, Princeton University Press.


\bibitem[\protect\citeauthoryear{Noda and Okada}{Noda and Okada}{2024}]{noda2024no}
\textsc{Noda, S. and G.~Okada} (2024): \enquote{No Screening is More Efficient with Multiple Objects,} \emph{arXiv preprint arXiv:2408.10077}.


\bibitem[\protect\citeauthoryear{Nichols and Zeckhauser}{Nichols and Zeckhauser}{1982}]{nichols}
\textsc{Nichols, A.~L. and R.~J. Zeckhauser} (1982): \enquote{Targeting Transfers through Restrictions on Recipients,} \emph{The American Economic Review}, 72, 372--377.

\bibitem[\protect\citeauthoryear{Nichols, Smolensky, and Tideman}{Nichols et~al.}{1971}]{nichols1971discrimination}
\textsc{Nichols, D., E.~Smolensky, and T.~N. Tideman} (1971): \enquote{Discrimination by waiting time in merit goods,} \emph{The American Economic Review}, 61, 312--323.

\bibitem[\protect\citeauthoryear{Rochet and Choné}{Rochet and Choné}{1998}]{rochetChone}
\textsc{Rochet, J.-C. and P. Choné} (1998): \enquote{Ironing, Sweeping, and Multidimensional Screening,} \emph{Econometrica}, 66, 783--826.


\bibitem[\protect\citeauthoryear{Seierstad and Sydsaeter}{Seierstad and Sydsaeter}{1986}]{seierstad1986optimal}
\textsc{Seierstad, A. and K.~Sydsaeter} (1986): \emph{Optimal control theory with economic applications}, Elsevier North-Holland, Inc.


\bibitem[\protect\citeauthoryear{Tokarski}{Tokarski}{2026}]{tokarski2026targeting}
\textsc{Tokarski, F.} (2026): \enquote{Targeting Without Transfers,} \emph{arXiv preprint arXiv:2602.00487}.

\bibitem[\protect\citeauthoryear{Van Den~Berg and Verhoef}{Van Den~Berg and Verhoef}{2011}]{van2011winning}
\textsc{Van Den~Berg, V. and E.~T. Verhoef} (2011): \enquote{Winning or losing from dynamic bottleneck congestion pricing?: The distributional effects of road pricing with heterogeneity in values of time and schedule delay,} \emph{Journal of Public Economics}, 95, 983--992.

\bibitem[\protect\citeauthoryear{Van~Ommeren and Van~der Vlist}{Van~Ommeren and Van~der Vlist}{2016}]{van2016households}
\textsc{Van~Ommeren, J.~N. and A.~J. Van~der Vlist} (2016): \enquote{Households' willingness to pay for public housing,} \emph{Journal of Urban Economics}, 92, 91--105.

\bibitem[\protect\citeauthoryear{Vickrey}{Vickrey}{1973}]{vickrey1973pricing}
\textsc{Vickrey, W.} (1973): \emph{Pricing, metering, and efficiently using urban transportation facilities}, 476.

\bibitem[\protect\citeauthoryear{Waldinger}{Waldinger}{2021}]{waldinger2021targeting}
\textsc{Waldinger, D.} (2021): \enquote{Targeting in-kind transfers through market design: A revealed preference analysis of public housing allocation,} \emph{American Economic Review}, 111, 2660--2696.

\bibitem[\protect\citeauthoryear{Wang, Xu, and Qin}{Wang et~al.}{2014}]{wang2014will}
\textsc{Wang, L., J.~Xu, and P.~Qin} (2014): \enquote{Will a driving restriction policy reduce car trips?—The case study of Beijing, China,} \emph{Transportation Research Part A: Policy and Practice}, 67, 279--290.

\bibitem[\protect\citeauthoryear{Whitehead and Scanlon}{Whitehead and Scanlon}{2007}]{whitehead2007social}
\textsc{Whitehead, C. and K.~J. Scanlon} (2007): \emph{Social housing in Europe}, London School of Economics and Political Science.

\bibitem[\protect\citeauthoryear{{{World Health Organization}}}{{{World Health Organization}}}{2023}]{world2023high}
\textsc{{{World Health Organization}}} (2023): \enquote{High-value referrals: learning from challenges and opportunities of the COVID-19 pandemic: concept paper,} \emph{High-value referrals: learning from challenges and opportunities of the COVID-19 pandemic: concept paper}.

\bibitem[\protect\citeauthoryear{Yang}{Yang}{2021}]{yang2021costly}
\textsc{Yang, F.} (2021): \enquote{Costly multidimensional screening,} \emph{arXiv preprint arXiv:2109.00487}.

\end{thebibliography}
\end{document}